\documentclass[11pt]{article}

\usepackage[T1]{fontenc}
\usepackage[utf8]{inputenc}
\usepackage{lmodern}
\usepackage[margin=1in]{geometry}
\usepackage{amsthm,amsmath,amsfonts,amssymb}
\usepackage[authoryear,round]{natbib}
\usepackage{bm}
\usepackage{xcolor}
\usepackage{multirow}
\usepackage{enumitem}
\usepackage{float}
\usepackage{graphicx}
\usepackage{booktabs}
\usepackage{chngcntr}
\usepackage[ruled,linesnumbered,lined]{algorithm2e}
\usepackage[colorlinks=true,citecolor=blue,linkcolor=blue,urlcolor=blue]{hyperref}

\graphicspath{{./art/}}
\DontPrintSemicolon
\SetKwInput{KwIn}{Input}
\SetKwInput{KwOut}{Output}
\SetAlgoCaptionSeparator{:}

\theoremstyle{plain}
\newtheorem{theorem}{Theorem}[section]
\newtheorem{proposition}[theorem]{Proposition}
\newtheorem{lemma}[theorem]{Lemma}
\newtheorem{corollary}[theorem]{Corollary}

\theoremstyle{definition}
\newtheorem{assumption}[theorem]{Assumption}
\newtheorem{definition}[theorem]{Definition}
\newtheorem{example}[theorem]{Example}
\newtheorem{remark}[theorem]{Remark}

\newcommand{\tbl}[1]{\caption{#1}}
\newenvironment{tabnote}
  {\par\smallskip\begin{minipage}{\linewidth}\footnotesize\textit{Note: }\ignorespaces}
  {\end{minipage}\par}
\newcommand{\suppref}[1]{\ref{#1}}

\title{Design-based Estimation and Inference on Quantile Exposure Effect under General Interference}
\makeatletter
\renewcommand{\@fnsymbol}[1]{%
  \ifcase#1\or\textdagger\or\textdaggerdbl\or\textsection\or\textasteriskcentered\else\@ctrerr\fi}
\makeatother
\author{%
  Haoxiang Wang\thanks{School of Mathematical Sciences, Peking University. Email: \texttt{whxwhx@pku.edu.cn}.}
  \and
  Lan Wang\thanks{Department of Management Science, University of Miami. Email: \texttt{lanwang@mbs.miami.edu}.}\textsuperscript{,,}\footnotemark[4]
  \and
  Xiao-Hua Zhou\thanks{Beijing International Center for Mathematical Research and Department of Biostatistics, Peking University. Email: \texttt{azhou@math.pku.edu.cn}.}\textsuperscript{,,}\footnotemark[4]
}
\date{\today}

\begin{document}
\maketitle
\begingroup
\renewcommand{\thefootnote}{\fnsymbol{footnote}}
\footnotetext[4]{Corresponding author.}
\endgroup
\begin{abstract}
Many applications in public health, environmental science, and economics feature spillovers across connected units, violating the Stable Unit Treatment Value Assumption (SUTVA) underlying classical quantile treatment effect  methods. We develop a general framework for defining, estimating, and conducting inference for quantile exposure effects (QEEs) under network interference, encompassing quantile direct and spillover effects as leading cases.
Studying QEEs under interference faces three substantive challenges. First, because a single exposure level represents many neighborhood treatment configurations, the causal estimand must aggregate over these configurations in an interpretable manner. Second, quantile estimands are intrinsically non-smooth, placing them outside much of the existing network theory involving only Lipschitz or differentiable functionals. Third, in design-based, finite-population network settings, conventional smoothness assumptions on outcome densities cannot be imposed directly. Using conditional neighborhood dependence, we establish asymptotic normality under explicit network degree conditions and derive a locally uniform convergence rate for kernel density estimation. We also characterize the variance estimation bias arising from heterogeneous unit-specific score means and construct asymptotically conservative confidence intervals.
Extensive simulations and an application to an educational intervention in school friendship networks demonstrate that the proposed framework reveals heterogeneous exposure effects, including tail-specific impacts, that are missed by analyses based solely on average treatment effects.
\end{abstract}

\section{Introduction}

In many applications, interventions propagate through social, spatial, or economic networks, so a unit's outcome may depend on treatments assigned to connected units. In such settings, average effects can conceal changes in dispersion, downside risk, or the tails of the outcome distribution. The quantile treatment effect (QTE), which compares corresponding quantiles of potential-outcome distributions, is designed to reveal these distributional changes. For example, an intervention may substantially improve outcomes among the most vulnerable units while having little effect near the center of the distribution, or it may reduce extreme outcomes without changing the mean. Such distinctions are important in public health, environmental science, and economics~\citep{sun2021causal, frolich2013unconditional,ullah2023impact}. Classical QTE methods, however, rely on the Stable Unit Treatment Value Assumption (SUTVA) and therefore exclude network spillovers. A general theory of QTEs under interference is consequently needed.

 Developing a valid definition and inferential theory for QTEs under interference in general networks raises three distinct challenges. First, an exposure mapping reduces the combinatorial treatment space by grouping neighborhood assignments into exposure levels, but an exposure level generally still represents many treatment configurations. These configurations can induce different potential outcomes and occur with different randomization probabilities across units. Thus, before taking a quantile, one must construct a well-defined and interpretable exposure-specific distribution that aggregates over both configurations and units. Unlike a mean, the quantile of the resulting mixture cannot be recovered by averaging unit-level or configuration-specific quantiles, so the distribution itself must be the primitive estimand. Second, quantiles are non-smooth functionals as their estimating equations contain indicators and their natural objective is the non-differentiable check loss. Limit theories for network estimators that require the target statistic to be Lipschitz or differentiable therefore do not transfer directly. Third, in a design-based finite-population analysis, the exposure-specific outcome distribution is discrete at every sample size, whereas standard QTE arguments rely on smooth population densities. A rigorous theory must connect this sequence of discrete, randomization-induced distributions to suitable smooth limits while retaining network dependence.

Motivated by the above challenges, this paper develops a general framework to define, estimate and make inference on quantile exposure effects (QEEs) under interference. Our framework begins by defining a finite-population mixture distribution induced by the treatment randomization for each exposure level. Contrasting a given attribute across these exposure-specific distributions yields a general class of distributional causal effects, with quantile exposure effects as the principal case. This construction makes explicit both the population over which each quantile is taken and the weights assigned to treatment configurations corresponding to the same exposure. Building on this estimand, we develop estimation and inference for quantile exposure effects under general, non-clustered networks. The proposed estimator minimizes an exposure-weighted check-loss criterion whose population target is precisely the quantile of the corresponding mixture distribution. Our asymptotic analysis adopts a design-based finite-population perspective~\citep{aronow2017estimating,leung2022causal,gao2025causal, shi2026berry} as at each sample size, the potential outcomes and observed network are treated as fixed, all randomness arises from the treatment assignment, and asymptotics are taken along a sequence of expanding finite populations and networks. 
We establish asymptotic normality for non-smooth quantile estimation under network interference,  through the conditional neighborhood dependence framework of \citet{lee2019stable} and, alternatively, through Janson's dependency-graph
central limit theorem \citep{janson1988normal}.
The two approaches yield complementary sufficient network conditions. The Lee--Song conditions incorporates both maximal and average two-step degrees, allowing greater
maximal connectivity when highly connected units are concentrated in a small part of the network; while Janson's formulation can permit faster maximal-degree growth when
average and maximal degrees are comparable.
Examples involving localized hubs and complete communities illustrate how network heterogeneity affects these conditions. Both approaches retain the same quantile
linearization argument, which connects the discrete finite-population target to a smooth local approximation without imposing a parametric model for network formation.

Valid inference on QEE requires additional control of density and variance estimation beyond the conditions
used to establish asymptotic normality.
We develop feasible inference within the Lee--Song framework, under additional conditions controlling density and variance
estimation. Two challenges arise. First, the influence scores involve densities evaluated at the target quantiles. We establish a locally uniform convergence rate for kernel
density estimation under network interference, explicitly quantifying the effect of network degree and justifying evaluation at estimated quantiles. This result may also
be of independent interest. Second, variance estimation must account for dependence among neighboring scores and heterogeneity in their expectations. A natural plug-in
covariance estimator centers the scores by their sample average, whereas the target variance depends on unobserved
unit-specific score means. 
We characterize the resulting leading centering term and construct a conservative positive-part spectral correction that truncates the negative eigenvalues of the centered two-step dependence matrix.
The corrected estimator is nonnegative and, under the density- and variance-estimation conditions together with an explicit spectral-rate condition, yields asymptotically conservative confidence intervals.

The main contributions of this paper are as follows:
\begin{itemize}[itemsep=-10pt, topsep=5pt, leftmargin=9pt]
    \item 
    We adopt a design-based finite-population perspective and define 
    a general class of distributional causal effects through mixture distributions induced by treatment randomization.  
    By accounting for the multiple treatment configurations represented by a single exposure, we provide a general framework for comparing exposures that differ in either a unit’s own treatment or its neighbors’ treatment. It encompasses quantile direct effects and quantile spillover effects as leading cases of quantile exposure effects, while recovering average direct, spillover, and more general average exposure effects as special cases.\\
    \item 
    We develop estimator and asymptotic theory for quantile exposure effects, addressing both the non-smoothness of quantile estimation and the discreteness of the target distributions. We establish
asymptotic normality under complementary Lee--Song and Janson conditions and characterize how maximal connectivity and degree heterogeneity affect their applicability.
We also establish a locally uniform convergence rate for kernel density estimation that explicitly captures the effect of network degree.\\
    \item We develop feasible variance estimation under network dependence and characterize the leading centering term arising from heterogeneity in unit-specific score means. To address its possibly negative contribution, we propose a positive-part spectral correction and establish asymptotically conservative coverage under an explicit condition on the negative spectral component of the centered two-step dependence matrix.
\end{itemize}

We illustrate our approach using the anti-conflict intervention data of \citet{paluck2016changing}, a network experiment conducted in 56 schools and subsequently analyzed for average effects under interference~\citep{leung2022causal,gao2025causal}. To capture broader changes in student behavior, we construct a social-norm summary score that combines wristband wearing, which rewards observed anti-conflict behavior, with students’ normalized counts of reported anti-norm behaviors. Using the observed friendship network, we estimate overall, direct, and spillover quantile exposure effects. The estimated effects are close to zero in the lower and upper tails but substantially positive around the middle of the outcome distribution. This pattern appears for both direct and spillover effects, although the direct effect remains positive over a wider range of quantiles. These findings refine the positive average effects documented in previous studies: the intervention does not shift the social-norm distribution uniformly, but instead produces improvements concentrated near its middle. The application therefore demonstrates how quantile analysis can uncover distributional heterogeneity and distinguish the effects of students’ own treatment from peer-mediated spillovers that are obscured by average effects.

The remaining of the paper is organized as follows.
Section~\ref{sec:review} reviews the related literature. Section~\ref{sec:def_and_framework} defines the 
quantile exposure effects for networked data, while Section~\ref{sec:estimation} introduces the estimator.
Section~\ref{sec:asp_and_thm} develops the theoretical properties. Sections~\ref{sec:simu} and~\ref{sec:exp} report the simulations and empirical application, respectively. The technical proofs and additional numerical results are provided in the online supplement.

\section{Related Work}\label{sec:review}
Our work lies at the intersection of three strands of literature: quantile treatment effects, exposure-based causal inference under interference, and design-based asymptotic theory for network experiments.
Existing methods do not jointly address quantiles of exposure-specific finite-population distributions, which require accounting for multiple configurations within an exposure, network dependence, and non-smooth estimation. We review these three strands in turn.

\medskip
\noindent\textbf{Quantile treatment effects.} In settings without interference, QTEs compare quantiles of potential-outcome distributions and have been used to study distributional heterogeneity, downside risk, and tail behavior in economics, environmental science, and public health~\citep{firpo2007efficient, frolich2013unconditional,ullah2023impact,sun2021causal}.See also related work on robust procedures~\citep{zhang2012causal,xie2020multiply}, semiparametric propensity-score methods~\citep{zhan2024estimation}, and debiased machine-learning approaches in high-dimensional settings~\citep{kallus2024localized}. These methods generally rely on SUTVA~\citep{imbens2015causal} and do not address the exposure-induced mixtures or network dependence considered here. \citet{cheng2026nonparametric} studies quantile treatment effect under partial interference in clustered networks. However, standard QTE theory is commonly formulated for smooth population distributions, whereas our design-based target is a discrete finite-population distribution that changes with the network experiment. Our contribution is to resolve these definitional, non-smoothness, and finite-population issues jointly under general interference.

\medskip

\textbf{Causal inference under interference.} Since the seminal work of \citet{sobel2006randomized} and \citet{hudgens2008toward}, causal inference under interference has developed rapidly; see \citet{halloran2016dependent,bhadra2025causal} for reviews. Applications include vaccination, social interactions, and policy diffusion in networks~\citep{ogburn2014vaccines,cai2015social,paluck2016changing,viviano2025policy}. Much of this literature adopts a design-based perspective: potential outcomes and the observed network are treated as fixed, and randomness arises from treatment assignment~\citep{aronow2017estimating,leung2022causal,gao2025causal, shi2026berry}. The dominant targets are averages of unit-level direct, indirect, or exposure effects. In contrast, our target is an exposure-specific outcome distribution and its quantiles. This distinction changes both the definition of the estimand and the tools required for inference.

Under interference, a binary assignment vector can generate exponentially many neighborhood treatment configurations. Exposure mappings reduce this complexity by summarizing a unit's own treatment and relevant features of its neighbors' assignments~\citep{baird2018optimal,forastiere2021identification}. Common examples include the proportion of treated neighbors~\citep{li2022random} and an indicator that at least one neighbor is treated~\citep{leung2022causal}. The literature has also studied misspecified exposure mappings~\citep{savje2023causal} and methods for learning mappings from network data~\citep{ma2021causal,huang2023modeling}. Existing exposure-based estimands are predominantly means. Our framework instead uses the randomization distribution to aggregate all assignment configurations associated with an exposure into a finite-population outcome distribution. Applying different attributes to this distribution produces a class of distributional exposure effects. Because quantiles do not commute with aggregation, a quantile exposure effect is not an average of unit-level or configuration-specific quantile effects. Average exposure effects are recovered as a benchmark, but quantile exposure effects require their own definition and asymptotic analysis.

\medskip

\textbf{Network structure and asymptotic frameworks.} Asymptotic analyses under interference differ in how they treat the network. One approach specifies a graph-generating model and averages over network randomness~\citep{hu2022average,li2022random}. A second follows a design-based finite-population perspective: at each sample size, the potential outcomes, covariates, and observed network are fixed, and all randomness arises from treatment assignment. The estimands therefore pertain to the realized population and network, while asymptotics are taken along a sequence of expanding finite-population experiments. This perspective extends classical finite-population theory~\citep{li2017general} to settings with interference~\citep{hudgens2008toward,liu2014large,aronow2017estimating,leung2022causal,gao2025causal} without requiring a model for network formation.

Within the fixed-network tradition, methods differ in how they control dependence. Under partial or clustered interference, spillovers are confined to groups, permitting large-sample inference as the number of independent groups increases~\citep{hudgens2008toward,liu2014large,basse2018analyzing,barkley2020causal,imai2021causal}. Cluster randomization has also been used in graph and spatial experiments where spillovers may cross cluster boundaries~\citep{ugander2013graph,leung2022rate}. For general, non-clustered networks, existing work instead imposes degree or local-dependence restrictions~\citep{aronow2017estimating,ogburn2022causal}. The approximate neighborhood interference framework of \citet{leung2022causal} allows the effects of distant assignments to decay and establishes inference for IPW and Hajek estimators of average exposure effects; \citet{gao2025causal} develops regression-based extensions. Other work studies finite-sample randomization tests~\citep{athey2018exact,basse2019randomization,puelz2022graph}, unknown interference structures~\citep{egami2021spillover,savje2021average}, and misspecified exposure mappings~\citep{savje2023causal}.

 Existing asymptotic arguments primarily concern averages or smooth transformations, whereas quantile estimation involves indicators and a non-differentiable check-loss objective. Moreover, the quantile of an exposure-specific mixture cannot be obtained by averaging configuration-specific quantiles. We use the conditional neighborhood dependence framework of \citet{lee2019stable} to characterize the randomization-induced dependence of the relevant empirical processes and develop the additional local approximation required for quantile estimation. Thus, although CND provides a foundation, the resulting theory is specific to the non-smooth distributional structure of the problem and to the connection between discrete finite-population targets and their smooth limiting approximations.

\section{Quantile Exposure Effect for Networked Data}
\label{sec:def_and_framework}

Section~\ref{notation} introduces basic network notation and exposure mapping. Section~\ref{distributional} presents our framework for distributional effects, including quantile exposure effects. 

\subsection{Network Notation and Exposure Mapping}\label{notation}
We consider a finite-population randomized experiment of $n$ units indexed by $N_n=\{1,\ldots,n\}$, where the only source
of randomness is the treatment assignment. For each unit $i\in N_n$, let $X_i$ denote the vector of associated covariates and $W_i$ the binary treatment assigned based on known assignment mechanism.
Write $\bm{X}=(X_1,\ldots,X_n)\in\bm{\mathcal{X}}_n$, $\bm{W}=(W_1,\ldots,W_n)\in\{0,1\}^n$. Let $Y_i$ be the observed outcome under the realized assignment $\bm{W}$. Throughout, treatment assignments $\{W_i:i\in N_n\}$
are mutually independent
conditional on the network, covariates, and potential-outcome, with possibly unit-specific treatment probabilities.
For any set 
$K\subset N_n$, define $\bm{X}_K=\{X_i\}_{i\in K}$ and $\bm{W}_K=\{W_i\}_{i\in K}$ (and interpret these as
$\emptyset$ when $K=\emptyset$). 

For two positive deterministic sequences $a_n$ and $b_n$, we write $a_n\asymp b_n$ if there exists $0<c<C<\infty$ and $ n_0\in\mathbb N$ such that $c\le \frac{a_n}{b_n}\le C$ for all $n\ge n_0.$ Write $a_n=O(b_n)$ when $|a_n|/b_n$ is bounded, and $a_n=o(b_n)$ when
$a_n/b_n\to0$.  The stochastic analogues $O_p(b_n)$ and $o_p(b_n)$ mean,
respectively, that $a_n/b_n$ is bounded in probability and converges to zero
in probability.

Let $2^{N_n}$ be the collection of all subsets of $N_n$.  The fixed network
has vertex set $N_n$ and a symmetric adjacency matrix
$A_n\in\mathcal{A}_n$ with $A_{ii}=0$.  For $i\ne j$, $A_{ij}=1$ if units
$i$ and $j$ are connected and $A_{ij}=0$ otherwise.  We write $A_n$ as $A$
when $n$ is clear from context, and let $A_K$ be the $|K|\times|K|$
submatrix indexed by $K$.  The corresponding undirected graph is
$\mathcal{G}_n=(N_n,A_n)$. Let $\ell_A(i,j)$ be the path distance between
units $i$ and $j$ in $\mathcal{G}_n$, defined as the length of the shortest
path connecting them.
We set $\ell_A(i,j)=0$ if $i=j$ and $\ell_A(i,j)=\infty$ if $i$ and $j$ are disconnected.

Given $\mathcal{G}_n$, consider a neighborhood system which associate a set of neighbors to each individual through mapping $\nu_n:N_n\rightarrow 2^{N_n}.$ For individual $i$, define its (open) neighborhood as $\nu_n(i)
=
\{j\in N_n\setminus\{i\}:A_{ij}=1\}$ and define the closed neighborhood as $
\bar\nu_n(i)=\nu_n(i)\cup\{i\}$, respectively. Let $\nu_n(i,j)=(\bar\nu_n(i)\cup \bar\nu_n(j))\setminus\{i,j\}$ be the joint neighborhood. For a subset $S\subseteq N_n$, let $|S|$ denote the cardinality of the set, and let $
\bar\nu_n(S)
=
S\cup\bigcup_{i\in S}\nu_n(i),$ $
\nu_n(S)=\bar\nu_n(S)\setminus S.$ Thus, $\nu_n(S)$ is the boundary of $S$, whereas
$\bar\nu_n(S)$ contains both $S$ and its boundary. For fixed positive integer $k$, define the $k$-step neighborhood as \(\nu_n^{(k)}(i)
=
\{j\in N_n\setminus\{i\}:\ell_A(i,j)\le k\}\), and let \(\bar\nu_n^{(k)}(i)=\nu_n^{(k)}(i)\cup\{i\}\) be the closed neighborhood.

 Let $Y_i(\bm{w})\in\mathcal{Y}_n$ denote the potential outcome of unit $i$ under the treatment vector
$\bm{w}=\{w_1,\ldots,w_n\}\in\{0,1\}^n$.
We adopt a
design-based perspective in which \(\{Y_i(\bm w)\}\) and the network
\(\mathcal{G}_n\) are fixed, and all randomness arises from the treatment
assignment mechanism. For notational simplicity, we suppress this conditioning and write the induced
probability measure $\mathrm{pr}(\cdot \mid \mathcal{G}_n,\{Y_i(\bm{w})\}_{i\in N_n,\bm{w}\in\{0,1\}^n})$ simply as $\mathrm{pr}(\cdot)$ and similarly for the expectations.

Under interference, $Y_i(\bm{w})$ may depend on a large portion of the global assignment vector \(\bm w\), which quickly becomes high-dimensional and makes causal contrasts difficult to define and estimate. To obtain
tractable estimands we follow \citet{aronow2017estimating,leung2022causal,gao2025causal} and summarize the relevant features of \(\bm w\) (and the
network) through a low-dimensional {\it exposure mapping}, so that causal
effects can be expressed as contrasts across exposure conditions. Formally, let
\[
\Phi_n: N_n \times \{0,1\}^n \times \mathcal{A} \to \boldsymbol{G}
\]
map \((i,\bm W,A)\) to a discrete exposure space
\(\boldsymbol{G}\subset\mathbb{R}^{d_{\boldsymbol{G}}}\)
, and define the
realized exposure
\begin{equation}\label{Gi}
G_i=\Phi_n(i,\bm W,A).
\end{equation}

Different assignments within a unit’s neighborhood may yield the same exposure but different outcomes. We impose the following locality condition on \(\Phi_n\), which ensures that a
unit’s exposure depends only on the treatment assignments and network structure
in its neighborhood.
\begin{assumption}[Locality of exposure mapping]
\label{asp:exp_map}
For the fixed network $A$, every $i\in N_n$, and all
$\bm w,\bm w'\in\{0,1\}^n$,
\[
\bm w_{\bar\nu_n(i)}=\bm w'_{\bar\nu_n(i)}
\quad\Longrightarrow\quad
\Phi_n(i,\bm w,A)=\Phi_n(i,\bm w',A).
\]
\end{assumption}

Assumption \ref{asp:exp_map} is standard in causal inference for mean treatment effects under interference~\citet{leung2022causal,gao2025causal}. It formalizes the idea that exposure is a low-dimensional, neighborhood-based summary of the global assignment, thereby reducing the effective dimensionality of \(\bm W\) and mitigating sparsity in treatment configurations. Importantly, Assumption \ref{asp:exp_map} restricts only the information used to construct
exposures \(\Phi_n\), while it does not require that potential outcomes depend on neighbors’ treatments \emph{only} through \(G_i\), which is more general than many interference models such as
\citep{hudgens2008toward,aronow2017estimating,forastiere2021identification}.

\begin{assumption}[Locality of potential outcomes]
\label{asp:local-potential-outcomes}
For every unit \(i\in N_n\) and every pair of treatment assignments
\(w,w'\in\{0,1\}^n\),
\[
    w_{\bar{\nu}_n(i)}
    =
    w'_{\bar{\nu}_n(i)}
    \quad\Longrightarrow\quad
    Y_i(w)=Y_i(w').
\]
\end{assumption}

Assumption~\ref{asp:local-potential-outcomes} imposes local
interference that the potential outcome of unit \(i\) may depend
arbitrarily on the entire treatment configuration within its closed
neighborhood, but not on treatments outside that neighborhood. This assumption bridges connection between the evaluation on interference and neighboring effects, and is a common assumption that inherently holds in most existing interference frameworks such as~\citet{leung2022causal,li2022random,gao2025causal}. We next give two examples on exposure mappings that satisfy Assumption~\ref{asp:exp_map} above.

\begin{example}[Existence of treated neighbors]
\label{example:existence}
Consider an advertising intervention where $w_i=1$ indicates that unit $i$ is shown an advertisement. The exposure mapping records the customer's own treatment and whether at least one of the customer's neighbors is treated:

\[
G_i=\left(w_i, I \left\{\sum_{j\in\nu_n(i)} w_j > 0\right\}\right),
\]
where $I(\cdot)$ denotes the indicator function.
\end{example}
\begin{example}[Number of treated neighbors]
\label{example:amount}
Consider a vaccination study where $w_i=1$ indicates that unit $i$ is vaccinated. Holding the neighborhood structure fixed, individuals with more vaccinated neighbors may face lower infection risk due to spillovers (e.g., herd immunity). To ensure overlap especially when degrees vary across units, a common choice is to consider the \emph{count} of treated neighbors as exposure. The exposure mapping is therefore
\[
G_i=\Phi_n(i,\bm{w},A)=\sum_{j\in\nu_n(i)} w_j.
\]
\end{example}

\subsection{Quantile Exposure Effect}\label{distributional}
We now define the causal estimand of interest. Given $g\in\boldsymbol{G}$, let $\boldsymbol{G}$ denote the support of the exposure mapping, i.e.,
\[
\boldsymbol{G}:= \big\{ g\in \mathbb{R}^{d_{\boldsymbol{G}}}:\ \exists\, i\in N_n,\,\bm w\in\{0,1\}^n \text{ such that } \Phi_n(i,\bm w,A)=g \big\}.
\]
Even though potential outcomes \(\{Y_i(\bm w)\}\) are fixed, the randomization
over \(\bm W\) induces, for each exposure level \(g\), a finite-population
``mixture'' distribution:

\begin{equation}\label{eq:Fng}
F_{ng}(y)=\frac{1}{n}\sum_{i\in N_n}\sum_{\bm{w}\in\{0,1\}^n}I(Y_i(\bm{w})\le y)\mathrm{pr}(\bm{W}=\bm{w}\mid G_i=g),
\end{equation}
where $G_i=\Phi_n(i,\bm W,A)$ is the exposure for unit $i$ as defined in (\ref{Gi}). To interpret \(F_{ng}\), we introduce an auxiliary random index \(\xi\) that is drawn
uniformly from \(N_n=\{1,\ldots,n\}\) and independently of the treatment
assignment \(\bm W\). Define the randomly selected unit's outcome and exposure by
$Y_{\xi}(\bm W) = Y_{\xi}(\bm W_1,\ldots,\bm W_n)
= \sum_{i\in N_n}I\{\xi=i\}\,Y_i(\bm W),$ $G_{\xi} = \Phi_n(\xi,\bm W,A)
= \sum_{i\in N_n}I\{\xi=i\}\,G_i.$ With \(\xi\) introduced this way,
\[
F_{ng}(y)=E\!\left[\mathrm{pr} \big(Y_{\xi}(\bm W)\le y \,\mid \, \xi,\,
G_{\xi}=g\big)\right]
= \frac{1}{n}\sum_{i\in N_n}\mathrm{pr} \big(Y_i(\bm W)\le y \,\mid \, G_i=g\big).
\]
If \(\mathrm{pr}(G_i=g)\) does not depend on \(i\), then the right-hand side
also equals \(\mathrm{pr}(Y_{\xi}(\bm W)\le y \mid G_{\xi}=g)\). We use the subscript $n$ to emphasize that $F_{ng}$ is defined for a finite population of size $n$.

Let $\mathcal{B}$ be the set of all Borel distribution functions on $\mathbb{R}$, and let
$\mathcal{A}:\mathcal{B}\to\mathbb{R}$ be a real-valued attribute of interest applied to $F_{ng}$.
We consider estimands of the form:
\begin{equation}
    \tau_{\mathcal{A}}(g,g')=\mathcal{A}(F_{ng})-\mathcal{A}(F_{ng'}),
    \label{classicdef}
\end{equation} for any two exposure levels $g,g'\in\boldsymbol{G}$. This class of estimands encompasses a variety of causal effects studied under interference. In particular, Proposition~\ref{prop:equiv} shows that it nests the average exposure effect considered in, for example, \citet{leung2022causal,savje2023causal,gao2025causal}.

\begin{proposition}
\label{prop:equiv}
    Let $\mathcal{A}_{Mean}(F)=\int_0^\infty \{1-F(y)\} dy-\int_{-\infty}^0 F(y) dy$. Then 
    $$\tau_{\mathcal{A}_{Mean}}(g,g')=\frac{1}{n}\sum_{i\in N_n} \big\{\mu_i(g)-\mu_i(g')\big\},$$ where
    $\mu_i(g)=\sum_{\bm{w}\in\{0,1\}^n}Y_i(\bm{w})\mathrm{pr}(\bm{W}=\bm{w}\mid G_i=g).$
\end{proposition}

For some fixed $q\in(0,1)$ and any $F\in \mathcal{B}$,
define the q-th quantile functional  
$$\mathcal{A}_q(F)=\inf\{y: F(y)\geq q\}.$$

For each exposure level $g\in\boldsymbol{G}$, let
$\beta_{nq}(g)=\mathcal{A}_q(F_{ng})$. 
\begin{definition}[Quantile exposure effect]
For any $g,g'\in\boldsymbol{G}$, we refer to the quantile treatment effect under interference as the \emph{quantile exposure effect}, defined by
\begin{equation}
\label{eq:qte_def}
\tau_{nq}(g,g')=\beta_{nq}(g)-\beta_{nq}(g').
\end{equation}
\end{definition}
Quantile exposure effect contains a wide range of effects under interference. For instance, under the exposure mapping in Example~\ref{example:existence}, when $g=(1,t)$, $g'=(0,t),$ $t\in\{0,1\},$ we can examine the direct effect which measures the impact from one's own treatment on the outcome. In contrast, exposure pair $g=(t,1)$, $g'=(t,0)$ defines the spillover effect measuring the neighborhood impact in networks. The choice on effect of interest depends on the scientific question one wants to answer in application.

\section{Estimating Quantile Treatment Effect with Interference}\label{sec:estimation}
\subsection{The Estimator}
In the no-interference setting, a classical approach to estimating quantile treatment effects is inverse-propensity-weighted (IPW) quantile regression \citep{sun2021causal}. Under interference, however, each outcome $Y_i$ may depend on the entire assignment vector $\bm W$, so we instead work with the exposure $G_i=\Phi_n(i,\bm W,A)$ and its associated generalized propensity score. Specifically, for each $g\in\boldsymbol{G}$, the generalized propensity score is defined as
\[
\pi_i(g)=\Pr(G_i=g).
\]
In this paper, we focus on randomized experiments where the generalized propensity score can be computed from the treatment assignment mechanism, exposure mapping, covariates, and network. For example, under a Bernoulli design in which each unit is independently treated with probability $r$, the probability of
binary exposure $G_i=I\{w_i+\sum_{j\in\nu_n(i)}w_j>0\}$ can be computed by
\[
\Pr(G_i=1)=1-(1-r)^{d_i+1},
\]
where $d_i=|\nu_n(i)|$. We allow the generalized propensity score to depend on covariates, so the
experiment need not be completely randomized.

We estimate $\beta_{nq}(g)$ via inverse generalized-propensity-weighted quantile regression:
\begin{equation}
\label{eq:minimization}
    \widehat{\beta}_{nq}(g)=\inf\left\{\arg\min_{\beta} \sum_{i \in N_n}\frac{I(G_i=g)}{\pi_i(g)}\rho_q(Y_i-\beta)\right\},
\end{equation} 
where $\rho_q(u)=u (q - I(u\leq0))$ is the check loss function. The resulting estimator of the quantile exposure effect is  $$
\widehat{\tau}_{nq}(g,g')=\widehat{\beta}_{nq}(g)-\widehat{\beta}_{nq}(g').$$ This procedure generalizes the standard IPW estimator for quantile treatment effects by contrasting exposure levels rather than individual treatment indicators. In Section~\ref{sec:asp_and_thm}, we present the main causal and network-structure assumptions and provide additional examples of exposure mappings. We then derive the asymptotic distribution of $\widehat{\tau}_{nq}(g,g')$ and develop a corresponding variance estimator on the proposed estimator.

\subsection{Conditional Neighborhood Dependence}\label{CND}

Our theoretical development builds on the
notion of conditional neighborhood dependence (CND) introduced by ~\cite{lee2019stable}. Informally, CND captures
the idea that network dependence is \emph{local}: once we condition on
information carried by a suitable ``buffer'' neighborhood, outcomes in two
well-separated sets of vertices become conditionally independent. This property enables stable limit theorems for network-dependent data without requiring a parametric model for the network.

Let \(\nu_n\) be a neighborhood system on \(N_n\) and let
\(\mathcal{M}\equiv\{\mathcal{M}_i\}_{i\in N_n}\) be a collection of
\(\sigma\)-fields. For a set \(S\subseteq N_n\), recall that
\(\nu_n(S)=\bar{\nu}_n(S)\setminus S\), define
\(\mathcal{F}_S:=\sigma\!\big(\{Y_i\}_{i\in S}\big)\), and
\(\mathcal{M}_S:=\vee_{i\in S}\mathcal{M}_i\).
For two \(\sigma\)-fields \(\mathcal A\) and \(\mathcal B\),
$\mathcal A \vee \mathcal B$
denotes the \emph{smallest \(\sigma\)-field containing both} \(\mathcal A\) and \(\mathcal B\). Equivalently,
$\mathcal A \vee \mathcal B = \sigma(\mathcal A \cup \mathcal B)$.

\begin{definition}[Conditional neighborhood dependence]
\label{def:CND}
We say that \(\{Y_i\}_{i\in N_n}\) is \emph{conditionally neighborhood dependent
(CND)} with respect to \((\nu_n,\mathcal{M})\) if, for any subset
\(N_n'\subset N_n\) and any \(A,B\subseteq N_n'\) satisfying
\[
A \subseteq N_n'\setminus \bar{\nu}_n(B)
\quad\text{and}\quad
B \subseteq N_n'\setminus \bar{\nu}_n(A),
\]
the \(\sigma\)-fields \(\mathcal{F}_A \vee \mathcal{M}_A\) and
\(\mathcal{F}_B \vee \mathcal{M}_B\) are conditionally independent given
\(\mathcal{M}_{\nu_n(A)}\); that is,
\[
(\mathcal{F}_A \vee \mathcal{M}_A)\ \perp\!\!\!\perp\
(\mathcal{F}_B \vee \mathcal{M}_B)\ \big|\ \mathcal{M}_{\nu_n(A)}.
\]
\end{definition}

Intuitively, the set \(\nu_n(A)\) acts as a separating ``buffer'': once we
condition on \(\mathcal{M}_{\nu_n(A)}\), the outcomes in \(A\) contain no
additional information about outcomes in any set \(B\) that lies outside this
neighborhood (and vice versa). Note that CND is a property of the \emph{entire}
collection \(\{Y_i\}_{i\in N_n}\). It asserts conditional independence of
outcomes in two well-separated subsets after conditioning on appropriate
neighborhood information.

\begin{figure}[ht]
    \centering
    \includegraphics[width=0.6\linewidth]{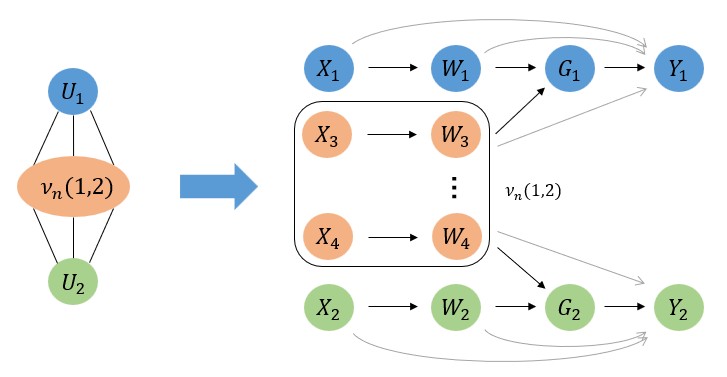}
    \caption{The connection of two units $U_1,U_2$ in the social network (left) and the corresponding causal DAG (right).}
    \label{fig:illu}
\end{figure}

Figure~\ref{fig:illu} provides a simple illustration under
Assumption~\ref{asp:exp_map}and~\ref{asp:local-potential-outcomes}, which implies that the outcome \(Y_i\) exposure mapping \(G_i\) for each unit $i$ is only affected by treatments in the closed neighborhood $\bm{W}_{\bar\nu_n(i)},$ apart from its own treatments and covariates. Besides, for two units \(U_1\) and \(U_2\) which are not adjacent, and \(\nu_n(1,2)\neq\emptyset\) denotes their joint neighborhood. In the DAG, any dependence
between \(Y_1\) and \(Y_2\) can arise only through the treatments and covariates $\{\bm{W}_{\nu_n(1,2)},\bm{X}_{\nu_n(1,2)}\}$ associated with the joint neighborhood $\nu_n(1,2)$. Conditioning on this
neighborhood information therefore blocks the dependence between the two
outcomes. Note that for randomized experiment in finite population discussed in this paper, the covariates are fixed and the dependence arises from $\bm{W}_{\nu_n(1,2)}$ only. However, Figure~\ref{fig:illu} accommodates to observational studies as well, and therefore we keep covariates to increase its adaptability. Consequently, the DAG makes clear that the outcomes satisfy conditional neighborhood dependence 
property stated in the proposition below.

\begin{proposition}[Conditional neighborhood dependence]
\label{prop:CND}
Suppose that the treatment assignments
\(\{W_i:i\in N_n\}\) are mutually independent under the design-based
probability measure. Suppose further that Assumption~\ref{asp:exp_map} and~\ref{asp:local-potential-outcomes} hold, and let
\[
    \mathcal{M}_i=\sigma(W_i),
    \qquad
    \mathcal{M}=\{\mathcal{M}_i\}_{i\in N_n},
\]
Then
\(\{(Y_i,G_i)\}_{i\in N_n}\) is conditionally neighborhood dependent
with respect to \((\nu_n,\mathcal{M})\). More generally, for any collection of
measurable functions \(\{h_i\}_{i\in N_n}\), the array
\[
    \bigl\{h_i(Y_i,G_i)\bigr\}_{i\in N_n}
\]
is CND with respect to \((\nu_n,\mathcal{M})\).
\end{proposition} 

Proposition~\ref{prop:CND} formalizes the idea that dependence in a network is
transmitted locally: outcomes in two disconnected regions are conditionally
independent given the treatments in the relevant neighborhood in finite sample.  Compared with
alternative frameworks for network interference such as approximate neighborhood
interference (ANI) \citep{leung2022causal} and approaches that rely primarily on
Stein's method \citep{li2022random,ogburn2022causal}, CND provides a particularly
convenient route to weak convergence for functionals that are not Lipschitz and
may be non-differentiable in the outcome. This feature is especially important
for quantile-based estimands and other settings with non-smooth objectives, e.g.,
causal inference for survival outcomes~\citep{martinussen2022causality}. This CND structure is central to asymptotic analysis on network dependent variables. To build the stable limit theorem, we further introduce the following corollary on the two-step dependency graph.

\begin{corollary}[Two-step dependency graph]
\label{cor:two-step-dependency}
Suppose that the treatment assignments  are mutually
independent and the conditions in Proposition~\ref{prop:CND} hold. Recall that \(\bar\nu_n^{(2)}(i)\) is the two-step closed neighborhood for individual $i$. Then, for any collection of measurable functions $\{h_i\}_{i\in N_n}$ and any disjoint
$A,B\subseteq N_n$ such that
\[
A\subseteq N_n\setminus\bar\nu_n^{(2)}(B),
\qquad
B\subseteq N_n\setminus\bar\nu_n^{(2)}(A),
\]
the array
\[
S_i=h_i(W_i,G_i,Y_i,X),\qquad i\in N_n,
\]
satisfies that $\{S_i:i\in A\}$ and $\{S_j:j\in B\}$ are independent.
\end{corollary}

\section{Theoretical Properties}
\label{sec:asp_and_thm}

\subsection{Main Assumptions}
In this section, we introduce additional assumptions beyond the exposure mapping in
Assumption~\ref{asp:exp_map}, which are needed to establish the statistical properties of our
estimator for the quantile exposure effect. We begin with assumptions on the generalized
propensity score, followed by regularity conditions on the outcome distribution and
restrictions on the network structure.

\begin{assumption}[Overlap]
\label{asp:overlap}
    For any two exposure levels of interest
    $g,g'\in\boldsymbol{G}$, the generalized propensity
scores satisfy: 
\[
\pi_i(g),\ \pi_i(g') \in [\underline{\pi},\overline{\pi}]\subset(0,1),
\quad\text{for all } i\in N_n,
\]
where $0<\underline{\pi}<\overline{\pi}<1$ are constants.
\end{assumption}

Assumption~\ref{asp:overlap} requires that the generalized propensity scores for the exposure levels under comparison are uniformly
bounded away from $0$ and $1$. This condition can
often be enforced (at least approximately) through the experimental design. Two remarks are
in order. First, overlap under interference is more delicate than in the no-interference setting,
because exposure probabilities may vary substantially with local network features such as
degree. For instance, under the ``existence of a treated neighbor'' exposure in
Example~\ref{example:existence}, units with larger neighborhoods are much more likely to
have $G_i=1$ (and hence much less likely to have $G_i=0$). As discussed in
\citet{leung2022causal}, one way to improve the plausibility of overlap is to randomize
treatment only within a subset of eligible units and to restrict attention to a
subpopulation whose neighborhoods satisfy pre-specified eligibility criteria. Second, overlap is only required for the exposure levels of interest. For
example, under the ``number of treated neighbors'' exposure in
Example~\ref{example:amount}, $\Pr(G_i=5)=0$ for any unit $i$ with
$\sum_{j\in N_n}A_{ij}<5$. If the target contrast is between $g=0$ and $g'=2$ and treatment
is assigned with $\Pr(W_i=1)\in(0,1)$ for all $i$, then it suffices to restrict the
analysis to units with at least two neighbors, for whom both exposure levels are attainable
with positive probability.

Next, we introduce the following assumption on the outcome distribution function. Write $F(x-)=\lim_{\epsilon\to 0+}F(x-\epsilon)$ as the left limit of function $F.$

\begin{assumption}
[Local regularity and existence of limiting density function]
\label{asp:pdf}
Let $q\in (0,1)$ be the quantile level of interest. 
Then for any exposure level of interest $g\in\boldsymbol{G},$ the value of $F_{ng}$ at its $q$-th quantile $\beta_{nq}(g)$ satisfies $$
F_{ng}\{\beta_{nq}(g)\}-F_{ng}\{\beta_{nq}(g)-\}
=o(n^{-1/2}).
$$ 

Moreover, there exists a continuous density function \(f_g\) such that \(0<c_f\le f_g(\beta_{nq}(g))\le C_f<\infty\) for all sufficiently large \(n\). Besides, \(f_g\) is locally Lipschitz around \(\beta_{nq}(g)\), i.e., there are constants $L_f<\infty$ and $\eta_f>0$ such that, uniformly for all sufficiently large $n$,
\[
|f_g(y)-f_g(z)|\le L_f|y-z|
\quad\text{whenever}\quad\max\{
|y-\beta_{nq}(g)|,|z-\beta_{nq}(g)|\}\le\eta_f.
\]
In addition, for every fixed \(M>0\),
\[
\sup_{|t|\le M/\sqrt n}
\left|
F_{ng}\left(\beta_{nq}(g)+t\right)
-
F_{ng}\left(\beta_{nq}(g)\right)
-
f_g\left(\beta_{nq}(g)\right)t
\right|
=
o(n^{-1/2}).
\]    
\end{assumption}

Assumption~\ref{asp:pdf} is a local regularity condition that makes quantile
asymptotics feasible in a design-based, finite-population setting.
While \(F_{ng}\) defined in (\ref{eq:Fng}) is a mixture distribution of potential outcomes and is therefore a discontinuous function, the first part of Assumption~\ref{asp:pdf} assumes that the step size of $F_{ng}$ at its $q$-th quantile is at an order of $o(n^{-1/2}),$ which essentially regularizes $F_{ng}$ by ruling out big jump locally at the quantile. The second part gives
a local linear approximation to the CDF \(F_{ng}\), uniformly on every fixed
multiple of the $n^{-1/2}$ neighborhood in the locality of finite-sample quantile $\beta_{nq}(g)$, with a positive and locally Lipschitz
limiting density. In Assumption~\ref{asp:pdf}, the asymptotic is interpreted as the limit of a sequence of experiments with size $n\to\infty.$ This
design-based interpretation of asymptotics is standard in the interference literature; see, e.g., \citet{leung2022causal,gao2025causal, shi2026berry}. Assumption~\ref{asp:pdf} is applied to approximate the estimator with the maximum of a continuous quadratic process to prove the asymptotic property of the quantile estimator under interference.

Next, we introduce notation that will be used in the asymptotic expansions. Let $\mathcal{U}$ be an
open interval containing $0$. 
For exposure level $g\in\boldsymbol{G}$ and $u\in\mathcal{U}$, define 
$$
R_{q}(Y_i,u;g)=\left\{Y_i-(\beta_{nq}(g)+u/\sqrt{n})\right\}\left\{I(Y_i\leq \beta_{nq}(g))-I(Y_i\leq \beta_{nq}(g)+u/\sqrt{n})\right\},
$$
and
$$D_{q}(Y_i;g)=q-I\{Y_i\leq \beta_{nq}(g)\}.$$
Their inverse-generalized-propensity--weighted versions are defined as

\[
\psi_{i,q}(y,g):=I\{G_i=g\}\,\pi_i^{-1}(g)\,H_{g,q}(y),
\]
where $H_{g,q}(y):=f_g^{-1}(\beta_{nq}(g))D_{q}(y;g).$ For fixed $g,g'\in\boldsymbol{G}$, define the contrast
\begin{equation}
\label{eq:def_psi}
    \psi_{i,q}=\psi_{i,q}(Y_i,g)-\psi_{i,q}(Y_i,g').
\end{equation} 
These quantities will appear repeatedly in our asymptotic expansions. 

For the neighborhood network, write the node degrees as
\[
d_i=|\nu_n(i)|,\qquad
d_{mx}=\max_{i\in N_n}d_i,\qquad
d_{\mathrm{av}}=\frac1n\sum_{i\in N_n}d_i.
\]
And for the two-step dependency graph in
Corollary~\ref{cor:two-step-dependency}, define
\[
d_i^{(2)}=|\nu_n^{(2)}(i)|,\qquad
d_{mx}^{(2)}=\max_{i\in N_n}d_i^{(2)},\qquad
d_{\mathrm{av}}^{(2)}=\frac1n\sum_{i\in N_n}d_i^{(2)}.
\]
It is straightforward that
\[
d_{mx}^{(2)}\le d_{mx}^2,
\qquad
d_{\mathrm{av}}^{(2)}
\le d_{mx}d_{\mathrm{av}}.
\]

For an array $M=\{M_i\}_{i\in N_n}$, let
\[
\sigma_n^2(M)
=
\operatorname{Var}\left(\sum_{i\in N_n}M_i\right),
\qquad
\mu_p(M)
=
\max_{i\in N_n}
\left[
E\left|
\frac{M_i-EM_i}{\sigma_n(M)}
\right|^p
\right]^{1/p},
\quad p=3,4,
\]
and define
\[
a_n^{(2)}(M)
=
n \{1+d_{mx}^{(2)}\}\{1+d_{\mathrm{av}}^{(2)}\}
\mu_3^3(M),
\qquad
b_n^{(2)}(M)
=
n\{1+d_{mx}^{(2)}\}^2\{1+d_{\mathrm{av}}^{(2)}\}
\mu_4^4(M).
\]

\begin{assumption}[Degree and moment conditions]
\label{asp:degree_rate}
Assume that
\[
a_n^{(2)}(\psi_{\cdot,q})\longrightarrow0
\]
and
\[
\left|\log a_n^{(2)}(\psi_{\cdot,q})\right|
\left\{b_n^{(2)}(\psi_{\cdot,q})\right\}^{1/2}
\longrightarrow0.
\]
\end{assumption}

Assumption~\ref{asp:degree_rate} is used to establish the asymptotic normality of the estimator in~\eqref{eq:minimization}. The assumption controls the joint impact of
network connectivity through degree conditions $d_{mx}^{(2)}$ and $d_{\mathrm{av}}^{(2)}$ and 
standardized score moments conditions. Alternative approaches characterize network restrictions
through a graph formation model
(e.g., \citealp{li2022random}) or restrictions on how
potential outcomes vary with network distance
\citep{leung2022causal}. These perspectives are closely
related because interference depends jointly on the
network structure and the potential-outcome model. A convenient sufficient condition,
within the regime $d_{mx}^{(2)}=o(\sqrt n)$, is
$d_{mx}^{(2)}d_{av}^{(2)}=o(\sqrt n)$.
In particular, $d_{mx}^{(2)}=o(n^{1/4})$ is sufficient.
Section ~\ref{sec:alternative-clt}  derives these implications and compares them
with an alternative condition based on Janson's dependency-graph CLT.

\subsection{Asymptotic Normality}

We now establish the asymptotic normality of the quantile exposure effect estimator. Our analysis follows a design-based
finite-population framework, with asymptotics taken along a sequence of fixed finite populations and networks
$\{\mathcal G_n\}$ as $n\to\infty$ \citep{aronow2017estimating, li2017general, leung2022causal,gao2025causal}.
At each sample size, the network, covariates, and
potential-outcome schedule are fixed, and randomness arises
only from treatment assignment. The assumptions stated above
restrict the sequence of experiments, including their
assignment mechanisms, outcome distributions, and network
dependence structures.

\begin{theorem}[Asymptotic normality]
\label{thm:asy_nor}
For
$q\in(0,1)$ and exposure levels $g,g'\in\boldsymbol G$, suppose that
Assumptions~\ref{asp:exp_map}, \ref{asp:local-potential-outcomes},
\ref{asp:overlap}, \ref{asp:pdf}, and~\ref{asp:degree_rate} hold. Define
\[
\sigma_n^2
=
\operatorname{Var}\left[
n^{-1/2}\sum_{i\in N_n}
\left\{
\psi_{i,q}(Y_i,g)-\psi_{i,q}(Y_i,g')
\right\}
\right].
\]
Suppose also
that, for some constants $0<c_\sigma<C_\sigma<\infty$ and all sufficiently
large $n$,
\[
\max_{s\in\{g,g'\}}
\operatorname{Var}\left[
n^{-1/2}\sum_{i\in N_n}\psi_{i,q}(Y_i,s)
\right]\le C_\sigma,
\qquad
\sigma_n^2\ge c_\sigma.
\]
Then if \(d_{mx}^{(2)}=o(n^{1/2}),\)
we have
\[
\sqrt n\{
\widehat\tau_{nq}(g,g')-\tau_{nq}(g,g')
\}
=
n^{-1/2}\sum_{i\in N_n}
\left\{
\psi_{i,q}(Y_i,g)-\psi_{i,q}(Y_i,g')
\right\}
+o_p(1).
\]
And therefore
\[
\sqrt n\sigma_n^{-1}\{
\widehat\tau_{nq}(g,g')-\tau_{nq}(g,g')
\}
\ \longrightarrow\ N(0,1)
\]
in distribution.
\end{theorem}

Theorem~\ref{thm:asy_nor} establishes asymptotic normality by first obtaining a linear expansion on the QEE estimator. 
 There are two fundamental challenges in establishing the theory. First, comparing to the asymptotic theories for average exposure effect with interference, a key difference is that the random variable $\psi_{i,q}(y,g)$ is not a Lipschitz function with respect to $y$. This non-smoothness makes it difficult to directly apply network CLTs
developed for Lipschitz functionals, such as those used under approximate neighborhood
interference \citep{leung2022causal}. We instead exploit the conditional neighborhood
dependence (CND) property, which allows us to apply stable
limit theorems under moment and degree conditions as in Assumption~\ref{asp:degree_rate}. Second, $\psi_{i,q}(y,g)$ involves the limiting density function in Assumption~\ref{asp:pdf}. In a
finite-population setting, the distribution $F_{ng}$ need not admit a density. We address
this by assuming that $F_{ng}$ is well approximated by a smooth limit, allowing us to use the derivative $f_g$ to justify the
local linearization underlying the expansion.

The proof sketch is as follows, with full details provided in the supplement. Define $U_{q,n}(\beta;g)=\sum_{i\in N_n}I(G_i=g)\pi_i^{-1}(g)\rho_q(Y_i-\beta)$ so that $\widehat{\beta}_{nq}(g)=\arg\min_{\beta}U_{q,n}(\beta;g)$. 
To derive the asymptotic distribution of
$\widehat{u}_{q}=\sqrt{n}(\widehat{\beta}_{nq}(g)-\beta_{nq}(g))$, 
we re-parameterize 
and consider the equivalent
optimization problem
\[
\widehat{u}_q=\arg\min_{u\in\mathbb{R}} Q_q(u;g),
\]
where
\[
Q_q(u;g)
=
U_{q,n}\{\beta_{nq}(g)+u/\sqrt n;g\}
-
U_{q,n}\{\beta_{nq}(g);g\}.
\]
Following the technique in~\citet{hjort2011asymptotics},
we obtain a quadratic
approximation to the localized objective. Specifically, we construct a quadratic function
$\widetilde{Q}_q(u;g)$ with regard to $u$ and show that
\[
\sup_{u\in\mathcal{U}}\big|Q_q(u;g)-\widetilde{Q}_q(u;g)\big|=o_p(1),
\]
by controlling the remainder terms using network-dependent arguments, in particular, the
conditional neighborhood dependence property. 
By a convexity argument (Lemma 2 of \citet{hjort2011asymptotics}), 
this uniform approximation implies nearness of minimizers:
\[
|\widehat{u}_q-\widetilde{u}_q|=o_p(1),
\qquad
\widetilde{u}_q=\arg\min_{u\in\mathbb{R}}\widetilde{Q}_q(u;g).
\]
Moreover, the exposure-specific quadratic minimizer is tight and admits the linear representation
\[
\widetilde{u}_q=n^{-1/2}\sum_{i\in N_n}\psi_{i,q}(Y_i,g)=O_p(1),
\]
which yields the leading term in the asymptotic expansion. The asymptotic normality of the centered leading term follows by
applying Corollary 3.1 of Lee and Song (2019) to the exact two-step
dependency graph established in
Corollary~\ref{cor:two-step-dependency}, with trivial conditioning
sigma-fields. The degree and moment requirements are precisely those
in Assumption~\ref{asp:degree_rate}.

\subsection{Sufficient Network Conditions and an Alternative CLT}
\label{sec:alternative-clt}

Theorem~\ref{thm:asy_nor} establishes asymptotic normality
using the Lee--Song normal approximation under
Assumption~\ref{asp:degree_rate}. This assumption combines
network degrees with standardized moments of the contrast
scores. We first derive interpretable sufficient degree
conditions and then present an alternative based on Janson's
dependency-graph central limit theorem. The resulting
conditions differ in how they accommodate network
heterogeneity.

For the fixed exposure contrast $(g,g')$, let
\[
Z_{i,n}=\psi_{i,q}-E\psi_{i,q},\qquad
S_n=\sum_{i\in N_n}Z_{i,n},\qquad
s_n^2=\operatorname{Var}(S_n)=n\sigma_n^2.
\]
By Corollary~\ref{cor:two-step-dependency}, this two-step graph is an exact
dependency graph for the array $\{Z_{i,n}:i\in N_n\}$.

\medskip
\noindent\textbf{A degree-based implication of the Lee--Song conditions.}
Under Assumptions~\ref{asp:overlap} and~\ref{asp:pdf}, the centered contrast scores are
uniformly bounded: $|Z_{i,n}|\le B:=4/(\underline\pi c_f)$ for all sufficiently
large $n$. If $\liminf_n\sigma_n^2>0$, then
$s_n\gtrsim\sqrt n$, so
$\mu_3^3(\psi_{\cdot,q})=O(n^{-3/2})$ and
$\mu_4^4(\psi_{\cdot,q})=O(n^{-2})$.
Consequently, the two quantities in Assumption~\ref{asp:degree_rate} satisfy
\[
a_n^{(2)}(\psi_{\cdot,q})
=O\left(\frac{\{1+d_{mx}^{(2)}\}\{1+d_{\mathrm{av}}^{(2)}\}}{\sqrt n}\right),
\qquad
b_n^{(2)}(\psi_{\cdot,q})
=O\left(\frac{\{1+d_{mx}^{(2)}\}^2\{1+d_{\mathrm{av}}^{(2)}\}}{n}\right).
\]
Moreover, boundedness yields
\[
b_n^{(2)}(\psi_{\cdot,q})
\leq \frac{B \{1+d_{mx}^{(2)}\}}{s_n}a_n^{(2)}(\psi_{\cdot,q})
\lesssim \frac{1+d_{mx}^{(2)}}{\sqrt n}a_n^{(2)}(\psi_{\cdot,q}).
\]
Thus, within the regime $d_{mx}^{(2)}=o(\sqrt n)$ required for the local quadratic
approximation, a convenient sufficient condition is
\begin{equation}
\label{eq:lee-song-product}
\{1+d_{mx}^{(2)}\}\{1+d_{\mathrm{av}}^{(2)}\}=o(\sqrt n)
\end{equation}
Indeed, this condition implies
$a_n^{(2)}(\psi_{\cdot,q})\to0$ and
$d_{mx}^{(2)}=o(\sqrt n)$ because $1+d_{\mathrm{av}}^{(2)}\ge1$.
The logarithmic condition in
Assumption~\ref{asp:degree_rate} then follows from the
preceding bound and
$x^{1/2}|\log x|\to0$ as $x\downarrow0$. A simpler sufficient condition
follows from $d_{av}^{(2)}\le d_{mx}^{(2)}$ as below.

\begin{proposition}[Maximum-degree sufficient condition]
\label{prop:max-degree-two-step}
Suppose that treatment assignments are mutually independent
and Assumptions~\ref{asp:exp_map},
\ref{asp:local-potential-outcomes}, \ref{asp:overlap},
and~\ref{asp:pdf} hold. Suppose further that, for constants
$0<c_\sigma<C_\sigma<\infty$ and all sufficiently large $n$,
\begin{equation}
\label{eq:clt-variance-bounds}
\max_{s\in\{g,g'\}}
\operatorname{Var}\left[
n^{-1/2}\sum_{i\in N_n}\psi_{i,q}(Y_i,s)
\right]\le C_\sigma,
\qquad
\sigma_n^2\ge c_\sigma.
\end{equation}
If $d_{mx}^{(2)}=o(n^{1/4})$, then
Assumption~\ref{asp:degree_rate} holds.
\end{proposition}

Proposition~\ref{prop:max-degree-two-step} provides a
sufficient condition requiring only the maximal two-step
degree. imilar constraints are also required in related works such as~\citet{aronow2017estimating}. It follows directly from
$\{1+d_{mx}^{(2)}\}\{1+d_{av}^{(2)}\}\le \{1+d_{mx}^{(2)}\}^2$ and the preceding argument.
For approximately regular two-step graphs, where
$d_{av}^{(2)}\asymp d_{mx}^{(2)}$, this maximum-degree condition
also characterizes the product restriction
in~\eqref{eq:lee-song-product}.

\medskip
\noindent\textbf{An alternative based on Janson's CLT.}
Janson's dependency-graph CLT~\citep{janson1988normal}, restated in
\citet[Theorem~4.3]{feray2018weighted}, applies to centered random variables
bounded by $A_n$ whose dependency graph has maximal degree $\Delta_n-1=d_{mx}^{(2)}$.  If
$v_n^2$ is the variance of their sum and, for some fixed integer $m\geq3$,
\[
\left(\frac{n}{\Delta_n}\right)^{1/m}
\frac{\Delta_n A_n}{v_n}\longrightarrow0,
\]
then the standardized sum converges in distribution to $N(0,1)$.  In our
setting, take $A_n$ as constant, and
$v_n=s_n=\sqrt n\,\sigma_n$. The relevant condition is
formalized as follows.

\begin{assumption}[Alternative dependency-graph condition]
\label{asp:janson}
The variance bounds in~\eqref{eq:clt-variance-bounds} hold.
In addition, for some fixed integer $m\ge3$,
\begin{equation}
\label{eq:janson-condition}
\left(\frac{n}{\Delta_n}\right)^{1/m}
\frac{\Delta_n}{\sqrt n\,\sigma_n}
\longrightarrow0.
\end{equation}
\end{assumption}

When $\sigma_n\asymp1$ from the variance upper and lower bound in Proposition~\ref{prop:max-degree-two-step}, for fixed $m$,
condition~\eqref{eq:janson-condition} is equivalent to
\begin{equation}
\label{eq:janson-degree-rate}
\Delta_n=o\left(n^{(m-2)/\{2(m-1)\}}\right).
\end{equation}
For $m=6$, the exponent is $2/5$, and it approaches $1/2$
as $m$ increases. For any fixed $\delta\in(0,1/2)$, choosing a fixed integer
$m>1+1/(2\delta)$ gives 
$\frac{m-2}{2(m-1)}>1/2-\delta$.
Therefore, $d_{mx}^{(2)}=O(n^{1/2-\delta})$ is sufficient
for a suitable fixed $m$. However, this does not cover every
sequence satisfying $d_{mx}^{(2)}=o(\sqrt n)$, for example, $d_{mx}^{(2)}\asymp\sqrt n/\log n$
fails~\eqref{eq:janson-degree-rate} for every fixed $m$.

\begin{theorem}[Asymptotic normality under Janson's condition]
\label{thm:asy-nor-janson}
Suppose that the treatment assignments $\{W_i:i\in N_n\}$ are mutually
independent.  For $q\in(0,1)$ and
$g,g'\in\boldsymbol{G}$, suppose that
Assumptions~\ref{asp:exp_map}, \ref{asp:local-potential-outcomes},
\ref{asp:overlap}, \ref{asp:pdf}, and~\ref{asp:janson} hold.  Then
\[
\sqrt n\{\widehat\tau_{nq}(g,g')-\tau_{nq}(g,g')\}
=n^{-1/2}\sum_{i\in N_n}
\{\psi_{i,q}(Y_i,g)-\psi_{i,q}(Y_i,g')\}+o_p(1),
\]
and
\[
\sqrt n\,\sigma_n^{-1}
\{\widehat\tau_{nq}(g,g')-\tau_{nq}(g,g')\}
\ \longrightarrow\ N(0,1)
\]
in distribution, where $\sigma_n^2$ is defined in
Theorem~\ref{thm:asy_nor}.
\end{theorem}

Theorem~\ref{thm:asy-nor-janson} establishes asymptotic normality for the QEE
estimator under Janson's condition. Comparing with Theorem~\ref{thm:asy_nor}, it
replaces the Lee--Song moment-and-degree condition in
Assumption~\ref{asp:degree_rate} with Assumption~\ref{asp:janson}. For the fixed
$m$ appearing in that assumption, rate~\eqref{eq:janson-degree-rate} implies
$d_{mx}^{(2)}=o(n^{1/2})$, so it also supplies the maximal-degree restriction
needed for the local quadratic expansion. The proof can therefore reuse that
expansion and invoke Janson's CLT only in the final step for the centered
linearized contrast score.

\medskip
\noindent\textbf{Comparison across network structures.}
To make more concrete comparison, we provide the following examples comparing the sufficient product
condition~\eqref{eq:lee-song-product} with the Janson
rate~\eqref{eq:janson-degree-rate}. Throughout, the
boundedness and variance conditions discussed above are
maintained. These are comparisons of sufficient conditions,
rather than necessary restrictions for asymptotic normality. To make the degree restrictions concrete, define
$\gamma_m=\frac{m-2}{2(m-1)}$,$m\geq 3$,
so that the sufficient Janson rate in~\eqref{eq:janson-degree-rate} is
$d_{mx}^{(2)}=o(n^{\gamma_m})$. 

\begin{example}[General networks]
\label{ex:general-clt-network}
Suppose $d_{mx}^{(2)}\asymp n^\alpha$ and
$d_{\mathrm{av}}^{(2)}\asymp n^\beta$, where
$0\leq\beta\leq\alpha$.
The Lee--Song product condition holds when
$\alpha+\beta<1/2$, which also implies
$\alpha<1/2$.
Janson's condition holds when
$\alpha<\gamma_m$ for some fixed $m\geq3$ and does not involve $\beta$.
\end{example}

\begin{example}[Approximately regular networks]
\label{ex:regular-clt-network}
Suppose the two-step dependence graph is approximately regular in the sense that vertex degrees are similar so that
$d_{\mathrm{av}}^{(2)}\asymp d_{mx}^{(2)}$.  The Lee--Song product condition
becomes $\{d_{mx}^{(2)}\}^2=o(\sqrt n)$, or
$d_{mx}^{(2)}=o(n^{1/4})$.  
Janson's condition instead permits
$d_{mx}^{(2)}=o(n^{\gamma_m})$.  
Since $\gamma_m\uparrow1/2$, 
this allows faster maximal-degree growth.
This regime includes networks that are disjoint unions of complete communities of comparable sizes, for which the maximal and average two-step neighborhood sizes have the same asymptotic order.
\end{example}

\begin{example}[Networks with localized hubs]
\label{ex:hub-clt-network}
Suppose the original network contains $k_n$ disjoint star components, each
with $h_n$ leaves, while all remaining units have uniformly bounded two-step
neighborhoods.  Squaring the original graph turns each star into a clique of
size $h_n+1$, and therefore
\[
d_{mx}^{(2)}\asymp h_n,
\qquad
d_{\mathrm{av}}^{(2)}
=O\left(1+\frac{k_nh_n^2}{n}\right).
\]
The Lee--Song product condition becomes
\[
h_n\left(1+\frac{k_nh_n^2}{n}\right)=o(\sqrt n).
\]
If $k_n=O(1)$, this permits $h_n=o(\sqrt n)$.  If the stars contain a
nonvanishing fraction of the population, so that $k_nh_n\asymp n$, then
$d_{\mathrm{av}}^{(2)}\asymp h_n$ and the condition reduces to
$h_n=o(n^{1/4})$. Janson's condition
requires $h_n=o(n^{\gamma_m})$ in both cases.
Thus, the Lee--Song product condition can be less restrictive when hubs are few, whereas Janson's condition can be less restrictive when hubs are widespread.
\end{example}

\begin{example}[Community networks with heterogeneous sizes]
\label{ex:community-clt-network}
Suppose the original network is a disjoint union of complete communities with
sizes $c_{1,n},\ldots,c_{K_n,n}$ and $\sum_kc_{k,n}=n$.  The two-step graph is
the same union of cliques.  Writing $c_{\max,n}=\max_kc_{k,n}$ gives
\[
d_{mx}^{(2)}=c_{\max,n}-1,
\qquad
d_{\mathrm{av}}^{(2)}
=\frac1n\sum_{k=1}^{K_n}c_{k,n}(c_{k,n}-1).
\]
When community sizes are comparable,
$d_{\mathrm{av}}^{(2)}\asymp d_{mx}^{(2)}\asymp c_{\max,n}$, the comparison in
Example~\ref{ex:regular-clt-network} applies.
The present example also allows unequal community sizes,
for which this comparability condition may fail.
In general, Lee--Song's product
condition requires
\[
\frac{c_{\max,n}}{n}\sum_{k=1}^{K_n}c_{k,n}^2
=o(\sqrt n),
\]
whereas Janson's condition requires
$c_{\max,n}=o(n^{\gamma_m})$.  

If only a fixed number of communities are large
and the others have bounded size, then
$d_{\mathrm{av}}^{(2)}=O(1+c_{\max,n}^2/n)$. The Lee--Song product condition then permits
$c_{\max,n}=o(\sqrt n)$, 
while Janson's condition
continues to require a fixed polynomial gap  below $\sqrt n$. Thus, concentrating large community sizes in a small
part of the population can substantially relax the
Lee--Song restriction relative to the balanced case.
\end{example}

Table~\ref{tab:clt-degree-comparison} summarizes the sufficient rates, where
the first two columns specify the network regime and its typical average-degree
behavior. The comparison shows that Janson's
condition can accommodate faster degree growth in approximately regular networks, while the Lee--Song product condition can be less restrictive when high connectivity is concentrated in a small part of the
network.

\begin{table}[ht]
	\caption{Comparison of sufficient degree rates under the Lee--Song's and Janson's conditions.}
	\label{tab:clt-degree-comparison}
	\centering
	\footnotesize
	\setlength{\tabcolsep}{4pt}
	\begin{tabular}{p{0.25\textwidth}p{0.19\textwidth}
			p{0.22\textwidth}p{0.20\textwidth}}
		\hline
		\raggedright Network structure
		& Typical $d_{\mathrm{av}}^{(2)}$
& Lee--Song's rate
& Janson's rate \\
\noalign{\vskip 2.5pt}
\hline
\noalign{\vskip 2.5pt}
\raggedright Approximately regular balanced community
		& $d_{\mathrm{av}}^{(2)}\asymp d_{mx}^{(2)}$
		& $d_{mx}^{(2)}=o(n^{1/4})$
		& $d_{mx}^{(2)}=o(n^{\gamma_m})$ \\
		
		\raggedright Sparse with slowly growing neighborhoods
		& $d_{\mathrm{av}}^{(2)}\asymp n^{1/8}$
		& $d_{mx}^{(2)}=o(n^{3/8})$
		& $d_{mx}^{(2)}=o(n^{\gamma_m})$ \\
		
		\raggedright Sparse localized hubs
		& $d_{\mathrm{av}}^{(2)}\asymp1$
		& $d_{mx}^{(2)}=o(n^{1/2})$
		& $d_{mx}^{(2)}=o(n^{\gamma_m})$ \\
		
		\raggedright Polynomial growth
& $d_{\mathrm{av}}^{(2)}\asymp n^\beta$
& $d_{mx}^{(2)}=o(n^{1/2-\beta})$
& $d_{mx}^{(2)}=o(n^{\gamma_m})$ \\[2.5pt]
\hline
	\end{tabular}
	\begin{tabnote}
		$\gamma_m=(m-2)/\{2(m-1)\}$ and $\gamma_m\uparrow1/2$ as the fixed integer
		$m$ increases. Relations
		$d_{mx}^{(2)}\leq d_{mx}^2$ and
		$d_{\mathrm{av}}^{(2)}\leq d_{mx}d_{\mathrm{av}}$ give the conservative
		rate condition on one-step dependence graph.
	\end{tabnote}
\end{table}

The two sets of sufficient conditions are therefore complementary. We retain the Lee--Song formulation in the main theorem because it makes the role of both maximal and average connectivity explicit, and provide
Janson's condition as an alternative.
Furthermore, \citet{lee2019stable} also derives uniform convergence theorem on network-dependent variables, which is essential to prove the statistical inference validity in the next section.

\section{Inference on Quantile Exposure Effect}
Next, we consider estimation of the finite sample variance $\sigma_n^2.$ We begin by studying a plug-in variance estimator and showing that it induces bias due to finite-sample heterogeneity. Since the bias term need
not be nonnegative, we then construct a conservative variance estimator via positive-part spectral correction and proves its asymptotically conservative coverage
under an explicit spectral-rate condition.
\medskip

\textbf{Plug-in Variance Estimator.} By the definition of the contrast score in
\eqref{eq:def_psi},
\[
\begin{aligned}
\frac1n\sum_{i\in N_n}E\psi_{i,q}
={}&
\frac{q-F_{ng}\{\beta_{nq}(g)\}}
{f_g\{\beta_{nq}(g)\}}
-
\frac{q-F_{ng'}\{\beta_{nq}(g')\}}
{f_{g'}\{\beta_{nq}(g')\}}
=o(n^{-1/2}).
\end{aligned}
\]
The final equality follows because the distance between $q$ and the value of
each finite-population distribution at its quantile is bounded by the jump at
that quantile, which is $o(n^{-1/2})$ under Assumption~\ref{asp:pdf}.
Consequently,
\[
\frac1nE\left(\sum_{i\in N_n}\psi_{i,q}\right)^2
=\frac1n\operatorname{Var}\left(\sum_{i\in N_n}\psi_{i,q}\right)
+\frac1n\left(\sum_{i\in N_n}E\psi_{i,q}\right)^2
=\sigma_n^2+o(1).
\]
For each pair of units $i,j\in N_n$ such that $\ell_{A}(i,j)>2$, $\bar{\nu}_n(i)\cap\bar{\nu}_n(j)=\emptyset,$ and $\psi_{i,q}$, $\psi_{j,q}$ are therefore independent with $E\{(\psi_{i,q}-E\psi_{i,q})(\psi_{j,q}-E\psi_{j,q})\}=0$ since they do not share common neighbors. This motivates the following variance estimator
that sums only cross-products for units $i,j$ with distance $\ell_{A}(i,j)\le 2$,
\begin{equation}
\label{eq:var_est}
    \widehat{\sigma}_n^2=n^{-1}\sum_{i,j\in N_n,\ell_{A}(i,j)\le 2}(\widehat{\psi}_{i,q}-\widehat{\psi}^M_q)(\widehat{\psi}_{j,q}-\widehat{\psi}^M_q),
\end{equation}
where \(\widehat{\psi}_{i,q}=\widehat{\psi}_{i,q}(Y_i,g)-\widehat{\psi}_{i,q}(Y_i,g')\), and $\widehat{\psi}^M_q=n^{-1}\sum_{i \in N_n}\widehat{\psi}_{i,q}$. Here, the plug-in estimator 
is $\widehat{\psi}_{i,q}(y,g)=I(G_i=g)\pi_i^{-1}(g)\widehat{H}_{g,q}(y),$ where $\widehat{H}_{g,q}(Y_i)=\widehat{f}_{ng}(\widehat{\beta}_{nq}(g))^{-1}\big\{ q - I(Y \leq \widehat{\beta}_{nq}(g)) \big\}$
and $\widehat{f}_{ng}$ is the density estimator. For the case with no interference, the estimator reduces to $\widehat{\sigma}_n^2=\frac{1}{n}\sum_{i\in N_n}(\widehat{\psi}_{i,q}-\widehat{\psi}^M_q)^2$,
which coincides with
the variance estimator for standard IPW quantile treatment effects \citep{firpo2007efficient}. Using this variance estimator, we summarize our procedure for estimating the quantile
exposure effect and its variance in Algorithm~\ref{algo:main}.

Next, to discover statistical properties on the variance estimator, mention that $\widehat{\psi}_{i,q}$ depends on the density estimator $\widehat{f}_{ng}$, we start by
imposing the following assumption to control the density estimation error.
\begin{assumption}[Error on Density Function]
\label{asp:error_pdf}
There exists a sequence $\epsilon_n\to0$ such that, for every fixed
$M<\infty$ and each exposure level $s\in\{g,g'\}$, the estimated density
function satisfies
\[
\sup_{y\in\mathcal I_{n,M}(s)}
\left|
\widehat f_{ns}(y)-f_s(y)
\right|
=
O_p(\epsilon_n),
\qquad
\mathcal I_{n,M}(s)
=
\left\{y: |y-\beta_{nq}(s)|\le M/\sqrt n\right\}.
\]
\end{assumption}

In practice, we adopt the kernel density estimator $$\widehat{f}_{ng}(y)=(nh)^{-1}\sum_{i\in N_n}I(G_i=g)\pi_i^{-1}(g)K_1((Y_i-y)/h),$$ where $K_1$ is Gaussian kernel and $h$ is the bandwidth parameter. The following proposition shows the validity of Assumption~\ref{asp:error_pdf} for regular kernel estimators.

\begin{proposition}
\label{prop:kernel_pdf_main}
Under the local-uniform bandwidth-scale regularity conditions in the
supplement, let $h=h_n$ satisfy $h\to0$ and
$\{1+d_{mx}^{(2)}\}\log n/(nh)\to0$. For every fixed $M<\infty$ and
every $g\in\boldsymbol G$,
    \[
\sup_{y\in\mathcal I_{n,M}(g)}
\left|
\widehat f_{ng}(y)-f_g(y)
\right|
=
O_p\left(
\sqrt{(1+d_{mx}^{(2)})\log n/(nh)}
+
h^2
\right).
\]
In particular, under the optimal bandwidth
$h \asymp
\left\{(1+d_{mx}^{(2)})\log n/n\right\}^{1/5},$
    \[
    \sup_{y\in\mathcal I_{n,M}(g)}
\left|
\widehat f_{ng}(y)-f_g(y)
\right|
=
O_p\left[
\left\{(1+d_{mx}^{(2)})\log n/n\right\}^{2/5}\right].
    \]
\end{proposition}

Comparing with the properties on density estimator leveraged in previous studies on QTE such as~\citet{firpo2007efficient}, Proposition~\ref{prop:kernel_pdf_main} steps forward in two aspects. First, it allows  the existence of interference which accounts for $d_{mx}^{(2)}$ in the rate as it captures correlation between units. Second, it extends point-wise convergence on kernel estimator by deriving the locally uniform convergence rate, which appends $\log n$ term in the convergence rate. As demonstrated in the appendix, the regularity conditions required in Proposition~\ref{prop:kernel_pdf_main} on kernel functions hold for common compactly supported kernels such as Epanechnikov kernel and uniform kernel, as well as for kernels with unbounded support such as Gaussian kernel $K_1(x)=\frac{1}{\sqrt{2\pi}}e^{-x^2/2}$ when bandwidth $h$ is at a polynomial rate with respect to $n.$ Specifically when $d_{mx}^{(2)}=o(n^{1/2})$ as required in Theorem~\ref{thm:asy_nor}, Proposition~\ref{prop:kernel_pdf_main} implies $\sup_{y\in\mathcal I_{n,M}(g)}
\left|
\widehat f_{ng}(y)-f_g(y)
\right|=o_p[\{(\log n)^2/n\}^{1/5}].$

We further impose the following degree and moment condition, which will be used to
analyze the variance estimator. Denote $\psi_q^M=n^{-1}\sum_{i\in N_n} \psi_{i,q}$ as the sample mean of $\psi_{i,q}.$

\begin{assumption}[Degree condition for variance estimation]\label{asp:degree_rate_var}
Suppose that
\[
\left(1+d_{av}^{(2)}\right)
\left(1+d_{mx}^{(2)}\right)^2=o(n)
\]
and
\[
\left(1+d_{mx}^{(2)}\right)\left(\epsilon_n+n^{-1/2}\right)=o(1).
\]
\end{assumption}

The first condition in Assumption~\ref{asp:degree_rate_var} follows from
$d_{mx}^{(2)}=o(n^{1/3})$ because
$d_{av}^{(2)}\le d_{mx}^{(2)}$.  The second condition separately links network
growth to the density-estimation rate and is not implied by the degree rate
alone.  Thus variance estimation generally requires a stronger degree
restriction than the $d_{mx}^{(2)}=o(n^{1/2})$ condition used for the local
quadratic approximation in Theorem~\ref{thm:asy_nor}.

Combining the conditions above, we obtain the following result on the asymptotic behavior
of the variance estimator in~\eqref{eq:var_est}.
\begin{theorem}
\label{thm:var_est}
Suppose the conditions of Theorem~\ref{thm:asy_nor} and
Assumptions~\ref{asp:error_pdf} and~\ref{asp:degree_rate_var} hold. Then
\[
\widehat{\sigma}_n^2=\sigma_n^2+\operatorname{Bias}_n+o_p(1),
\]
where
\[
\operatorname{Bias}_n
=\frac1n\sum_{\ell_A(i,j)\le2}
(E\psi_{i,q}-E\psi_q^M)(E\psi_{j,q}-E\psi_q^M),
\qquad
E\psi_q^M=\frac1n\sum_{i\in N_n}E\psi_{i,q}.
\]
\end{theorem}

The leading centering term in Theorem~\ref{thm:var_est} appears because the variance estimator centers each score by the
\emph{sample average} $\psi_q^M$, whereas the target variance $\sigma_n^2$ is
defined using the \emph{individual expectations} $E\psi_{i,q}$.
Thus, the bias occurs because the scores $\psi_{i,q}$ may have
\emph{heterogeneous nonzero means}. The true variance subtracts each unit's
own mean $E\psi_{i,q}$, but the estimator subtracts only the common sample
mean $\psi_q^M$.  The supplementary proof also accounts for the covariance
change induced by subtracting the random sample mean and shows that it is
$o(1)$ under Assumption~\ref{asp:degree_rate_var}.  Thus
$\operatorname{Bias}_n$ is the leading, rather than exact finite-sample,
centering contribution.  A similar term arises in estimating the variance of the average treatment effect under interference
\citep{leung2022causal}, analogous
to the finite-sample bias of Neyman-style variance estimators \citep{imbens2015causal}.
In the special case where
$E\psi_{i,q}=E\psi_q^M$, for all $i$,
$\operatorname{Bias}_n=0$.
Similarly, if the score is exactly centered for every unit,
$E\psi_{i,q}=0$, for all $i$, then the bias vanishes.

\medskip

\textbf{Conservative Variance Estimator.}
Theorem~\ref{thm:var_est} shows that the sample-centered covariance estimator in~\eqref{eq:var_est} differs from the true
variance by a centering term $Bias_n$ involving the heterogeneous score means. This term is
generally not identifiable from a single randomized experiment because it depends on
$E\psi_{i,q}$ for each unit. Moreover, unlike the case without interference, $Bias_n$ is not guaranteed to be non-negative. The plug-in estimator therefore need not be asymptotically conservative. We address this problem by bounding the unidentified centering component.

Let $\mathcal P_n=\{(i,j)\in N_n^2:\ell_A(i,j)\le 2\}$ be the set of two-step neighboring pairs, and let $H_n=(h_{ij})_{i,j\in N_n}$ be the two-step dependence matrix, where $h_{ij}=I\{(i,j)\in\mathcal P_n\}.$ Write
\[
P_n=I_n-\frac1n\mathbf 1\mathbf 1^\top,
\qquad
C_n=P_nH_nP_n,
\]
and let
\[
C_n=U_n\operatorname{diag}(\lambda_{n1},\ldots,\lambda_{nn})U_n^\top
\]
be an eigendecomposition. We isolate the negative spectral component by setting every positive eigenvalue equal to zero:
\begin{equation}
\label{eq:negative_spectral_part}
C_n^{-}
=
U_n\operatorname{diag}\{\min(\lambda_{n1},0),\ldots,
\min(\lambda_{nn},0)\}U_n^\top.
\end{equation}
Thus $C_n^{-}$ is negative semidefinite. It is convenient to write its nonnegative magnitude and the positive part of $C_n$ as
\[
K_n^{-}=-C_n^{-}\succeq0,
\qquad
C_n^{+}=C_n+K_n^{-}
=U_n\operatorname{diag}\{\max(\lambda_{n1},0),\ldots,
\max(\lambda_{nn},0)\}U_n^\top\succeq0.
\]
This eigentruncation is analogous to the modified network-HAC construction of
\citet{gao2025causal}, but here it is applied to the centered quantile-score
dependence matrix.

For $\widehat\psi=(\widehat\psi_{1,q},\ldots,
\widehat\psi_{n,q})^\top$, define the centered score vector
$\widehat v=P_n\widehat\psi$. Since
$\widehat\sigma_n^2=n^{-1}\widehat v^\top C_n\widehat v$, define
\begin{equation}
\label{eq:negative_spectral_adjustment}
\widehat B_{-,n}
=
\frac1n\widehat v^\top K_n^{-}\widehat v
=
-\frac1n\widehat v^\top C_n^{-}\widehat v
\ge0.
\end{equation}
The proposed conservative variance estimator is
\begin{equation}
\label{eq:plug_in_cons_var}
\widehat\sigma^2_{n,\mathrm{cons}}
=
\widehat\sigma_n^2+\widehat B_{-,n}
=
\frac1n\widehat v^\top C_n^{+}\widehat v.
\end{equation}
Consequently, the estimator itself is nonnegative. The adjustment is exactly
zero when $C_n$ is positive semidefinite and otherwise charges each negative
eigendirection according to its own eigenvalue, rather than charging every
direction by the magnitude of the most negative eigenvalue.

The corresponding conservative $(1-\alpha)$ confidence interval is then constructed by
 \[ \mathrm{CI}_{1-\alpha}^{\mathrm{cons}}(g,g') = \Big[ \widehat\tau_{nq}(g,g') - z_{1-\alpha/2} \widehat\sigma_{n,\mathrm{cons}}/\sqrt{n}, \quad \widehat\tau_{nq}(g,g') + z_{1-\alpha/2} \widehat\sigma_{n,\mathrm{cons}}/\sqrt{n}\Big]. 
 \]

The asymptotic conservativeness on the adjusted variance estimator is shown as follows. Let
\begin{equation}
\label{eq:negative_part_row_norm}
\rho_n^{-}
=
\|K_n^{-}\|_{\infty}
:=
\max_{i\in N_n}\sum_{j\in N_n}|(K_n^{-})_{ij}|,
\end{equation}
which controls the possibly dense matrix created by
eigenvalue truncation. The following theorem proves the asymptotic conservativeness of the adjusted varaince estimator under spectral condition.

\begin{theorem}
\label{thm:cons_var}
    Suppose the assumptions of Theorem~\ref{thm:var_est} hold and
$\inf_n\sigma_n^2>0$. Let $\epsilon_n$ be the rate sequence in
Assumption~\ref{asp:error_pdf}, and suppose that the following condition holds:
\begin{equation}
\rho_n^{-}
\left\{
\epsilon_n+n^{-1/2}
+
\left(\frac{1+d_{mx}^{(2)}}{n}\right)^{1/2}
\right\}
=o(1).
\label{eq:cons_rate_condition}
\end{equation}
Then
\[
    \sigma_n^2
    \le\widehat\sigma^2_{n,\mathrm{cons}}+o_p(1),
\] which implies the coverage of the confidence interval $$
\liminf_{n\to\infty}\Pr(\tau_{nq}(g,g')\in \mathrm{CI}_{1-\alpha}^{\mathrm{cons}}(g,g'))\ge 1-\alpha.
$$
\end{theorem}

\begin{algorithm}[ht]
  \caption{Quantile Exposure Effect Estimation}
  \label{algo:main}

  \KwIn{Outcomes $\bm{Y}$, treatments $\bm{W}$,
    covariates $\bm{X}$, adjacency matrix $A$,
    sample size $n$, quantile level $q$,
    exposure mapping $\Phi_n$, exposure levels $g,g'$,
    and density-estimation bandwidth $h$.}

  \KwOut{Effect estimator $\widehat{\tau}$ and selected
    variance estimator
    $\widehat{\sigma}_{n,\mathrm{sel}}^2$.}

  Compute exposure levels
  $G_i=\Phi_n(i,\bm{W},A)$ for $i=1,\ldots,n$\;

  For each $s\in\{g,g'\}$, compute
  $\pi_i(s)=\Pr(G_i=s)$ and estimate
  $\widehat{\beta}_{nq}(s)$
  using~\eqref{eq:minimization}\;

  Compute
  $\widehat{\tau}
    =\widehat{\beta}_{nq}(g)
     -\widehat{\beta}_{nq}(g')$\;

  For each $s\in\{g,g'\}$, estimate $\widehat{f}_{ns}$
  with bandwidth $h$ and a kernel satisfying
  Proposition~\ref{prop:kernel_pdf_main}, such as
  the Epanechnikov kernel
  $K_1(t)=\frac{3}{4}(1-t^2)I(|t|\leq 1)$\;

  For each $s\in\{g,g'\}$, compute the influence scores
  \[
    \widehat{\psi}_{i,q}(Y_i,s)
    =
    \frac{
      I(G_i=s)
      \bigl\{q-I(Y_i\leq\widehat{\beta}_{nq}(s))\bigr\}
    }{
      \pi_i(s)\,
      \widehat{f}_{ns}\bigl(\widehat{\beta}_{nq}(s)\bigr)
    },
    \qquad i=1,\ldots,n,
  \]
  and form the contrast score
  $\widehat{\psi}_{i,q}(Y_i,g)
   -\widehat{\psi}_{i,q}(Y_i,g')$\;

  Compute $\widehat{\sigma}_{n,\mathrm{sel}}^2$
  using either the plug-in estimator
  in~\eqref{eq:var_est} or the conservative estimator
  in~\eqref{eq:plug_in_cons_var},
  as selected by the analyst\;

  \Return{$\widehat{\tau}$,
    $\widehat{\sigma}_{n,\mathrm{sel}}^2$}\;
\end{algorithm}

\section{Simulation}
\label{sec:simu}
In the first simulation, we evaluate the performance of our QEE estimator and inference in random graph without community structure. At the beginning of
the experiment, we generate a connected degree-four random regular graph with
$n=1000$ units such that each unit has exactly four neighbors. Then with a total number of 500 replications, in every round of replication, we independently draw the covariates $X_i\sim N(0,I_p)$ with $p=6$ and assign treatments independently as
$W_i\sim\operatorname{Bernoulli}(1/2)$.
Let $N_i=\sum_{j\ne i}A_{ij}=4$ and define
\[
 G_{i1}=W_i,\qquad
 G_{i2}=\frac1{N_i}\sum_{j\ne i}A_{ij}W_j.
\]
Thus $G_{i2}$ is the treated proportion in unit $i$'s neighborhood. Since each unit has four neighbors, the $G_{i2}\in\mathcal{G}_2:=\{0,1/4,1/2,3/4,1\}$ in any replication. The outcomes
are generated by the following linear model
\[
 Y_i(\bm W)=X_i^\top\boldsymbol 1_p+5G_{i1}+25G_{i2}+\epsilon_i,
\]
where $\epsilon_i$ is an independent error. To compare our method under different tail distributions, we consider three types of error distributions, Gaussian, Student-$t$, and Slash errors, all standardized to
have variance one. Specifically, $\epsilon_i=Z_i$,
$\epsilon_i=T_{4i}/\sqrt{2}$, and
$\epsilon_i=Z_i/(\sqrt{5}U_i^{2/5})$, where
$Z_i\sim N(0,1)$, $T_{4i}\sim t_4$, and $U_i\sim\operatorname{Uniform}(0,1)$
are independent.
We evaluate two exposure mappings on these same outcomes.  First, we consider a
two-dimensional mapping $(G_{i1},G_{i2})$ containing all treatment-assignment
information entering the outcome. For fixed quantile $q\in (0,1),$ we evaluate the Quantile Direct Effect (QDE) under own-treatment contrast at neighborhood
exposure $t=1/2$ defined as
\[
\mathrm{QDE}_{1/2}
 =\tau_{nq}\bigl((1,1/2),(0,1/2)\bigr).
\]
From the linear model of the potential outcome, the true value of $\mathrm{QDE}_{1/2}(q)$ equals to $5$ at every quantile. From the independent Bernoulli trials, the generalized
propensity score is $\pi_i(w,1/2)=3/16.$ Second, consider a scaler own-treatment exposure mapping $G_i=W_i$ and consider the following quantile exposure effect defined by its contrast
\[
 \mathrm{QEE}=\tau_{nq}(1,0).
\]
This QEE can be viewed as an integration of the exposure-specific quantile direct effects over the neighboring assignments conditional on $W_i$.  Because $W_i$ is independent of its neighbors' treatments and the true value of $\mathrm{QDE}_t=5$ for every $t\in\mathcal{G}_2$, hence the true value of $\mathrm{QEE}$ also equals to $5$ at every quantile. 

For both proposed contrasts, we use the exact generalized propensity score,
the Epanechnikov density estimator, and the plug-in and conservative two-step
network variance estimators proposed in (\ref{eq:var_est}) and (\ref{eq:plug_in_cons_var}) for inference.  To evaluate the necessity of taking account of interference in analyzing quantile effects, we compare our method with the inference based on previous QTE estimator without taking account of interference. We use the semiparametric IPW estimator of \citet{firpo2007efficient}, including its
propensity-score first-step adjustment and conventional i.i.d. variance
formula.  We evaluate nominal $90\%$ intervals at
$q\in\{0.2,0.5,0.8\}$. 

\begin{table}[ht]
\centering
\caption{Coverage and average interval width for the quantile effects in random regular graph.}
\label{tab:random_graph_linear}
\setlength{\tabcolsep}{2.2pt}
\renewcommand{\arraystretch}{1.02}
\scriptsize
\begin{tabular}{llrrrrrrrrrr}
\toprule
& & \multicolumn{4}{c}{$\mathrm{QDE}_{1/2}$}
& \multicolumn{4}{c}{$\mathrm{QEE}(1,0)$}
& \multicolumn{2}{c}{QTE} \\
\cmidrule(lr){3-6} \cmidrule(lr){7-10} \cmidrule(l){11-12}
& & \multicolumn{2}{c}{Plug-in} & \multicolumn{2}{c}{Conservative}
& \multicolumn{2}{c}{Plug-in} & \multicolumn{2}{c}{Conservative}
& \multicolumn{2}{c}{\citet{firpo2007efficient}} \\
\cmidrule(lr){3-4} \cmidrule(lr){5-6} \cmidrule(lr){7-8}
\cmidrule(lr){9-10} \cmidrule(l){11-12}
Distribution & $q$ & Coverage & Width & Coverage & Width
& Coverage & Width & Coverage & Width & Coverage & Width \\
\midrule
\multirow{3}{*}{Gaussian} & .2
& 90.6 & 1.303 & 98.6 & 1.897
& 91.0 & 2.503 & 98.0 & 3.309 & 86.6 & 2.008 \\
& .5 & 91.4 & 1.170 & 98.6 & 1.707
& 93.8 & 2.403 & 98.6 & 3.099 & 85.6 & 1.769 \\
& .8 & 90.2 & 1.279 & 97.8 & 1.866
& 90.2 & 2.490 & 98.2 & 3.291 & 83.4 & 1.993 \\
\addlinespace
\multirow{3}{*}{Student-t} & .2
& 88.8 & 1.293 & 98.0 & 1.879
& 91.6 & 2.501 & 97.6 & 3.307 & 83.2 & 2.009 \\
& .5 & 91.4 & 1.164 & 99.0 & 1.696
& 92.6 & 2.378 & 97.8 & 3.080 & 82.0 & 1.763 \\
& .8 & 91.2 & 1.282 & 99.4 & 1.867
& 89.2 & 2.482 & 97.4 & 3.290 & 82.0 & 1.995 \\
\addlinespace
\multirow{3}{*}{Slash} & .2
& 90.4 & 1.281 & 99.0 & 1.856
& 91.8 & 2.508 & 97.8 & 3.318 & 82.4 & 2.009 \\
& .5 & 88.6 & 1.147 & 98.6 & 1.669
& 94.0 & 2.388 & 98.6 & 3.085 & 84.0 & 1.758 \\
& .8 & 88.0 & 1.264 & 97.6 & 1.838
& 90.8 & 2.489 & 98.2 & 3.292 & 85.4 & 1.990 \\
\addlinespace[2pt]
\bottomrule
\end{tabular}
\end{table}

Table~\ref{tab:random_graph_linear} reports coverage and average interval width
for the exposure specific direct effect $\mathrm{QDE}_{1/2}$ and the own-treatment QEE, together with the
conventional QTE without interference.  We make the following
important observations from the results.  First, QTE undercovers in all nine
distribution--quantile combinations, with coverage between $82.0\%$ and
$86.6\%$.  In contrast, plug-in coverage ranges from $88.0\%$ to $91.4\%$ for
$\mathrm{QDE}_{1/2}$ and from $89.2\%$ to $94.0\%$ for QEE. Their mean absolute
coverage errors are $1.09\%$ and $1.84\%$ respectively, compared
with $6.16\%$ for QTE. It implies that the own-treatment QEE remains
substantially better calibrated even after the neighboring assignments are
integrated out and the treated and control units are more balanced as the propensity arises from $3/16$ to $1/2$.

Second, the efficiency gain reflects the amount of exposure information retained in considering interference. With the neighborhood exposure conditioning on $G_{i2}=1/2$ removes the large spillover component from the
within-exposure outcome variation, the plug-in $\mathrm{QDE}_{1/2}$
intervals remain $33.9\%$ to $36.5\%$ shorter than the QTE intervals with higher coverage. By contrast, the integrated QEE retains variation from the spillover component, and its plug-in intervals are therefore $24.4\%$ to $35.9\%$ wider than the QTE intervals without taking account interference. However, the latter's smaller width occurs together with systematic
undercoverage and hence does not provide evidence of valid precision. The conservative intervals provide the intended coverage protection, with
coverage of $97.6\%$ to $99.4\%$ for $\mathrm{QDE}_{1/2}$ and
$97.4\%$ to $98.6\%$ for QEE. Relative to the corresponding plug-in intervals,
their widths are $44.9\%$ to $45.9\%$ greater for $\mathrm{QDE}_{1/2}$ and $29.0\%$ to $32.5\%$ greater for QEE. These patterns remain stable across all distributions. Our method also provide accurate point estimation. Across the nine
distribution--quantile cases, the largest absolute bias is $0.025$ for
$\mathrm{QDE}_{1/2}$ and $0.091$ for QEE. Bias and empirical SD are reported in Supplementary
Table~\suppref{tab:random_graph_linear_diagnostics}. We also analyze the performance on quantile spillover effect, and provide an additional simulation under 
nonlinear treatment--exposure interaction in community-network, with similar observation from the results as in the first simulation. See Appendix~\suppref{sec:nonlinear_community} for details.

We also compare the variance estimators at different quantiles and network density. In this simulation, all $n=1000$ units are connected in a ring, and each of them is linked with $d/2$ nearest neighbors on either side, which forms a degree-$d$ ring with
$d\in\{2,4,6,8\}.$ Such graph mimics the social network in which each individual is connected to a small pack of neighbors and the network itself forms a large community. Throughout 500 replications, we compare the QEE plug-in and conservative variance estimates with the empirical sampling
variance of the QEE estimator and the variance estimate reported from previous QTE inference method without interference.  

\begin{figure}[ht]
\centering
\begin{minipage}[t]{0.485\linewidth}
\centering
\includegraphics[width=\linewidth]{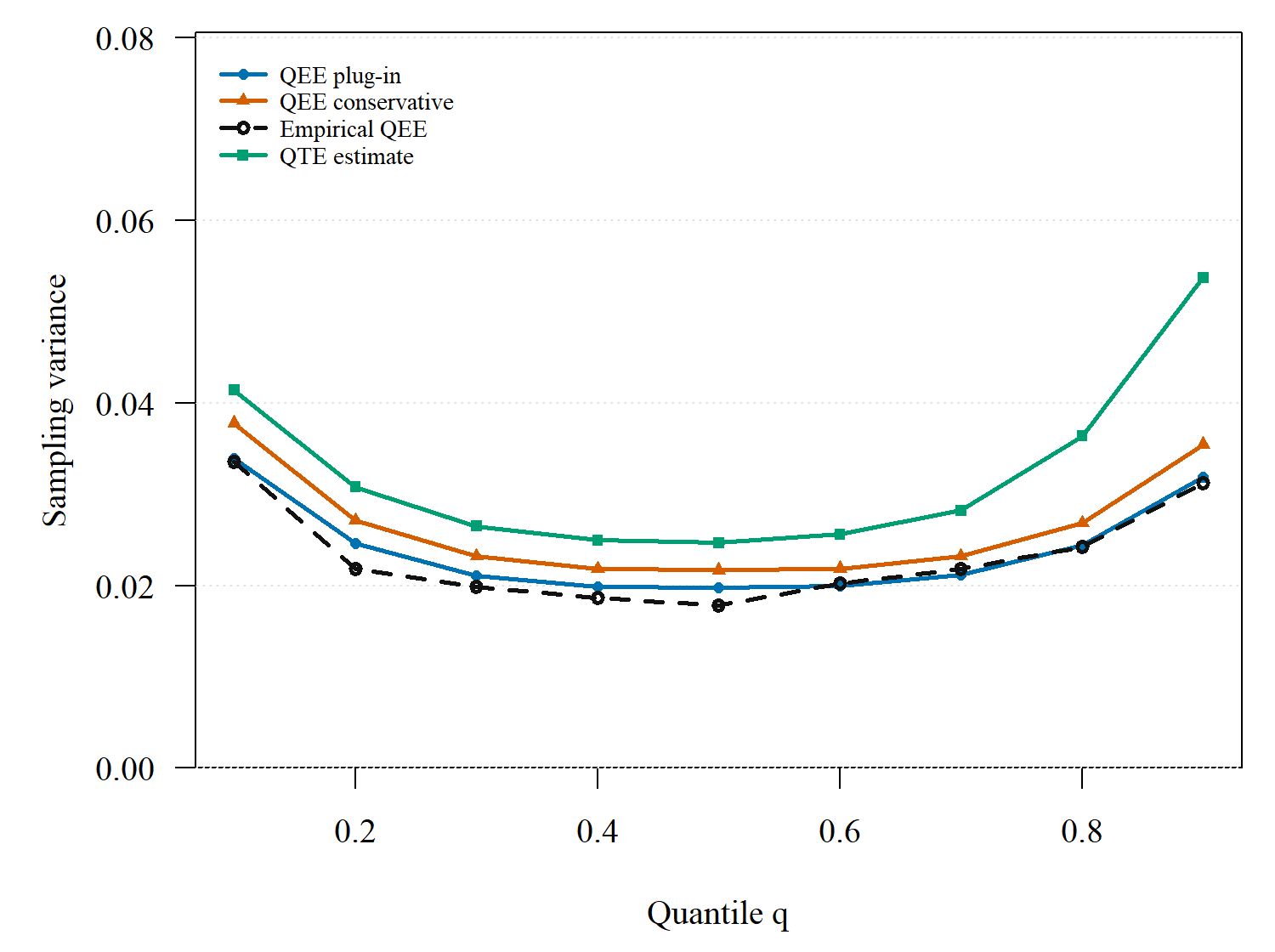}
\end{minipage}\hfill
\begin{minipage}[t]{0.485\linewidth}
\centering
\includegraphics[width=\linewidth]{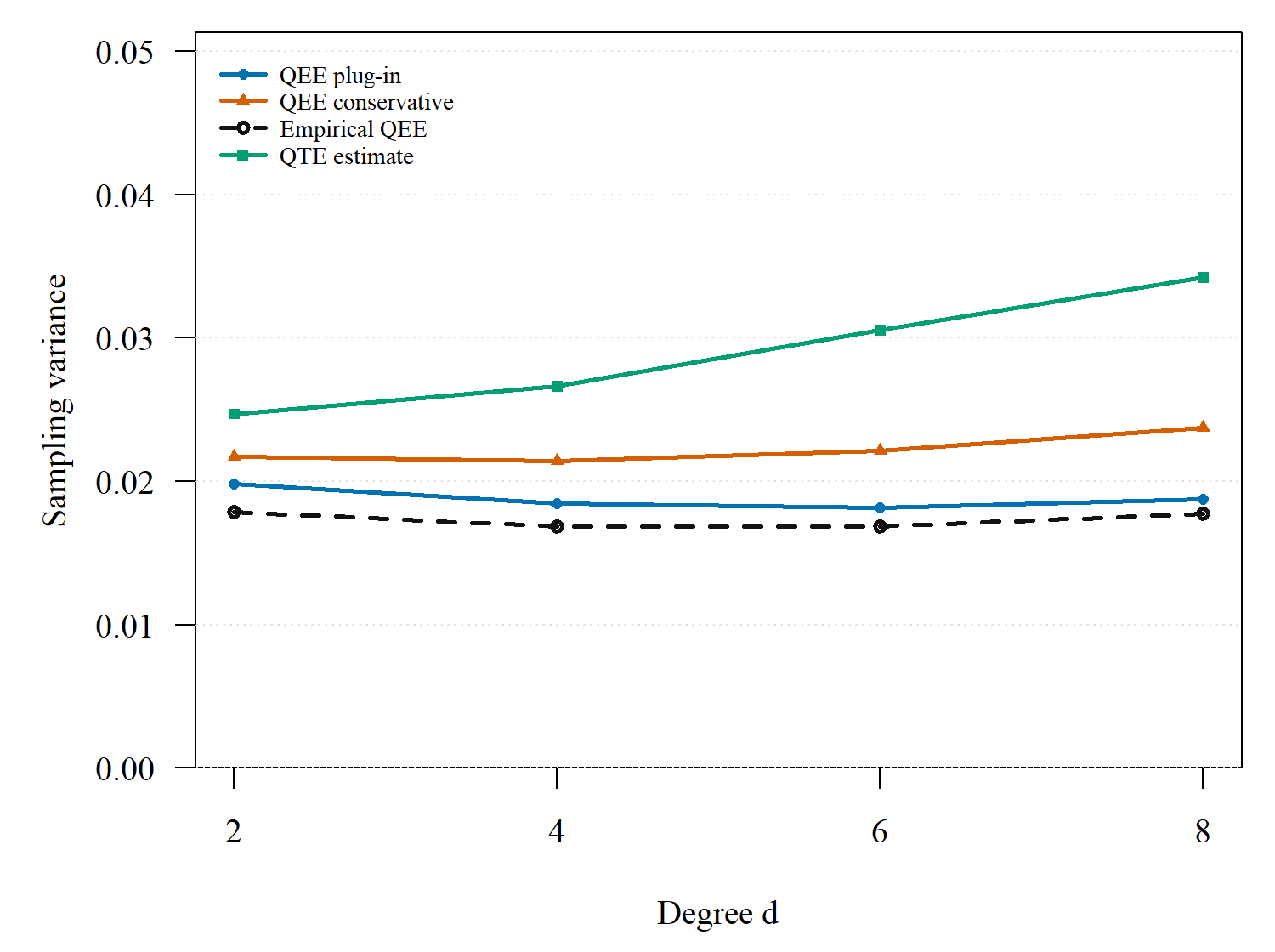}
\end{minipage}
\caption{Variance comparison in non-community design.  The left
panel varies $q$ from $0.1$ to $0.9$ at degree $d=2$; the right panel varies
$d\in\{2,4,6,8\}$ at $q=0.5$.}
\label{fig:noncommunity_variance_comparison}
\end{figure}

The result of variance comparison is shown in Figure~\ref{fig:noncommunity_variance_comparison}, with the complete figures given in Appendix~\suppref{sec:noncommunity_variance_comparison}. The plug-in curve closely tracks the empirical QEE variance throughout the
quantile range, while the conservative curve remains above it. At the median, the plug-in estimate is comparatively stable across degrees, while the conservative variance estimate rises slightly. Meanwhile, the QTE variance estimate lies above both QEE variance estimates across varying quantiles and degrees, and becomes progressively larger more rapidly than QEE variance estimates as the degree increases. This suggests the stability of our variance estimates under increasing degree size by adding the correlation from neighboring terms with distance $\ell_A(i,j)\le 2$ into variance estimation.

\section{Real Data Example}
\label{sec:exp}
We evaluate the performance of our method in a real-world dataset originated from a classical network experiment studied in previous works on average treatment effects with interference~\citep{paluck2016changing,leung2022causal,gao2025causal}. The experiment studies the effect of attending a series of anti-conflict course to the social norm of students at school. The experiment enrolls 56 schools with 28 of them assigned as treated schools in block randomization. Then in the second stage, a group of eligible students in the treated schools are selected based on covariates, and half of the eligible students in each treated school are assigned with treatments. Treated students are subsequently invited to attend a series of anti-conflict courses to help them identify social conflicts at school as well as provide suggestions on how to handle such conflicts.

While previous studies focus on wristband-wearing status $Y_{\text{wr}}$ as outcome, which indicates whether the students wear a wristband as a reward to those observed engaging in anti-conflict behaviors, this outcome is binary and thus not suitable to interpret the quantile treatment effects. Meanwhile, the dataset also collects the records on whether the student has been nominated as violating any of the 51 different behaviors that are against the social norm, such as bullying, spreading rumor, etc. Let $S_{\text{anti}}$ be the amount of records each individual has in terms of the anti-norm behaviors above, ranging from a minimum of 0 record to a maximum of 27 records, and let $Y_{\text{anti}}=S_{\text{anti}}/51$ be its normalized score with regard to the amount of behaviors. To ensure the outcome representing the overall performance of students in terms of abiding by social norms while is suitable for interpreting quantile effects, instead of wristband-wearing status, in this paper, we focus on a social-norm summary score $Y_{\text{sn}}=Y_{\text{wr}}-Y_{\text{anti}}.$ Students wearing wristband as a reward to anti-conflict behaviors or have less records on anti-norm behaviors result in a larger summary score $Y_{\text{sn}}$, therefore, the score evaluates the comprehensive performance of students in terms of avoiding conflicts as well as anti-norm behaviors.

To evaluate the overall impact of the education program on the social-norm score, in the first experiment, we focus on the treated schools and eligible students that receives either treatment or control in the experiment. This results in an amount of $n=1456$ students, with half of them assigned with treatment and therefore required to attend the educational courses. In the survey, each student lists at most 10 closest friends. We pick out the eligible students from the friend lists, and connect a pair of students if either of them is listed by the other as close friend. This forms a friend network with adjacency matrix $A$ among all eligible students. The average degree is 1.492, with the largest degree to be 8. Similar with~\cite{leung2022causal,gao2025causal}, we consider the exposure mapping $G_i=I(\sum_{j}A_{ij}W_j>0),$ and therefore $G_i=1$ if either the student or one of his / her friends receives treatment. The quantile exposure effect of interest is $\tau_q(1,0)$ with $q\in (0,1)$ as defined in (\ref{eq:qte_def}). We refer such effect as overall quantile exposure effect.

\begin{figure}
    \centering
    \includegraphics[width=0.3\linewidth]{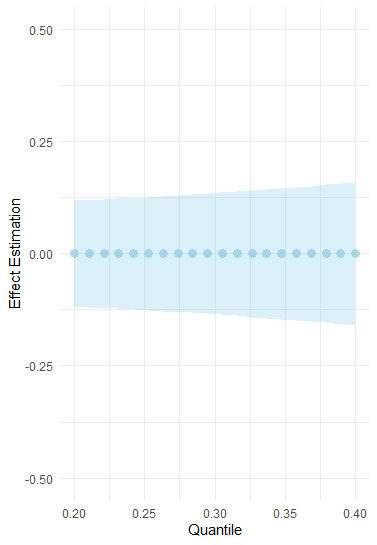}
    \includegraphics[width=0.3\linewidth]{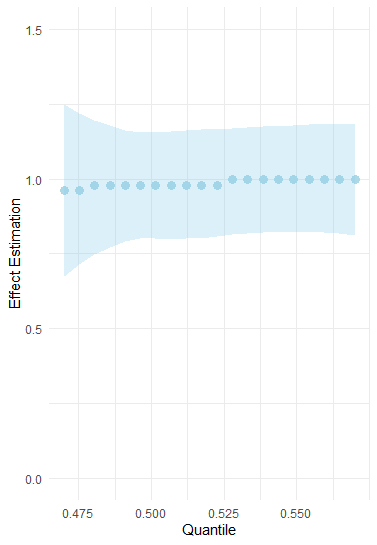}
    \includegraphics[width=0.3\linewidth]{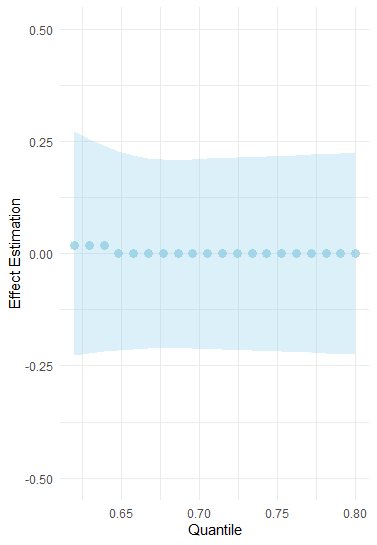}
    \caption{Overall quantile exposure effects $\hat{\tau}_q(1,0)$ at low (left), medium (middle) or high (right) quantiles.}
    \label{fig:experiment}
\end{figure}

We estimate the overall effect as well as its variance at different quantiles following the settings described above. The result is shown in Figure~\ref{fig:experiment}, featuring the key observations as follows. First, the overall quantile exposure effects vary across different levels of quantiles. At low and high quantile, the quantile exposure effects approximates zero while a positive effect is observed for students at medium quantile level. This implies that the education course has a significant effect on improving the social norm for students with a medium level of social-norm obedience, while the effect of the course is less evident for either more courteous or more brutal students comparing to the average level. This aligns with the conclusion in previous studies that the courses helps improve the social norm of students, as the average exposure effect is $0.143$ following the estimator in~\citet{leung2022causal}. Comparing to the average effect, the quantile exposure effects provide us further implication that the improvement on social-norm score mainly comes from the students with medium level of social norm, while the course may be less effective for the students who already have a good social manner or with an extremely bad social manner. Second, the result showcases a relatively low standard deviation in most cases. The standard deviations at quantile $q=0.3,0.5,0.8$ are $0.14,0.19,0.22$ respectively. Meanwhile, we also observe a jump in quantile exposure effect estimation at about $q=0.45$ and $q=0.6$ due to the discontinuity in the empirical distribution of outcome, which results in a larger variance estimation at this jumping quantiles.

\begin{table}[H]
    \centering
    \tbl{Quantile exposure effect in treated and spillover populations}
\begin{tabular}{lcccc}
    \toprule
      & \multicolumn{2}{c}{Treated} & \multicolumn{2}{c}{Spillover} \\
      Effect & Est. & SD & Est. & SD \\
      \midrule
      QEE - low  & 0.01 & 0.10 & 0.00 & 0.10 \\
      QEE - medium & 0.98 & 0.12 & 0.98 & 0.28\\
      QEE - high & 0.01 & 0.13 & 0.01 & 0.22 \\
      AEE  & 0.14 & 0.04 & 0.06 & 0.02 \\
      \bottomrule
    \end{tabular}
    \begin{tabnote}
     QEE - quantile: Quantile exposure effect at different levels of quantiles, AEE: Average exposure effect in~\citet{leung2022causal}, sd: estimated standard deviation. The low / medium / high quantiles intreated and spillover population are set as 0.2 / 0.5 / 0.8 and 0.3 / 0.6 / 0.9 respectively. 
    \end{tabnote}
    \label{tab:exp_res}
\end{table}

Next, we take a few extra settings to evaluate the direct and spillover effects on anti-social norm in quantiles. For direct quantile exposure effect, let exposure mapping $G_i=W_i$, and following the previous setting we choose the students in treated population who are eligible to treatment in assessment. Similar as the overall effect, this results in a set of $n=1456$ students with half of them receiving courses as treatment. For spillover quantile exposure effect, we focus on $G_i=I(\sum_{j\neq i}A_{ij}W_j>0),$ therefore, $G_i=1$ if one of student $i$'s friends receives treatment. To ensure the overlap assumption holds, following~\citet{leung2022causal}, we choose the students in spillover population such that at least one of his / her friends is eligible to the treatment. This results in a set of $n=2467$ students, with $G_i=1$ among 1439 students. The quantile exposure effects on treated and spillover populations are both $\tau_q(1,0)$ with corresponding exposure mappings. Table~\ref{tab:exp_res} indicates that, the quantile exposure effects on both treated and spillover populations are positive at medium quantiles, while approximate zero at low and high quantiles. Besides, treated population has a wider range of quantiles (length = 0.211) with positive effect than spillover population (length = 0.109). This is aligned with the findings in~\citet{leung2022causal,gao2025causal} that the course improves social norm in the average sense in both populations, and the effect is stronger on treated population than on spillover population. Moreover, quantile exposure effect provides extra information that the improvement of social norms from educational courses mainly comes from students with medium level of social manner.

\section{Discussion}

This work moves beyond average treatment effects by developing a general framework for defining, estimating, and conducting inference on quantile exposure effects in the presence of network interference. The proposed methodology is supported by a rigorous asymptotic theory and complemented by extensive simulation evidence and an empirical application, demonstrating that distributional impacts, especially in the tails, can be obscured when attention is restricted to mean effects.

Although our analysis centers on quantile exposure effects, the framework is not inherently tied to quantiles. The same exposure-mapping and generalized propensity score ideas can be used to define and study other distributional causal estimands under interference---for example, effects on distribution functions, tail probabilities, inequality measures, or other functionals of the outcome distribution. Developing estimation and inference procedures for these broader targets, and clarifying which classes of functionals admit tractable limit theory under network dependence, is a natural direction for future work.

Our theoretical results are established for a broad class of non-clustered networks under transparent conditions that control network complexity through restrictions on node-level dependence and degree growth. These conditions cover many empirically relevant social and economic networks, but they primarily correspond to sparse regimes, where neighborhood size grows at most polynomially in the population size. Understanding how the proposed estimators behave in dense networks, or under alternative structural features such as strong clustering, community structure, or heavy-tailed degree distributions, remains an important open issue. Extending the theory to accommodate such settings or developing robust procedures that remain reliable across a wider range of network architectures would further broaden the scope and applicability of quantile-based causal inference under interference.
\section*{Acknowledgments}
The authors thank their collaborators for offering constructive suggestions that helped improve this manuscript.

\bibliographystyle{plainnat}
\bibliography{paper-ref}

\clearpage
\begin{center}
{\Large\bfseries Supplementary Material on}\\[0.6em]
{\large\bfseries ``Design-based Estimation and Inference on Quantile Exposure Effect under General Interference''}
\end{center}
\vspace{1em}

\appendix
\counterwithin{equation}{section}
\counterwithin{figure}{section}
\counterwithin{table}{section}
In Appendix~\ref{apd:thm}, we provide the details on the proves to the propositions and theorems in the main paper. In Appendix~\ref{sec:supplementary_simulations}, we provide additional details on experimental settings and supplementary results.

\section{Theoretical Proofs}
\label{apd:thm}

\begin{proof}[Proof of Proposition~\ref*{prop:CND}]
For \(S\subseteq N_n\), write
\[
    \mathcal{M}_S
    =
    \bigvee_{j\in S}\mathcal{M}_j
    =
    \sigma(W_j:j\in S),
\]
and let
\[
    \mathcal{F}_S^Z
    =
    \sigma(Z_i:i\in S).
\]

Because the network and the potential-outcome schedule are fixed,
Assumption~3.1 implies that, for every \(i\in N_n\), there exists a
nonrandom measurable function \(\phi_{i,n}\) such that
\[
    G_i
    =
    \phi_{i,n}\!\left(W_{\bar{\nu}_n(i)}\right).
\]
Assumption~\ref*{asp:local-potential-outcomes} similarly implies that the outcome can be expressed in the form
\[
    Y_i
    =
    y_{i,n}\!\left(W_{\bar{\nu}_n(i)}\right).
\]
It follows that
\begin{equation}
    Z_i=(Y_i,G_i)
    \quad\text{is }
    \mathcal{M}_{\bar{\nu}_n(i)}\text{-measurable}.
    \label{eq:Z_measurable}
\end{equation}

Fix \(N_n'\subseteq N_n\), and let \(A,B\subseteq N_n'\) satisfy
the CND separation conditions
\begin{equation}
\label{eq:mutual_inclusion}
    A\subseteq N_n'\setminus\bar{\nu}_n(B),
    \qquad
    B\subseteq N_n'\setminus\bar{\nu}_n(A).
\end{equation}
Set
\[
    C=\nu_n(A),
    \qquad
    D=\bar{\nu}_n(B)\setminus C.
\]
Since \(C=\bar{\nu}_n(A)\setminus A\), we have \(A\cap C=\varnothing\).
Moreover, Eq. \eqref{eq:mutual_inclusion} implies
\[
    A\cap\bar{\nu}_n(B)=\varnothing,
\]
and hence \(A\cap D=\varnothing\). By construction,
\(C\cap D=\varnothing\). Therefore, \(A\), \(C\), and \(D\) are
pairwise disjoint.

By Eq. \eqref{eq:Z_measurable},
\[
    \mathcal{F}_A^Z\vee\mathcal{M}_A
    \subseteq
    \mathcal{M}_{\bar{\nu}_n(A)}
    =
    \mathcal{M}_{A\cup C}.
\]
Similarly,
\[
    \mathcal{F}_B^Z\vee\mathcal{M}_B
    \subseteq
    \mathcal{M}_{\bar{\nu}_n(B)}
    \subseteq
    \mathcal{M}_{C\cup D}.
\]

Because the treatment assignments are mutually independent and
\(A\), \(C\), and \(D\) are pairwise disjoint, the random vectors
\(W_A\), \(W_C\), and \(W_D\) are mutually independent. Consequently,
\[
    \mathcal{M}_{A\cup C}
    \ \perp\!\!\!\perp\
    \mathcal{M}_{C\cup D}
    \,\big|\,
    \mathcal{M}_C.
\]
Indeed, conditional on \(W_C\), every random variable measurable
with respect to \(\mathcal{M}_{A\cup C}\) is a measurable function
of \((W_A,W_C)\), whereas every random variable measurable with
respect to \(\mathcal{M}_{C\cup D}\) is a measurable function of
\((W_C,W_D)\); the desired conditional independence follows from
the independence of \(W_A\) and \(W_D\).

Combining the results above, and using the fact that conditional
independence is preserved when the two sigma-fields are replaced
by sub-sigma-fields, gives
\[
    \mathcal{F}_A^Z\vee\mathcal{M}_A
    \ \perp\!\!\!\perp\
    \mathcal{F}_B^Z\vee\mathcal{M}_B
    \,\big|\,
    \mathcal{M}_{\nu_n(A)}.
\]
This is precisely the CND property for
\(\{Z_i\}_{i\in N_n}\) with respect to
\((\nu_n,\mathcal{M})\).

Since
\[
    \sigma(Y_i)\subseteq\sigma(Z_i)
\]
and, for every measurable \(h_i\),
\[
    \sigma\!\left(h_i(Y_i,G_i)\right)
    \subseteq
    \sigma(Z_i),
\]
the conclusions for \(\{Y_i\}\) and
\(\{h_i(Y_i,G_i)\}\) follow immediately.
\end{proof}

\begin{proof}[Proof on Corollary~\ref{cor:two-step-dependency}]
By Assumptions 3.1 and 3.2, $S_i$ is measurable with respect to
$\sigma(W_{\bar\nu_n(i)})$. If $A$ and $B$ are nonadjacent in the
two-step graph, then
\[
\bar\nu_n(A)\cap\bar\nu_n(B)=\varnothing;
\]
otherwise, a common neighbor would generate a path of length at most two
between some $i\in A$ and $j\in B$. Hence $S_A$ and $S_B$ are measurable
with respect to disjoint collections of mutually independent treatment
assignments. The desired independence follows.
\end{proof}

\begin{proof}[Proof of Proposition~\ref*{prop:max-degree-two-step}]
Under Assumptions \ref{asp:overlap} and \ref{asp:pdf}, from the form of $\psi_{i,q}$, the contrast scores are uniformly bounded in that for some $C>0,$
\[
\max_{i\in N_n}
|\psi_{i,q}-E\psi_{i,q}|
\le C.
\]
The variance lower bound gives
$\sigma_n(\psi_{\cdot,q})\ge c\sqrt n$. Hence the moments on the contrast score satisfy
\[
\mu_p^p(\psi_{\cdot,q})=O(n^{-p/2}),
\qquad p=3,4.
\]
Therefore
\[
a_n^{(2)}(\psi_{\cdot,q})
=
O\left(\frac{1+(d_{\mathrm{mx}}^{(2)})^2}{\sqrt n}\right)=o(1).
\]
Uniform boundedness of the standardized scores also yields
\[
\mu_4^4(\psi_{\cdot,q})
\le
Cn^{-1/2}\mu_3^3(\psi_{\cdot,q}),
\]
and therefore
\[
b_n^{(2)}(\psi_{\cdot,q})
\le
C\frac{d_{\mathrm{mx}}^{(2)}}{\sqrt n}
a_n^{(2)}(\psi_{\cdot,q}).
\]
Consequently, since $d_{mx}^{(2)}=o(n^{1/4}),$
\[
\begin{aligned}
&|\log a_n^{(2)}(\psi_{\cdot,q})|
\{b_n^{(2)}(\psi_{\cdot,q})\}^{1/2}\\
\le&
C|\log a_n^{(2)}(\psi_{\cdot,q})|
\{a_n^{(2)}(\psi_{\cdot,q})\}^{1/2}
\left\{
\frac{d_{\mathrm{mx}}^{(2)}}{\sqrt n}
\right\}^{1/2}=o(1)\\
\end{aligned}
\]
as $x^{1/2}|\log x|\to0$ as $x\downarrow0$, and the proof on Proposition~\ref*{prop:max-degree-two-step} is done by combining the results above.
\end{proof}

Next, we propose the following proposition on the locality of distribution function without interference, which endorses the validity of Assumption~\ref*{asp:pdf} in the main paper.
\begin{proposition}
\label{prop:integral_conv_rate}
Suppose there is no interference, the exposure is own treatment $G_i=W_i$, and the potential-outcome vectors
$\{Y_i(0),Y_i(1)\}_{i\ge1}$ are i.i.d. draws from a superpopulation. For each
$g\in\{0,1\}$, suppose $Y_i(g)$ has a continuously differentiable CDF $F_g$
with a density that is strictly positive and locally Lipschitz on a
neighborhood of its $q$th quantile. Then, with respect to repeated
superpopulation sampling, the two conditions in Assumption~\ref*{asp:pdf}
hold in probability. In particular, for every fixed $M<\infty$,
\[
\sup_{|t|\le M/\sqrt n}
\left|
F_{ng}\left(\beta_{nq}(g)+t\right)
-
F_{ng}\left(\beta_{nq}(g)\right)
-
f_g\left(\beta_{nq}(g)\right)t
\right|
=
o_p(n^{-1/2}).
\]
\end{proposition}

\begin{proof}
We verify the components of Assumption~\ref*{asp:pdf} in the order
in which they appear in the main paper.
Under no interference and $G_i=W_i$, $F_{ng}$ is the empirical CDF based on
the i.i.d. sample $\{Y_i(g):i\in N_n\}$. Continuity implies that ties occur with probability
zero, so every empirical jump is $1/n=o(n^{-1/2})$. This verifies the quantile-jump condition in
Assumption~\ref*{asp:pdf}.

Let $b_g$ denote the population $q$th quantile. Uniform consistency of the
empirical quantile gives $\beta_{nq}(g)\to_p b_g$. Because $f_g$ is strictly
positive and locally Lipschitz around $b_g$, there exist constants
$0<c_f<C_f<\infty$, $L_f<\infty$, and $\eta_f>0$ such that, with probability
tending to one, the density lower and upper bounds and the local Lipschitz
condition in Assumption~\ref*{asp:pdf} hold uniformly around
$\beta_{nq}(g)$. On the same event, all
intervals with endpoints $\beta_{nq}(g)$ and
$\beta_{nq}(g)+t$, $|t|\le M/\sqrt n$, therefore lie in a fixed neighborhood
on which $f_g$ is bounded and Lipschitz.

The class of intervals in that neighborhood is a VC class. Its members having
length at most $M/\sqrt n$ have probability $O(n^{-1/2})$. Put
$\delta_n=Cn^{-1/2}$. The local maximal inequality under uniform
entropy in \citet[Theorem~2.1]{van2011local}, together with the
polynomial covering-number bound for VC classes, gives
\[
\sup_{I:\,P_g(I)\le\delta_n}
\left|(\mathbb P_n-P_g)\mathbf 1_I\right|
=O_p\left\{
\sqrt{\frac{\delta_n\log(1/\delta_n)}{n}}
+\frac{\log(1/\delta_n)}{n}
\right\}
=O_p(n^{-3/4}\sqrt{\log n}).
\]
Here $\mathbb P_n$ is the empirical measure of $\{Y_i(g):i\in N_n\}$, $P_g$ is their population law and $1_I$ is the indicator of interval $I$. Because the random interval between $\beta_{nq}(g)$ and
$\beta_{nq}(g)+t$ belongs to this class with probability tending to one, we
obtain
\[
\begin{aligned}
&\sup_{|t|\le M/\sqrt n}
\Big|[F_{ng}\{\beta_{nq}(g)+t\}-F_{ng}\{\beta_{nq}(g)\}]
-[F_g\{\beta_{nq}(g)+t\}-F_g\{\beta_{nq}(g)\}]\Big|\\
&\qquad=O_p(n^{-3/4}\sqrt{\log n})=o_p(n^{-1/2}).
\end{aligned}
\]
Finally, local Lipschitz continuity of $f_g$ gives, uniformly over the same
range of $t$,
\[
F_g\{\beta_{nq}(g)+t\}-F_g\{\beta_{nq}(g)\}
-f_g\{\beta_{nq}(g)\}t=O(t^2)=O(n^{-1}).
\]
Combining the stochastic increment bound with the deterministic
Taylor bound proves the uniform local linear expansion in
Assumption~\ref*{asp:pdf}. Together with the jump and density arguments above,
this verifies all components of that assumption in superpopulation
probability.
\end{proof}

Next, we provide proof details on the asymptotic normality of the proposed QTE estimator under interference. We start with introducing the following useful lemmas. Lemma~\ref{lem:CLT_in_networks} presents 
Berry-Esseen-type bound for sums of conditionally neighborhood
dependent (CND) arrays, which we use to establish asymptotic normality for the key pseudo influence-function terms.

\begin{lemma}[Central limit theorem on graph-dependent variables]
\label{lem:CLT_in_networks}
Let $\{M_i:i\in N_n\}$ be a triangular array satisfying
\[
EM_i=0,
\qquad
M_i\in\sigma(W_{\bar\nu_n(i)}),
\]
and let
\[
\sigma_n^2(M)
=
\operatorname{Var}\left(\sum_{i\in N_n}M_i\right)>0.
\]
Define
\[
\Delta_n(t)
=
\left|
\Pr\left\{
\frac{1}{\sigma_n(M)}
\sum_{i\in N_n}M_i\le t
\right\}
-\Phi(t)
\right|.
\]
On the event $a_n^{(2)}(M)\le1$, there exists an absolute constant
$C<\infty$ such that
\[
\sup_{t\in\mathbb R}\Delta_n(t)
\le
C\left[
\{a_n^{(2)}(M)\}^{1/2}
+
|\log a_n^{(2)}(M)|
\{b_n^{(2)}(M)\}^{1/2}
\right].
\]
Consequently, if
\[
a_n^{(2)}(M)\to0,
\qquad
|\log a_n^{(2)}(M)|
\{b_n^{(2)}(M)\}^{1/2}\to0,
\]
then
\[
\frac{\sum_{i\in N_n}M_i}{\sigma_n(M)}
\ \longrightarrow\ N(0,1).
\]
\end{lemma}

\begin{proof}[Proof of Lemma~\ref{lem:CLT_in_networks}]
If two index sets $A,B\subseteq N_n$ are nonadjacent in the two-step
graph, then
\[
    \overline\nu_n(A)\cap\overline\nu_n(B)=\varnothing .
\]
Because the treatment assignments are mutually independent and
$M_i\in\sigma(W_{\overline\nu_n(i)})$, it follows that
\[
    \sigma(M_i:i\in A)
    \quad\text{and}\quad
    \sigma(M_j:j\in B)
\]
are independent. Therefore, under the neighborhood system
$\nu_n^{(2)}$, the array is CND and Condition A
holds from Corollary \ref*{cor:two-step-dependency} taking all conditioning sigma-fields to be trivial. Since
$\mathbb E M_i=0$, the lemma follows directly the application of Corollary 3.1 in \citet{lee2019stable} with
neighborhood system $\nu_n^{(2)}$.
\end{proof}

The second lemma gives an important rate bound for sums of network-dependent variables, and
will be used repeatedly in subsequent derivations.

\begin{lemma}
\label{lem:CND_bound_rate}
Let $\{M_i\}_{i\in N_n}$ with $M_i=m(W_i,G_i,Y_i,\bm X)$ be the triangular array such that the CND property in Proposition~\ref*{prop:CND}. Then we have 
\[
\left|
\sum_{i\in N_n}(M_i-EM_i)
\right|
=
O_p\left[
\left\{
(1+d_{\mathrm{mx}}^{(2)})
\sum_{i\in N_n}\operatorname{Var}(M_i)
\right\}^{1/2}
\right].
\]
\end{lemma}
\begin{proof}[Proof of Lemma~\ref{lem:CND_bound_rate}]
    We observe that for each $i,j\in N_n$, $M_i$ and $M_j$ are independent whenever the
two-step closed neighborhoods are disjoint, i.e.,
$\bar\nu_n(i)\cap\bar\nu_n(j)=\varnothing$ (equivalently, $\ell_A(i,j)>2$). Therefore, $cov(M_i,M_j)=0$ for $\ell_A(i,j)>2$. Hence from Chebyshev's inequality, we have
\begin{equation*}
    \begin{aligned}
    \left|\sum_{i\in N_n} (M_i-EM_i)\right|&=O_p\left(\left\{\mathrm{var}\left(\sum_{i\in N_n}M_i\right)\right\}^{1/2}\right)\\
    &=O_p\left(\left\{\sum_{i\in N_n}\left(\mathrm{var}(M_i)+\sum_{j:j\neq i,\ell_A(i,j)\le 2}\mathrm{cov}(M_i,M_j)\right)\right\}^{1/2}\right)\\
    &\le O_p\left(\left\{\sum_{i\in N_n}\left(\mathrm{var}(M_i)+\sum_{j: j\neq i,\ell_A(i,j)\le 2}[\mathrm{var}(M_i)+\mathrm{var}(M_j)]/2\right)\right\}^{1/2}\right)\\
    &\le O_p\left(\left\{(1+d_{mx}^{(2)})\sum_{i\in N_n}\mathrm{var}(M_i)\right\}^{1/2}\right),\\
\end{aligned}
\end{equation*} which proves the lemma.
\end{proof}
\begin{corollary}
\label{cor:CND_rate}
    Under the setup of
    Lemma~\ref{lem:CND_bound_rate}, if $\operatorname{E}[M_i \mid \mathcal{M}_{\nu_n(i)}] = 0,$ a.e. for each $i \in N_n,$ then we have $$
    \sum_{i\in N_n} M_i\le O_p\left(\left\{(1+d_{mx})\sum_{i\in N_n}EM_i^2\right\}^{1/2}\right).
    $$
\end{corollary}

\begin{proof}[Proof of Corollary~\ref{cor:CND_rate}]
Fix $i\in N_n$. If $j\notin \nu_n(i)$, the CND property implies that $M_i$ and $M_j$ are
conditionally independent given $\mathcal M_{\nu_n(i)}$, so
\begin{equation*}
\begin{aligned}
    E\left[M_iM_j\right]=&E\left[E\left\{M_iM_j\mid \mathcal{M}_{\nu_n(i)}\right\}\right]=E\Big\{\underbrace{E\left(M_i\mid \mathcal{M}_{\nu_n(i)}\right)}_{=0}  E\left(M_j\mid \mathcal{M}_{\nu_n(i)}\right)\Big\}=0.\\
\end{aligned}
\end{equation*}
The diagonal term $j=i$ is bounded separately by $EM_i^2$. Repeating the
argument in Lemma~\ref{lem:CND_bound_rate} and applying Hölder's inequality,
the remaining covariance sum is restricted to $j\in\nu_n(i)$. Thus there are
at most $1+d_{mx}$ terms after including the diagonal, and
\[
\begin{aligned}
\sum_{i\in N_n} M_i
&=O_p\left[
\left\{\sum_{i\in N_n}\left(
EM_i^2+\sum_{j\in\nu_n(i)}|E(M_iM_j)|
\right)\right\}^{1/2}\right]\\
&\le O_p\left[
\left\{(1+d_{mx})\sum_{i\in N_n}EM_i^2\right\}^{1/2}\right].
\end{aligned}
\]
\end{proof}

Next, we present a result on the validity of Assumption~\ref*{asp:pdf}, where the limiting density function is estimated using the kernel approach. To start, we first state a regularity condition on $F_{ng}$ which guarantees second-order smoothness of $F_{ng}$ in a local neighborhood slightly larger than the one in Assumption~\ref*{asp:pdf}. 
\begin{definition}[Local and bandwidth-scale regularity]
\label{def:bw_reg}
Let $f_g$ be the limiting density function defined in Assumption~\ref*{asp:pdf}. Consider the bandwidth $h=h_n$ used in the kernel density estimator. We say that $F_{ng}$ is {\it bandwidth-scale smooth} in 
a neighborhood of $y\in\mathcal{Y}$, if for every fixed $M<\infty$,
\[
\sup_{|t|\le Mh}
\left|F_{ng}(y+t)-F_{ng}(y)
-f_g(y)t-\frac12 f_g'(y)t^2\right|
=O\left(h^3\right).
\]
Analogously, we say that $F_{ng}$ is {\it local-uniform bandwidth-scale smooth} at the target quantile if, for every fixed $L,M<\infty$,
\[
\sup_{s\in\mathcal I_{n,L}(g)}\sup_{|t|\le Mh}
\left|F_{ng}(s+t)-F_{ng}(s)
-f_g(s)t-\frac12 f_g'(s)t^2\right|
=O\left(h^3\right),
\]
where $\mathcal I_{n,L}(g)=\{s:|s-\beta_{nq}(g)|\le L/\sqrt n\}$.
\end{definition} 

Note that the  bandwidth-scale smoothness condition considers a neighborhood of size slightly larger than the $O(n^{-1/2})$ neighborhood used in Assumption 3. Let $Z_i^{g,h}=I(G_i=g)\pi_i^{-1}(g)K_1((Y_i-y)/h)/h$ and recall that $\widehat{f}_{ng}(y)=n^{-1}\sum_{i=1}^n Z_i^{g,h}$ is the kernel estimator of $f_g$. We now state the proposition on the convergence rate of the kernel density estimator under interference.
\begin{proposition}
\label{prop:pdf_conv}
    Let $K_1$ be a differentiable kernel function satisfying the following regularity conditions: 
    
    (i) $K_1$ is an absolutely continuous kernel satisfying:
    \[
    \int_{\mathbb R} K_1(v)\,dv=1,
    \qquad
    K_1(v)=K_1(-v),
    \qquad
    \int_{\mathbb R}|K_1'(v)|\,dv<\infty.
    \]

    (ii) $K_1$ is supported on a compact set $[-M,M]$, for some positive constant $M$.

    (iii) For the constant $M$ defined in $(ii)$, $F_{ng}$ is bandwidth scale smooth in a neighborhood of $y$.
    
Suppose $h=h_n\to0$ and $nh\to\infty$, and let $g\in\boldsymbol G$ be an
exposure level of interest. Then
    $$|\widehat{f}_{ng}(y)-f_g(y)|\le O_p\left(
    \sqrt{\frac{1+d_{mx}^{(2)}}{nh}}
    +
    h^2\right).$$ 
    In particular, under optimal bandwidth $h \asymp
    \left(
    \frac{1+d_{mx}^{(2)}}{n}
    \right)^{1/5}$,
    \[
    \widehat f_{ng}(y)-f_g(y)
    =
    O_p\left[
    \left(
    \frac{1+d_{mx}^{(2)}}{n}
    \right)^{2/5}
    \right].
    \]
\end{proposition}

\begin{proof}[Proof of Proposition~\ref{prop:pdf_conv}]

We decompose 
$$\widehat f_{ng}(y)-f_g(y)
=
\left\{\widehat f_{ng}(y)-E\widehat f_{ng}(y)\right\}
+
\left\{E\widehat f_{ng}(y)-f_g(y)\right\}.$$

We first bound the stochastic term. 
Since $\widehat{f}_{ng}(y)=n^{-1}\sum_{i=1}^n Z_i^{g,h}$, we have 
\[ 
\widehat f_{ng}(y)-E\widehat f_{ng}(y) = \frac1n\sum_{i=1}^n \left\{Z_i^{g,h}-E Z_i^{g,h}\right\}. 
\] By overlap and
the inverse-probability-weighting identity,
\[
\begin{aligned}
\sum_{i\in N_n}\operatorname{Var}\{Z_i^{g,h}\}
&\le \sum_{i\in N_n}E\{Z_i^{g,h}\}^2\\
&\le \frac{n}{\underline\pi h^2}
\int K_1^2\left(\frac{z-y}{h}\right)dF_{ng}(z).
\end{aligned}
\]
Because $K_1$ is bounded and supported on $[-M_K,M_K]$, the last integral is
bounded by a constant times the $F_{ng}$-mass of an interval of length
$2(M_K+1)h$ about $y$. Bandwidth-scale smoothness implies that this mass is
$O(h)$. Hence, directly under the design-based distribution,
\[
\sum_{i\in N_n}\operatorname{Var}\{Z_i^{g,h}\}=O(n/h).
\]
Applying Lemma~\ref{lem:CND_bound_rate}, 
    $$
    \begin{aligned}
        \sum_{i=1}^n (Z_i^{g,h}-E[Z_i^{g,h}])
        =&O_p\left(\left\{(1+d_{mx}^{(2)})\sum_{i\in N_n}\mathrm{Var}(Z_i^{g,h})\right\}^{1/2}\right)\\
        =&O_p\left((1+d_{mx}^{(2)})^{1/2}\cdot(n/h)^{1/2}\right).
    \end{aligned}
    $$ We thus obtain $\widehat f_{ng}(y)-E\widehat f_{ng}(y)=O_p\left(
\sqrt{\frac{1+d_{mx}^{(2)}}{nh}}
\right).$

To bound the second term, we note that
\[
E\widehat f_{ng}(y)
=\frac1n\sum_{i=1}^n
E\left[
\frac{I(G_i=g)}{\pi_i(g)}
\frac1h K_1\left(\frac{Y_i-y}{h}\right)
\right]
=\int \frac1h K_1\left(\frac{z-y}{h}\right)
\,dF_{ng}(z).
\]
Using integration by parts and applying the boundary condition on $K_1$,
\[
\begin{aligned}
E\widehat f_{ng}(y)
&=
-\frac{1}{h^2}
\int
K_1'\left(\frac{z-y}{h}\right)
F_{ng}(z)\,dz =
-\frac1h
\int
K_1'(v)F_{ng}(y+hv)\,dv,
\end{aligned}
\]
where the second equality follows from the change of variables $z=y+hv$. 
Define 
\[
r_{ng,y}(t)=F_{ng}(y+t)-F_{ng}(y)-f_g(y)t-
\frac12 f_g'(y)t^2.
\]
By the bandwidth-scale smoothness condition,  $ \sup_{|t|\le M h}|r_{ng,y}(t)| = O(h^3)$. 
Therefore,
\[
\begin{aligned}
E\widehat f_{ng}(y)
&=-\frac1h
\int^M_{-M} K_1'(v)F_{ng}(y+hv)\,dv \\
&=-\frac1h\int^M_{-M} K_1'(v)
\left\{
F_{ng}(y)+f_g(y)hv
+\frac12 f_g'(y)h^2v^2
+r_{ng,y}(hv)\right\}\,dv \\
&=
-\frac{F_{ng}(y)}{h}\int^M_{-M} K_1'(v)\,dv
-f_g(y)\int vK_1'(v)\,dv \\
&\quad -\frac12 f_g'(y)h\int^M_{-M} v^2K_1'(v)\,dv
-\frac1h\int^M_{-M} K_1'(v)r_{ng,y}(hv)\,dv .
\end{aligned}
\]
Because $K_1$ is compactly supported and absolutely continuous, it vanishes
at the endpoints. Hence
$\int_{-M_K}^{M_K}K_1'(v)\,dv=0$,
$\int_{-M_K}^{M_K}vK_1'(v)\,dv=-1$, and, by symmetry,
$\int_{-M_K}^{M_K}v^2K_1'(v)\,dv=0$.
Consequently,
\[ 
E\widehat f_{ng}(y)-f_g(y) = -\frac1h \int_{-M_K}^{M_K}K_1'(v)r_{ng,y}(hv)\,dv . 
\]
Since $\int_{-M_K}^{M_K}|K_1'(v)|\,dv<\infty$, it follows that 
\[ \left| E\widehat f_{ng}(y)-f_g(y) \right| \le \frac1h \left\{ \sup_{|t|\le M_K h}|r_{ng,y}(t)| \right\} \int_{-M_K}^{M_K}|K_1'(v)|\,dv = O(h^2). 
\]

Combining this with the stochastic term gives
\[
\widehat f_{ng}(y)-f_g(y)
=
O_p\left(
\sqrt{\frac{1+d_{mx}^{(2)}}{nh}}
+
h^2
\right).
\]
The bandwidth that balances the two terms satisfies
$h \asymp
\left(
\frac{1+d_{mx}^{(2)}}{n}
\right)^{1/5}$,
which gives
\[
\widehat f_{ng}(y)-f_g(y)
=
O_p\left[
\left(
\frac{1+d_{mx}^{(2)}}{n}
\right)^{2/5}
\right].
\]
\end{proof}

Next, we establish the following proposition which strengthens the point-wise argument in Proposition~\ref{prop:pdf_conv} to a
locally uniform one over $\mathcal I_{n,L}(g)=\{s:|s-\beta_{nq}(g)|\le L/\sqrt n\}$.
\begin{proposition}
\label{prop:pdf_unif_conv}
Under the conditions of Proposition~\ref{prop:pdf_conv}, suppose in addition
that $F_{ng}$ is local-uniform bandwidth-scale smooth at
$\beta_{nq}(g)$ and
\[
\frac{\{1+d_{mx}^{(2)}\}\log n}{nh}\longrightarrow0.
\]
Then, for every fixed $L<\infty$,
\[
\sup_{y\in\mathcal I_{n,L}(g)}
\left|\widehat f_{ng}(y)-f_g(y)\right|
=
O_p\left[
\left\{\frac{(1+d_{mx}^{(2)})\log n}{nh}\right\}^{1/2}
+h^2
\right].
\]
\end{proposition}

\begin{proof}[Proof of Proposition~\ref{prop:pdf_unif_conv}]
Take $\Delta_n=1+d_{mx}^{(2)}$ and
\[
\mathbb Z_n(y)=\widehat f_{ng}(y)-E\widehat f_{ng}(y)
=\frac1n\sum_{i\in N_n}
\left[Z_i^{g,h}-EZ_i^{g,h}\right].
\]
The two-step graph is an exact dependency graph for these summands. The
aggregate calculation in the proof of Proposition~\ref{prop:pdf_conv}, now
using local-uniform bandwidth-scale smoothness, gives
\[
\sup_{y\in\mathcal I_{n,L}(g)}
\sum_{i\in N_n}\operatorname{Var}\{Z_i^{g,h}\}
\le Cn/h.
\]
Moreover, boundedness of $K_1$ and overlap imply
$|Z_i^{g,h}-EZ_i^{g,h}|\le C/h$ uniformly in $i$ and $y$.
Equation~(3.8) of \citet{lee2019stable}, applied to the exact two-step
dependency graph with trivial conditioning fields, therefore gives constants
$c,C>0$, independent of $n$, $y$, and $x$, such that
\begin{equation}
\label{eq:kernel-bernstein}
\Pr\{|\mathbb Z_n(y)|>x\}
\le
2\exp\left\{
-\frac{cnhx^2}{\Delta_n(1+x)}
\right\}
\end{equation}
for $y\in\mathcal I_{n,L}(g)$ and $x>0$ in the range used below.

Let
\[
a_n=\left(\frac{\Delta_n\log n}{nh}\right)^{1/2}=o(1)
\]
and choose a grid over $\mathcal I_{n,L}(g)$ with mesh
$\eta_n=h^2a_n$. The rate condition implies that the number $N_n$ of grid
points is at most polynomial in $n$, so $\log N_n=O(\log n)$. Substituting
$x=Ma_n$ into \eqref{eq:kernel-bernstein}, using $a_n=o(1)$, and applying a
union bound gives, for $M$ sufficiently large,
\[
\max_{1\le j\le N_n}|\mathbb Z_n(y_j)|=O_p(a_n).
\]
Because $K_1$ is globally Lipschitz,
\[
|\mathbb Z_n(y)-\mathbb Z_n(y')|
\le C\frac{|y-y'|}{h^2}.
\]
Every point of $\mathcal I_{n,L}(g)$ lies within $\eta_n$ of a grid point,
and hence
\[
\sup_{y\in\mathcal I_{n,L}(g)}|\mathbb Z_n(y)|=O_p(a_n).
\]

The integration-by-parts calculation in Proposition~\ref{prop:pdf_conv} is
uniform under local-uniform bandwidth-scale smoothness and yields
\[
\sup_{y\in\mathcal I_{n,L}(g)}
|E\widehat f_{ng}(y)-f_g(y)|=O(h^2).
\]
Combining the stochastic and bias bounds proves the claim.
\end{proof}

\begin{corollary}
\label{cor:fhat_unif_conv}
Under the conditions of Proposition~\ref{prop:pdf_unif_conv}, suppose
\[
\left\{\frac{(1+d_{mx}^{(2)})\log n}{nh}\right\}^{1/2}+h^2\to0.
\]
Then, for every fixed $L<\infty$,
\[
\sup_{y\in\mathcal I_{n,L}(g)}
|\widehat f_{ng}(y)-f_g(y)|=o_p(1).
\]
If in addition
$\widehat\beta_{nq}(g)-\beta_{nq}(g)=O_p(n^{-1/2})$, then
\[
\widehat f_{ng}\{\widehat\beta_{nq}(g)\}
-f_g\{\widehat\beta_{nq}(g)\}=o_p(1).
\]
The same statements hold with $o_p(1)$ replaced by $O_p(\epsilon_n)$ whenever
the displayed uniform rate is $O_p(\epsilon_n)$.
\end{corollary}

\begin{proof}
The first assertion is Proposition~\ref{prop:pdf_unif_conv}. For the plug-in
evaluation, given $\eta>0$, choose a fixed $L$ such that
\[
\limsup_n\Pr\left\{
\sqrt n|\widehat\beta_{nq}(g)-\beta_{nq}(g)|>L
\right\}<\eta.
\]
On the complementary event that $\widehat\beta_{nq}(g)$ lies in
$\mathcal I_{n,L}(g)$, the uniform bound applies. The conclusion follows since $\eta$ is
arbitrary.
\end{proof}

\begin{remark}
When there is no network interference, $d_{mx}^{(2)}=O(1)$, the choice $h \asymp n^{-1/5}$ leads to the standard kernel density rate
$O_p(n^{-2/5})$.
The compact-support assumption is used only to simplify the bias calculation. 
    This conditions is satisfied by common used kernels such as the Epanechnikov kernel
\[
K(u)=\frac34(1-u^2)1\{|u|\le 1\},
\]
the uniform kernel
\[
K(u)=\frac12 1\{|u|\le 1\},
\]
and higher-order compactly supported kernels such as the biweight and triweight kernels. For kernels with unbounded support, $(ii)$ can be substituted with the following conditions on boundary and tail convergence:

    $(ii,a)$ Boundary conditions on $K_1:$
    \[
    K_1(v)\to 0,\qquad vK_1(v)\to 0,\qquad v^2K_1(v)\to 0
    \quad\text{as } |v|\to\infty.
    \]    

    $(ii,b)$ Tail convergence: For some fixed constant $\alpha>0$ satisfying $hn^{\alpha}\to 0$, 
    \[
    \int_{-\infty}^{-Mn^\alpha} |K_1'(v)|\,dv+\int_{Mn^\alpha}^{\infty} |K_1'(v)|\,dv\le O(h^3).
    \]

    And the smoothness condition on $F_{ng}$ can be substituted by the following condition, which enables local convergence on a slighted broader neighborhood of order $Mhn^\alpha$:
    
    $(iii,a)$ For $\alpha>0$ defined in Condition $(ii,b),$ $F_{ng}$ satisfies
    \[
\sup_{|t|\le Mh n^\alpha}
\left|
F_{ng}(y+t)
-
F_{ng}(y)
-
f_g(y)t
-
\frac12 f_g'(y)t^2
\right|
=
O\left(h^3\right).
\]
    
The Conditions $(ii,a)$, $(ii,b)$ and $(iii,a)$ are satisfied for common kernels such as Gaussian kernel, when $h$ is at a polynomial rate of $n$ and $\alpha>0$ can be set as an arbitrary positive number.
\end{remark}

\subsection{Proof of Theorem~\ref*{thm:asy_nor}}

In this subsection, we provide the details on the proof of Theorem~\ref*{thm:asy_nor} and derive the asymptotic normality of the quantile exposure effect estimator. The proof consists of two steps. In the first step, we express the parameter of interest as minimizer of stochastic process, and construct a quadratic-process approximation by bounding the error terms. In the second step, we prove the nearness between the minimizers of the target process and its quadratic approximation, which yields the linear expansion of the estimator, and the theorem is proved by applying CLT in Lemma~\ref{lem:CLT_in_networks}. Define the function $U_{q,n}(\beta;g)=\sum_{i\in N_n}\frac{I(G_i=g)}{\pi_i(g)}\rho_q(Y_i-\beta),$ where $\pi_i(g)$ is the generalized propensity score and $\rho_q(u)=u (q - I(u\leq0))$
is the check loss function. Recall that from Eq. (\ref*{eq:minimization}) in the main paper, the quantile estimator at exposure level $g\in\bm{G}$ is given by
$$
\begin{aligned}
\widehat{\beta}_{nq}(g)&\in\arg\min_{\beta} U_{q,n}(\beta;g)\\
&=\arg\min_{\beta} \sum_{i \in N_n}\frac{I(G_i=g)}{\pi_i(g)}(Y_i - \beta)(q-I\{Y_i\leq \beta\}),
\end{aligned}
$$
where $\beta_{nq}(g)$ denotes the true quantile. 

We have 
$$
\begin{aligned}
\widehat{\beta}_{nq}(g)\in&\underset{\beta}{\arg\min}\Big\{U_{q,n}(\beta;g)-U_{q,n}(\beta_{nq}(g);g)\Big\}\\
=&\underset{\beta}{\arg\min}\sum_{i \in N_n}\frac{I(G_i=g)}{\pi_i(g)}\big[(Y_i - \beta)(q-I\{Y_i\leq \beta\})-(Y_i - \beta_{nq}(g))(q-I\{Y_i\leq \beta_{nq}(g)\})\big]\\
=&\underset{\beta}{\arg\min}\sum_{i \in N_n}\frac{I(G_i=g)}{\pi_i(g)}\big[(I\{Y_i\leq \beta_{nq}(g)\}-q)(\beta - \beta_{nq}(g))\\
&+(Y_i - \beta)(I\{Y_i\leq \beta_{nq}(g)\}-I\{Y_i\leq \beta\})\big].
\end{aligned}
$$
Now, define $u=\sqrt{n}(\beta - \beta_{nq}(g)),\widehat{u}_{q}=\sqrt{n}(\widehat{\beta}_{nq}(g)-\beta_{nq}(g))$, and recall the definitions
\[
D_q(Y_i;g)=q-I\{Y_i\le \beta_{nq}(g)\},
\]
\[
R_q(Y_i,u;g)=\Big\{Y_i-\big(\beta_{nq}(g)+u/\sqrt{n}\big)\Big\}
\Big\{I(Y_i\le \beta_{nq}(g))
-I\big(Y_i\le \beta_{nq}(g)+u/\sqrt{n}\big)\Big\}.
\]
The quantity $D_q(Y_i;g)$ corresponds to the subgradient of the check loss evaluated at $\beta_{nq}(g)$ for observation $i$, while $R_q(Y_i,u;g)$ is the associated remainder term arising from the local linearization of the objective function.
Define the localized objective
\begin{equation}
\label{eq:def_Qq}
    Q_q(u;g)
=\sum_{i\in N_n}\frac{I(G_i=g)}{\pi_i(g)}
\left\{-\frac{u}{\sqrt{n}}D_q(Y_i;g)+R_q(Y_i,u;g)\right\}.
\end{equation}
We have
\[
\widehat u_q\in \arg\min_{u\in\mathbb{R}} Q_q(u;g).
\]
Next, define the quadratic process $\widetilde{Q}_{q}(u;g)$:
$$
\begin{aligned}
\widetilde{Q}_{q}(u;g)&=-\frac{u}{\sqrt{n}}\sum_{i \in N_n}\frac{I(G_i=g)}{\pi_i(g)}D_{q}(Y_i;g)+u^{2}\frac{f_g(\beta_{nq}(g))}{2},
\end{aligned}
$$
where $f_g$ is the limiting density function defined in Assumption~\ref*{asp:pdf}.

\textbf{Step 1.} In the first step, we show that $\xi_{q,n}(u; \pi, g):=Q_{q}(u;g)-\widetilde{Q}_{q}(u; g)$ is $o_p(1)$ for any fixed $u$. For fixed $g\in\boldsymbol{G}$ and $u\in\mathcal{U},$ denote $F_g(\tilde{y})=\int_{-\infty}^{\tilde{y}} f_g(y)\,dy$ as the corresponding cumulative distribution function. From Equation (\ref{eq:def_Qq}) and triangular inequality, the difference term can be bounded by
\begin{equation}
\label{eq:xi_bound}
    \begin{aligned}
    \left|\xi_{q,n}(u;\pi,g)\right|=&\left|\sum_{i \in N_n}\frac{I(G_i=g)}{\pi_i(g)}  R_{q}(Y_i,u;g)-u^{2} \frac{f_g(\beta_{nq}(g))}{2}\right|\\
    \le&\left|\sum_{i \in N_n}E\frac{I(G_i=g) R_{q}(Y_i,u;g)}{\pi_i(g)}-u^{2} \frac{f_g(\beta_{nq}(g))}{2}\right|\\
    &+\left|\sum_{i \in N_n}\frac{I(G_i=g)}{\pi_i(g)}  R_{q}(Y_i,u;g)-\sum_{i \in N_n}E\frac{I(G_i=g) R_{q}(Y_i,u;g)}{\pi_i(g)}\right|\\
   :=&  T_1+T_2.
\end{aligned}
\end{equation}

Without loss of generality, assume that $u>0.$ We now demonstrate that both terms $T_1,T_2$ in the last equation in  (\ref{eq:xi_bound}) is $o_p(1)$. For simplicity, denote $b_n=\beta_{nq}(g)$ and $a_n=u/\sqrt{n}.$

For the first term $T_1$, we start with the following observation

\begin{equation}
    \begin{aligned}
    &\frac{1}{n}\sum_{i \in N_n}E\frac{I(G_i=g)}{\pi_i(g)}  R_{q}(Y_i,u;g)\\
    =& \int R_{q}(Y_i,u;g) d F_{ng}(y)\\
    =& \int_{b_n}^{a_n+b_n}(a_n+b_n-y)d F_{ng}(y)\\
    =& -a_nF_{ng}(b_n)+\int_{b_n}^{a_n+b_n} F_{ng}(y)dy\\
    =& \int_{b_n}^{a_n+b_n} [F_{ng}(y)-F_{ng}(b_n)]dy\\
    =& \int_{b_n}^{a_n+b_n}\{f_g(b_n)(y-b_n)\}dy+o(n^{-1/2})a_n\\
    =& \frac{u^2f_g(b_n)}{2n}+o(n^{-1}),
\end{aligned}
\end{equation}
where the second last equation applies Assumption~\ref*{asp:pdf}. Multiplying $n$ at both sides, we obtain $\sum_{i \in N_n}E\frac{I(G_i=g)}{\pi_i(g)}  R_{q}(Y_i,u;g)=\frac{u^2f_g(b_n)}{2}+o_p(1)$ and therefore $T_1=o_p(1).$

For the second term $T_2,$ applying Lemma~\ref{lem:CND_bound_rate} with $M_i=\frac{I(G_i=g)  R_{q}(Y_i,u;g)}{\pi_i(g)}$ and recall that $a_n=u/\sqrt{n},$ $b_n=\beta_{nq}(g),$ we have
\begin{equation}
\label{eq:T1_bound}
    \begin{aligned}
    T_2=&\sum_{i \in N_n}\left\{\frac{I(G_i=g)  R_{q}(Y_i,u;g)}{\pi_i(g)}-E\frac{I(G_i=g)  R_{q}(Y_i,u;g)}{\pi_i(g)}\right\}\\
    \le& O_p\left(\left\{n(1+d_{mx}^{(2)})  \int_{\beta_{nq}(g)}^{\beta_{nq}(g)+u/\sqrt{n}}R_{q}^2(y,u;g)dF_{ng}(y)\right\}^{1/2}\right).\\
\end{aligned}
\end{equation}
To bound the last term above, similarly as the arguments for $T_1$, we have
\[
\begin{aligned}
&\int_{b_n}^{b_n+a_n}R_q^2(y,u;g)\,dF_{ng}(y)\\
&=\int_{b_n}^{b_n+a_n}(b_n+a_n-y)^2\,dF_{ng}(y)\\
&=
\left.(b_n+a_n-y)^2F_{ng}(y)\right|_{b_n}^{b_n+a_n}
+
2\int_{b_n}^{b_n+a_n}(b_n+a_n-y)F_{ng}(y)\,dy\\
&=
-a_n^2F_{ng}(b_n)
+
2\int_{b_n}^{b_n+a_n}(b_n+a_n-y)F_{ng}(y)\,dy\\
&=
2\int_{b_n}^{b_n+a_n}(b_n+a_n-y)
\{F_{ng}(y)-F_{ng}(b_n)\}\,dy.
\end{aligned}
\]
From the local expansion of $F_{ng}$ at $b_n=\beta_{nq}(g)$ in Assumption~\ref*{asp:pdf},
\[
F_{ng}(y)-F_{ng}(b_n)
=
f_g(b_n)(y-b_n)
+o(n^{-1/2})
\] for $|y-b_n|\le a_n.$ Therefore we obtain
\[
\begin{aligned}
&\int_{b_n}^{b_n+a_n}R_q^2(y,u;g)\,dF_{ng}(y)\\
&=
2f_g(b_n)\int_0^{a_n}(a_n-s)s\,ds
+
o(n^{-1/2})\int_0^{a_n}(a_n-s)\,ds\\
&=c_0f_g(b_n)n^{-3/2}+
o(n^{-3/2})
\end{aligned}
\]
for some constant $c_0$ and therefore 
\begin{equation}
\label{eq:Rq2_bound}
\int_{b_n}^{b_n+a_n}R_q^2(y,u;g)\,dF_{ng}(y)=O(n^{-3/2}).
\end{equation}

Combining (\ref{eq:T1_bound}) together with (\ref{eq:Rq2_bound}) and meanwhile note that $d_{mx}^{(2)}=o(n^{1/2})$, we finally get
$$
\begin{aligned}
    T_2\le&O_p\left(\left\{n(1+d_{mx}^{(2)})  \int_{\beta_{nq}(g)}^{\beta_{nq}(g)+u/\sqrt{n}}R_{q}^2(y,u;g)dF_g(y)+(1+d_{mx}^{(2)})n^{-1/2}\right\}^{1/2}\right)\\
    =& O_p\left(\left\{(1+d_{mx}^{(2)})n^{-1/2}\right\}^{1/2}\right)=o_p(1).\\
\end{aligned}
$$ Combining the results above, we finally have $|\xi_{q,n}(u; \pi, g)|\le T_1+T_2=o_p(1),$ and by convexity, for fixed $M>0$, the objective function 
$\sup_{u\in M}Q_{q}(u;g)$ can be uniformly approximated by the quadratic process $\widetilde{Q}_{q}(u;g)$ in the sense that $$\sup_{|u|\le M}\left|Q_{q}(u;g)-\widetilde{Q}_{q}(u;g)\right|=o_p(1).$$

\textbf{Step 2.}
Next, let $\widetilde{u}_{q}$ be the minimizer of the random quadratic function $\widetilde{Q}_{q}(u;g)$. Recall $\widehat{u}_q$ is the minimizer of $Q_{q}(u;\widehat{\pi},g)$. 
The nearness of minimizers of convex
random functions result in ~\citet{hjort2011asymptotics}, see for example Lemma A.4 of~\citet{firpo2007efficient}, ensures the following probabilistic bound on $\widehat{u}_q-\widetilde{u}_q$: for each $\varepsilon > 0$,
\[
\mathrm{pr}(|\widehat{u}_q - \widetilde{u}_q| \geq \varepsilon) \leq \mathrm{pr}\left( \sup_{|u - \widetilde{u}_q| \leq \varepsilon} |\xi_{q,n}(u;\pi,g)| \geq \frac{1}{4} f_g(\beta_{nq}(g)) \varepsilon^2 \right),
\] 
furthermore,
\[
\mathrm{pr}\left( \sup_{|u - \widetilde{u}_q| \leq \varepsilon} |\xi_{q,n}(u;\pi,g)| \geq \frac{1}{4} f_g(\beta_{nq}(g)) \varepsilon^2 \right) = o(1).
\] 
Therefore, $|\widehat{u}_q - \widetilde{u}_q|=o_p(1).$ Define $\widetilde{\beta}_q(g) = \widetilde{u}_q / \sqrt{n} + \beta_{nq}(g)$ and recall that $\widehat{\beta}_{nq}(g)=\widehat{u}_q/ \sqrt{n} + \beta_{nq}(g)$, we have
\[
\sqrt{n}|\widehat{\beta}_{nq}(g) - \widetilde{\beta}_q(g)| = |\sqrt{n}(\widehat{\beta}_{nq}(g) - \beta_{nq}(g)) - \sqrt{n}(\widetilde{\beta}_q(g) - \beta_{nq}(g))|
= |\widehat{u}_q - \widetilde{u}_q| = o_p(1).
\]
Finally, we observe that from the quadratic form of $\widetilde{Q}_{q}(u;g)$, the minimizer $\widetilde{u}_{q}$ has an explicit expression
\[
\begin{aligned}\widetilde{u}_{q}=\frac{-1}{\sqrt{n}f_g(\beta_{nq}(g))}\sum_{i \in N_n}\frac{I(G_i=g)}{\pi_i(g)}  D_{q}(Y_i;g).
\end{aligned}
\] 
Therefore, 
\[\sqrt{n}\left\{\widehat{\beta}_{nq}(g)-\beta_{nq}(g)\right\}=\frac{-1}{\sqrt{n}}\sum_{i \in N_n}\psi_{i,q}(Y_i,g),\]
where $\psi_{i,q}(y,g):=\frac{I\{G_i=g\}}{\pi_i^{-1}(g)}f_g^{-1}(\beta_{nq}(g))D_{q}(y;g)$. Comparing the contrasts on exposure levels $g,g'$ we get
$$
\begin{aligned}
    &\sqrt{n}\left\{\widehat{\tau}_{nq}(g,g')-\tau_{nq}(g,g')\right\}\\
    =&\frac{1}{\sqrt{n}} \sum_{i \in N_n}\Big\{\psi_{i,q}(Y_i,g)-\psi_{i,q}(Y_i,g')\Big\} + o_p(1)\\
    =&\frac{1}{\sqrt{n}} \sum_{i \in N_n}\Big\{\psi_{i,q}(Y_i,g)-\psi_{i,q}(Y_i,g')-E[\psi_{i,q}(Y_i,g)-\psi_{i,q}(Y_i,g')]\Big\}\\
    &+ \frac{1}{\sqrt{n}} \sum_{i \in N_n}\Big\{E\psi_{i,q}(Y_i,g)-E\psi_{i,q}(Y_i,g')\Big\}+o_p(1)\\
\end{aligned}
$$ The first term in the last equation above is asymptotically normal by applying the central limit theorem in networks in Lemma~\ref{lem:CLT_in_networks} with $m(W_i,G_i,Y_i,\bm{X})=n^{-1/2}[\psi_{i,q}(Y_i,g)-\psi_{i,q}(Y_i,g')-E\{\psi_{i,q}(Y_i,g)-\psi_{i,q}(Y_i,g')\}].$ 
To evaluate the second term, we note that since $\beta_{nq}(g)$ is the $q$-th quantile of $F_{ng}$, we have
$F_{ng}\{\beta_{nq}(g)-\}\le q \le F_{ng}\{\beta_{nq}(g)\}$ and
$
E\{q-n^{-1}\sum_{i\in N_n}\frac{I(G_i=g)}{\pi_i(g)}I(Y_i\le \beta_{nq}(g))\}
=
q-F_{ng}\{\beta_{nq}(g)\}=o(n^{-1/2})$.
This implies $\frac{1}{\sqrt{n}} \sum_{i \in N_n}E\psi_{i,q}(Y_i,g)=o_p(1)$ and combining the same result at exposure level $g'$ by symmetry we get $$
\frac{1}{\sqrt{n}} \sum_{i \in N_n}\Big\{E\psi_{i,q}(Y_i,g)-E\psi_{i,q}(Y_i,g')\Big\}=o_p(1).
$$
Theorem~\ref*{thm:asy_nor} is proved by combining the observations above.

\subsection{Proof of Theorem~\ref*{thm:asy-nor-janson}}

\begin{proof}
Fix \(q\in(0,1)\) and \(g,g'\in\boldsymbol G\).  We first verify that
Assumption~\ref*{asp:janson} preserves the local quadratic expansion used in
the proof of Theorem~\ref*{thm:asy_nor}.  Put
\[
\Delta_n=1+d_{mx}^{(2)}.
\]
The exposure-specific variance bound in Assumption~\ref*{asp:janson} and the
inequality \(\operatorname{Var}(X-Y)\leq
2\operatorname{Var}(X)+2\operatorname{Var}(Y)\) imply
\(\sigma_n^2=O(1)\).  Hence the rate condition in
Assumption~\ref*{asp:janson} gives
\[
n^{1/m-1/2}\Delta_n^{1-1/m}
=o(1),
\]
and therefore
\[
\Delta_n
=o\left(n^{(m-2)/\{2(m-1)\}}\right)
=o(n^{1/2}).
\]

We next recall how this rate enters the linearization.  For
\(s\in\{g,g'\}\), write
\[
b_{n,s}=\beta_{nq}(s),\qquad
D_{i,s}=q-I(Y_i\leq b_{n,s}),
\]
and let
\[
Q_{n,s}(u)
=U_{q,n}(b_{n,s}+u/\sqrt n;s)-U_{q,n}(b_{n,s};s).
\]
The check-loss identity used in the proof of
Theorem~\ref*{thm:asy_nor} yields, for every fixed \(u\),
\[
Q_{n,s}(u)
=-\frac{u}{\sqrt n}\sum_{i\in N_n}
  \frac{I(G_i=s)}{\pi_i(s)}D_{i,s}
+\frac{u^2}{2}f_s(b_{n,s})+o_p(1).
\]
Indeed, Assumption~\ref*{asp:pdf} gives the quadratic term in expectation.
For the centered remainder, the calculation in
Equations~\eqref{eq:T1_bound}--\eqref{eq:Rq2_bound} gives the bound
\[
O_p\!\left(\{\Delta_n n^{-1/2}\}^{1/2}\right)=o_p(1),
\]
where the last equality follows from the displayed degree rate.  The same
argument applies for negative \(u\).  Convexity then makes this approximation
uniform on every compact set.

To justify the minimizer step, define
\[
L_{n,s}=\frac1{\sqrt n}\sum_{i\in N_n}
\frac{I(G_i=s)}{\pi_i(s)}D_{i,s}.
\]
By the definition of the exposure-specific score,
\[
L_{n,s}
=f_s(b_{n,s})\frac1{\sqrt n}\sum_{i\in N_n}
\psi_{i,q}(Y_i,s).
\]
Moreover,
\[
E L_{n,s}
=\sqrt n\{q-F_{ns}(b_{n,s})\}=o(1).
\]
Indeed, \(F_{ns}(b_{n,s}-)\leq q\leq F_{ns}(b_{n,s})\), so
\[
0\leq F_{ns}(b_{n,s})-q
\leq F_{ns}(b_{n,s})-F_{ns}(b_{n,s}-)
=o(n^{-1/2})
\]
under Assumption~\ref*{asp:pdf}.  The exposure-specific variance bound in
Assumption~\ref*{asp:janson}, together with the upper bound on
\(f_s(b_{n,s})\), now gives \(L_{n,s}=O_p(1)\).
The quadratic approximation has the unique minimizer
\[
\widetilde u_{n,s}
=\frac{L_{n,s}}{f_s(b_{n,s})}
=\frac1{\sqrt n}\sum_{i\in N_n}\psi_{i,q}(Y_i,s),
\]
which is \(O_p(1)\) because \(f_s(b_{n,s})\geq c_f\).  The convex-minimizer
argument used in Step~2 of the proof of Theorem~\ref*{thm:asy_nor} therefore
implies
\[
\sqrt n\{\widehat\beta_{nq}(s)-\beta_{nq}(s)\}
=\frac1{\sqrt n}\sum_{i\in N_n}\psi_{i,q}(Y_i,s)+o_p(1).
\]
Subtracting the expansions for \(s=g\) and \(s=g'\) proves
\[
\sqrt n\{\widehat\tau_{nq}(g,g')-\tau_{nq}(g,g')\}
=\frac1{\sqrt n}\sum_{i\in N_n}\psi_{i,q}+o_p(1),
\]
where
\(\psi_{i,q}=\psi_{i,q}(Y_i,g)-\psi_{i,q}(Y_i,g')\).

It remains to establish the limiting distribution.  Let
\[
Z_{i,n}=\psi_{i,q}-E\psi_{i,q},\qquad
S_n=\sum_{i\in N_n}Z_{i,n},\qquad
s_n^2=\operatorname{Var}(S_n)=n\sigma_n^2.
\]
By Corollary~\ref*{cor:two-step-dependency}, the two-step graph is a
dependency graph for \(\{Z_{i,n}:i\in N_n\}\), and its maximal degree is at
most \(d_{mx}^{(2)}=\Delta_n-1\).  Assumptions~\ref*{asp:overlap}
and~\ref*{asp:pdf} imply, for all sufficiently large \(n\),
\[
|Z_{i,n}|\leq B:=\frac{4}{\underline\pi c_f}
\qquad\text{almost surely, uniformly in }i.
\]
Consequently, the condition in Janson's dependency-graph central limit
theorem \citep{janson1988normal} is
\[
\left(\frac n{\Delta_n}\right)^{1/m}
\frac{\Delta_n B}{s_n}
=B\left(\frac n{\Delta_n}\right)^{1/m}
\frac{\Delta_n}{\sqrt n\,\sigma_n}
\longrightarrow0,
\]
which is precisely the rate condition in Assumption~\ref*{asp:janson}, up to
the fixed constant \(B\).  Since \(\liminf_n\sigma_n^2>0\), \(s_n^2>0\) for all
sufficiently large \(n\), and Janson's theorem gives
\[
\frac{S_n}{s_n}
=\frac{n^{-1/2}\sum_{i\in N_n}Z_{i,n}}{\sigma_n}
\ \longrightarrow\ N(0,1)
\]
in distribution.

Finally, the centering term is negligible.  For each
\(s\in\{g,g'\}\),
\[
\frac1{\sqrt n}\sum_{i\in N_n}E\psi_{i,q}(Y_i,s)
=\frac{\sqrt n}{f_s(b_{n,s})}
\{q-F_{ns}(b_{n,s})\}
=o(1),
\]
again by the quantile-jump condition in Assumption~\ref*{asp:pdf}.  Thus
\(n^{-1/2}\sum_iE\psi_{i,q}=o(1)\).  Combining this fact with the linear
representation and the centered central limit theorem above, and using
\(\liminf_n\sigma_n^2>0\), yields
\[
\sqrt n\,\sigma_n^{-1}
\{\widehat\tau_{nq}(g,g')-\tau_{nq}(g,g')\}
=\frac{S_n}{s_n}+o_p(1)
\ \longrightarrow\ N(0,1)
\]
in distribution.  This proves both conclusions of the theorem.
\end{proof}

\subsection{Proof of Theorem~\ref*{thm:var_est}}
Recall that for fixed $g,g'\in\boldsymbol{G},$ define $\psi_{i,q}=\psi_{i,q}(Y_i,g)-\psi_{i,q}(Y_i,g')$ and $\widehat{\psi}_{i,q}=\widehat{\psi}_{i,q}(Y_i,g)-\widehat{\psi}_{i,q}(Y_i,g')$, where $$
    \widehat{\psi}_{i,q}(Y_i,g)=\frac{I(G_i=g)}{\pi_i(g)}\widehat{f}_{ng}^{-1}(\widehat{\beta}_{nq}(g))\left\{ q - I(Y_i \leq \widehat{\beta}_{nq}(g)) \right\}
    $$ 
    is the plug-in estimator using the estimated density $\widehat{f}_{ng}$ and quantiles $\widehat{\beta}_{nq}(g),\widehat{\beta}_{nq}(g')$. 
    We use superscript $M$ to denote the average over finite samples, for example, $\widehat{\psi}_{q}^M=n^{-1}\sum_{i\in N_n}\widehat{\psi}_{i,q}.$ We denote the variance estimator $\widehat{\sigma}_n^2=\frac{1}{n}\sum_{i,j\in N_n,\ell_{A}(i,j)\le 2}(\widehat{\psi}_{i,q}-\widehat{\psi}^M_q)(\widehat{\psi}_{j,q}-\widehat{\psi}^M_q)$ and the true variance $\sigma_n^2=E\{\frac{1}{\sqrt{n}}\sum_{i\in N_n}(\psi_{i,q}-E\psi_{i,q})\}^2.$ Note that the true variance can be expressed as:
    $$
    \sigma_n^2=\frac{1}{n}\sum_{i,j\in N_n, \ell_{A}(i,j)\le 2}E\left(\psi_{i,q}-E\psi_{i,q}\right)\left(\psi_{j,q}-E\psi_{j,q}\right),
    $$ 

Let
$\mathcal P_n=\{(i,j)\in N_n^2:\ell_A(i,j)\le 2\}$, and
$D_n=\max_i\#\{j:(i,j)\in\mathcal P_n\}$.
We have $D_n\lesssim 1+d_{mx}^{(2)}$.
Define
\[
\Delta_i=\widehat\psi_{i,q}-\psi_{i,q},
\qquad
\Delta^M=\frac1n\sum_{i\in N_n}\Delta_i.
\]
Then
\[
\widehat\psi_{i,q}-\widehat\psi_q^M
=(\psi_{i,q}-\psi_q^M)+(\Delta_i-\Delta^M).
\]
Define the oracle sample-centered quadratic form
\[
\widetilde\sigma_n^2
=\frac1n\sum_{(i,j)\in\mathcal P_n}
(\psi_{i,q}-\psi_q^M)(\psi_{j,q}-\psi_q^M).
\]
Then
\begin{equation}
\label{eq:var_err}
\widehat\sigma_n^2-\widetilde\sigma_n^2
=V_A+V_B+V_C,
\end{equation}
where
\[
V_A=\frac1n
\sum_{(i,j)\in\mathcal P_n}
(\Delta_i-\Delta^M)(\psi_{j,q}-\psi_q^M),
\]
\[
V_B=\frac1n
\sum_{(i,j)\in\mathcal P_n}
(\psi_{i,q}-\psi_q^M)(\Delta_j-\Delta^M),
\]
\[
V_C=
\frac1n\sum_{(i,j)\in\mathcal P_n}
(\Delta_i-\Delta^M)(\Delta_j-\Delta^M),
\]
Here $V_A,V_B,V_C$ are the plug-in terms. Since \(\mathcal P_n\) is symmetric,
$(i,j)\in\mathcal P_n \iff (j,i)\in\mathcal P_n$.
Therefore, relabeling the indices yields
$V_B=V_A$. Consequently, the decomposition in \eqref{eq:var_err} can be rewritten as
\[
\widehat\sigma_n^2-\widetilde\sigma_n^2
=2V_A+V_C.
\]
The next two lemmas show that the plug-in difference is $o_p(1)$.

\begin{lemma}  
\label{lem:VA}
Assume $D_n\left(\epsilon_n+n^{-1/2}\right)=o(1)$. Then
\[V_A=o_p(1).\]
\end{lemma}
\begin{proof}[Proof of Lemma~\ref{lem:VA}]
    Since \(\psi_{i,q}\) is uniformly bounded, we have
$\max_{j\in N_n}|\psi_{j,q}-\psi_q^M|=O_p(1)$.
Therefore,
\[
\begin{aligned}
|V_A|\le
\max_{j\in N_n}|\psi_{j,q}-\psi_q^M|\cdot
\frac1n\sum_{(i,j)\in\mathcal P_n}|\Delta_i-\Delta^M| \le O(D_n)\frac1n\sum_{i\in N_n}|\Delta_i-\Delta^M|.
\end{aligned}
\]
Since $|\Delta_i-\Delta^M|\le |\Delta_i|+|\Delta^M|$ and
$|\Delta^M|\le\frac1n\sum_{i\in N_n}|\Delta_i|$,
we have
\[
\frac1n\sum_{i\in N_n}|\Delta_i-\Delta^M|
\le
2\frac1n\sum_{i\in N_n}|\Delta_i|.
\]
Define
$\|\Delta\|_{1,n}:=\frac1n\sum_{i\in N_n}|\Delta_i|$,
then
\[
|V_A| \le O(D_n)\|\Delta\|_{1,n}.
\]

It remains to bound \(\|\Delta\|_{1,n}\). For \(s\in\{g,g'\}\), write
\[
W_i(s)=\frac{I(G_i=s)}{\pi_i(s)},
\qquad
a_s=f_s^{-1}\{\beta_{nq}(s)\},
\qquad
\widehat a_s=\widehat f_{ns}^{-1}\{\widehat\beta_{nq}(s)\},
\]
and
\[
H_{s,q}(Y_i)=a_s\{q-I(Y_i\le \beta_{nq}(s))\},
\qquad
\widehat H_{s,q}(Y_i)=\widehat a_s
\{q-I(Y_i\le \widehat\beta_{nq}(s))\}.
\]
Then
\[
\Delta_i=\widehat\psi_{i,q}-\psi_{i,q}
=W_i(g)\{\widehat H_{g,q}(Y_i)-H_{g,q}(Y_i)\}
-W_i(g')\{\widehat H_{g',q}(Y_i)-H_{g',q}(Y_i)\}.
\]
Write
\[
\begin{aligned}
\widehat H_{s,q}(Y_i)-H_{s,q}(Y_i)
=&
(\widehat a_s-a_s)\{q-I(Y_i\le \widehat\beta_{nq}(s))\}\\
&+
a_s\{I(Y_i\le \beta_{nq}(s))-I(Y_i\le \widehat\beta_{nq}(s))\}.\\
\end{aligned}
\]
Using boundedness of \(W_i(s)\), \(a_s\), and \(q-I(\cdot)\), we have
\[
\begin{aligned}
&\frac1n\sum_{i\in N_n}
|W_i(s)\{\widehat H_{s,q}(Y_i)-H_{s,q}(Y_i)\}|\\
&\le
c_0|\widehat a_s-a_s| 
+c_0\frac1n\sum_{i\in N_n}I\left(|Y_i-\beta_{nq}(s)|\le \rho_{n,s}\right),
\end{aligned}
\]
where $\rho_{n,g}=|\widehat\beta_{nq}(g)-\beta_{nq}(g)|$ and 
$c_0$ is some positive constant, 

Since $$
|\hat a_g-a_g|=\left|\frac{\hat f_{ng}(\hat\beta_{nq}(g))-f_g(\beta_{nq}(g))}{\hat f_{ng}(\hat\beta_{nq}(g))\cdot f_g(\beta_{nq}(g))}\right|,
$$ from Corollary~\ref{cor:fhat_unif_conv}, $|\hat f_{ng}(\hat\beta_{nq}(g))-f_g(\hat\beta_{nq}(g))|=o_p(\epsilon_n)$ from the convergence of density estimator $\hat f_{ng}$ in Assumption~\ref*{asp:error_pdf}. Meanwhile, from the differentiability of $f_g,$ $|f_{g}(\hat\beta_{nq}(g))-f_g(\beta_{nq}(g))|=O(|\hat\beta_{nq}(g)-\beta_{nq}(g)|)\le O_p(n^{-1/2}).$ Combining the results above we have $|\hat a_g-a_g|\le O_p(\epsilon_n+n^{-1/2}).$ Let
\[
M_{n,g}(r)=\frac1n\sum_{i\in N_n}\frac{I(G_i=g)}{\pi_i(g)}I\left(|Y_i-\beta_{nq}(g)|\le r\right).
\]
We will show $M_{n,g}(\rho_{n,g})=O_p(n^{-1/2})$.
Since $\sqrt{n}\rho_{n,g}$ is bounded in probability, for any fixed $\varepsilon>0$, there exists
a constant $C<\infty$ such that
\[
\limsup_{n\to\infty}P\left(\rho_{n,g}>Cn^{-1/2}\right)
\le\frac{\varepsilon}{2}.
\]
Because $M_{n,g}(r)$ is nondecreasing in $r$, on the event
$\mathcal E_{n,C}=
\left\{\rho_{n,g}\le Cn^{-1/2}\right\}$,
we have
\[
M_{n,g}(\rho_{n,g}) \le M_{n,g}(Cn^{-1/2}).
\]
Therefore, for any constant $K>0$,
\[
\begin{aligned}
P\left\{M_{n,g}(\rho_{n,g})>K n^{-1/2}\right\}
&\le P\left(\rho_{n,g}>Cn^{-1/2}\right) +P\left\{M_{n,g}(Cn^{-1/2})>K n^{-1/2}\right\}.
\end{aligned}
\]

By Assumption 3, for every fixed \(C<\infty\),
\[
E\{M_{n,g}(Cn^{-1/2})\}\leq C_0n^{-1/2},
\]
for some positive constant $C_0$, for $n$ sufficiently large. By Markov's inequality,
\[
\begin{aligned}
P\left\{M_{n,g}(Cn^{-1/2})>K n^{-1/2}\right\}
&\le
\frac{E\{M_{n,g}(Cn^{-1/2})\}}{K n^{-1/2}} 
&\le \frac{C_0}{K},
\end{aligned}
\]
Choose \(K\) large enough so that
$\frac{C_0}{K}\le \frac{\varepsilon}{2}$. Then
\[
\limsup_{n\to\infty}
P\left\{M_{n,g}(\rho_{n,g})>K n^{-1/2}\right\}
\le \varepsilon.
\]
Since \(\varepsilon>0\) is arbitrary, we conclude that
$M_{n,g}(\rho_{n,g})=O_p(n^{-1/2})$.

Consequently,
\[
\frac1n\sum_{i\in N_n}
|W_i(s)\{\widehat H_{s,q}(Y_i)-H_{s,q}(Y_i)\}|
=O_p(\epsilon_n+n^{-1/2}).
\]

Applying this bound to \(s=g\) and \(s=g'\), we obtain
\[
\|\Delta\|_{1,n}=O_p\left(\epsilon_n+n^{-1/2}\right).
\]
Therefore,
\[
V_A=O_p\left[D_n\left(\epsilon_n+n^{-1/2}\right)\right]=o_p(1).
\]
\end{proof} 

\begin{lemma}  Assume $D_n\left(\epsilon_n+n^{-1/2}\right)=o(1)$. Then
\[V_C=o_p(1).\]
\label{lem:VC}
\end{lemma}
\begin{proof}[Proof of Lemma~\ref{lem:VC}]
    Note that from AM-GM inequality,   
\[
\begin{aligned}
    |V_C|\le& \frac{1}{2n}\sum_{(i,j)\in\mathcal P_n}\left[(\Delta_i-\Delta^M)^2+(\Delta_j-\Delta^M)^2\right]\\
    \le&\frac{D_n}{n}\sum_{i\in N_n}(\Delta_i-\Delta^M)^2.\\
\end{aligned}
\]
We first show $\sup_{i\in N_n}|\Delta_i-\Delta^M|=O_p(1)$. Since
$|\Delta^M|
\le \frac1n\sum_{i\in N_n}|\Delta_i|
\le \sup_{i\in N_n}|\Delta_i|$,
it is enough to prove
$\sup_{i\in N_n}|\Delta_i|=O_p(1)$.
By definition,
\[
\sup_{i\in N_n}|\Delta_i|\le\sup_{i\in N_n}|\hat\psi_{i,q}(Y_i,g)-\psi_{i,q}(Y_i,g)|+\sup_{i\in N_n}|\hat\psi_{i,q}(Y_i,g')-\psi_{i,q}(Y_i,g')|.
\]
Following the notations and similar arguments as in 
the proof of Lemma~\ref{lem:VA}, 
\begin{equation}
\label{eq:diff_psihat_psi_new}
    \begin{aligned}
    &\sup_{i\in N_n}|\hat\psi_{i,q}(Y_i,g)-\psi_{i,q}(Y_i,g)|\\
    \le{}& \sup_{i\in N_n}|W_i(g)(\hat a_g-a_g)
    \{q-I(Y_i\le\hat\beta_{nq}(g))\}|\\
    &+\sup_{i\in N_n}|W_i(g)a_g
    \{I(Y_i\le\beta_{nq}(g))-I(Y_i\le\hat\beta_{nq}(g))\}|\\
    =& O(|\hat a_g-a_g|)+O_p(1),
    \end{aligned}
\end{equation}
where we used the uniform boundedness of \(W_i(g)\), \(a_g\), and the indicator functions.
Since $|\hat a_g-a_g|\le O_p(\epsilon_n+n^{-1/2})$ as shown in the derivation of 
Lemma~\ref{lem:VA}, we obtain that $\sup_{i\in N_n}|\hat\psi_{i,q}(Y_i,g)-\psi_{i,q}(Y_i,g)|\le O_p(1)$.
The same argument applies to $g'$. Hence
$$\sup_{i\in N_n}|\Delta_i|= O_p(1).$$

Consequently, we have
\[
V_C
\le O_p(1)\frac{D_n}{n}\sum_{i\in N_n}|\Delta_i-\Delta^M|
\le 
O_p(1)\frac{2D_n}{n}\sum_{i\in N_n}|\Delta_i|=O_p(D_n\|\Delta\|_{1,n}).
\]
It's derived in the proof of Lemma~\ref{lem:VA} that 
$\|\Delta\|_{1,n}=O_p\left(\epsilon_n+n^{-1/2}\right)$. 
Therefore,
\[
V_C=O_p\left[D_n\left(\epsilon_n+n^{-1/2}\right)\right]=o_p(1).
\]
\end{proof} 

Now return to the proof of Theorem~\ref{thm:var_est}. Let
\[
\mu=E\psi,\qquad
\xi=\psi-\mu,\qquad
P_n=I_n-n^{-1}\mathbf1\mathbf1^\top,
\qquad
H_n=(I\{(i,j)\in\mathcal P_n\})_{i,j\in N_n},
\]
and define
\[
m=P_n\mu,\qquad v=P_n\psi=m+P_n\xi,
\qquad \bar\xi=n^{-1}\mathbf1^\top\xi.
\]
Because $v_i=\psi_{i,q}-\psi_q^M$,
$\widetilde\sigma_n^2=n^{-1}v^\top H_nv$.  Moreover,
\[
\sigma_n^2=\frac1nE(\xi^\top H_n\xi),
\qquad
\operatorname{Bias}_n
=\frac1n m^\top H_nm
=\frac1n m^\top(P_nH_nP_n)m.
\]
Then we have
\begin{equation}
\label{eq:oracle_var_decomp}
\widetilde\sigma_n^2-\sigma_n^2-\operatorname{Bias}_n
=T_{1n}+T_{2n}+T_{3n},
\end{equation}
where
\[
T_{1n}=\frac2n m^\top H_nP_n\xi,
\quad
T_{2n}=\frac1n\{\xi^\top H_n\xi-E(\xi^\top H_n\xi)\},
\quad
T_{3n}=\frac1n\xi^\top(P_nH_nP_n-H_n)\xi.
\]

We first control $T_{2n}$, whose summands depend only on local scores. Let
\[
\boldsymbol S_n=\left\{(i,j,k,l)\in N_n^4:
(i,j),(k,l)\in\mathcal P_n,\
\bar\nu_n(\{i,j\})\cap\bar\nu_n(\{k,l\})\neq\emptyset\right\}.
\]
Pair products indexed by quadruples outside $\boldsymbol S_n$ are independent.
The scores are uniformly bounded, the number of ordered pairs in
$\mathcal P_n$ is $n(1+d_{av}^{(2)})$, and each pair has at most
$O\{(1+d_{mx}^{(2)})^2\}$ dependent pairs.  Therefore,
\[
ET_{2n}^2
\le C\frac{|\boldsymbol S_n|}{n^2}
=O\left\{
\frac{(1+d_{av}^{(2)})(1+d_{mx}^{(2)})^2}{n}
\right\}=o(1),
\]
so $T_{2n}=o_p(1)$.

Next let $\Gamma_n=E(\xi\xi^\top)$.  Its absolute row sums are bounded by
$C(1+d_{mx}^{(2)})$.  Since $m$ is uniformly bounded, writing
$r_i=\sum_j(H_n)_{ij}$ gives
\[
\|H_nm\|_2^2
\le C\sum_i r_i^2
\le Cn(1+d_{av}^{(2)})(1+d_{mx}^{(2)}).
\]
Because $P_n$ is an orthogonal projection,
\[
\begin{aligned}
ET_{1n}^2
&=\frac4{n^2}(P_nH_nm)^\top\Gamma_n(P_nH_nm)\\
&\le
C\frac{(1+d_{av}^{(2)})(1+d_{mx}^{(2)})^2}{n}
=o(1).
\end{aligned}
\]
Hence $T_{1n}=o_p(1)$.

Finally, write $r=H_n\mathbf1$ and
$L_n=\mathbf1^\top H_n\mathbf1=n(1+d_{av}^{(2)})$.  Direct expansion yields
\[
T_{3n}=-\frac2n\bar\xi\,r^\top\xi
+\frac{L_n}{n}\bar\xi^2.
\]
Uniform boundedness gives $|r^\top\xi|\le C L_n$.  Moreover, boundedness of
the scores implies
\[
\begin{split}
E(\bar\xi^2)
&=\frac1{n^2}\mathbf1^\top\Gamma_n\mathbf1\\
&\le C\frac{1+d_{av}^{(2)}}{n}.
\end{split}
\]
Thus
$\bar\xi=O_p[\{(1+d_{av}^{(2)})/n\}^{1/2}]$, and
\[
\begin{split}
|T_{3n}|
&\le C(1+d_{av}^{(2)})(|\bar\xi|+\bar\xi^2)\\
&=O_p\left[
\left\{\frac{(1+d_{av}^{(2)})^3}{n}\right\}^{1/2}
+\frac{(1+d_{av}^{(2)})^2}{n}
\right]
=o_p(1).
\end{split}
\]
The final equality follows directly from the first condition in
Assumption~\ref*{asp:degree_rate_var} and
$d_{av}^{(2)}\le d_{mx}^{(2)}$.

The plug-in lemmas and the second condition in
Assumption~\ref*{asp:degree_rate_var} give
$\widehat\sigma_n^2-\widetilde\sigma_n^2=o_p(1)$.  Combining this result with
\eqref{eq:oracle_var_decomp} proves
\[
\widehat\sigma_n^2
=\sigma_n^2+\operatorname{Bias}_n+o_p(1),
\qquad
\operatorname{Bias}_n
=\frac1n\sum_{(i,j)\in\mathcal P_n}
(E\psi_{i,q}-E\psi_q^M)(E\psi_{j,q}-E\psi_q^M).
\]
\hfill $\Box$

\subsection{Proof of Theorem~\ref*{thm:cons_var}}
\begin{proof}
Fix the exposure contrast $(g,g')$.  For $s\in\{g,g'\}$, recall the
oracle and plug-in exposure-specific scores
\[
\begin{aligned}
\psi_{i,q}(Y_i,s)
&=
\frac{I(G_i=s)}{\pi_i(s)f_s\{\beta_{nq}(s)\}}
\left[q-I\{Y_i\le\beta_{nq}(s)\}\right],\\
\widehat\psi_{i,q}(Y_i,s)
&=
\frac{I(G_i=s)}
{\pi_i(s)\widehat f_{ns}\{\widehat\beta_{nq}(s)\}}
\left[q-I\{Y_i\le\widehat\beta_{nq}(s)\}\right].
\end{aligned}
\]
Their contrast scores are
\[
\psi_{i,q}
=\psi_{i,q}(Y_i,g)-\psi_{i,q}(Y_i,g'),
\qquad
\widehat\psi_{i,q}
=\widehat\psi_{i,q}(Y_i,g)-\widehat\psi_{i,q}(Y_i,g').
\]
Let
\[
\psi=(\psi_{1,q},\ldots,\psi_{n,q})^\top,
\qquad
\widehat\psi=(\widehat\psi_{1,q},\ldots,
\widehat\psi_{n,q})^\top,
\qquad
\mu=E\psi=(E\psi_{1,q},\ldots,E\psi_{n,q})^\top.
\]
Also recall
\[
H_n=\bigl(I\{\ell_A(i,j)\le2\}\bigr)_{i,j\in N_n},
\qquad
P_n=I_n-\frac1n\mathbf1\mathbf1^\top,
\qquad
C_n=P_nH_nP_n,
\]
and let $K_n^{-}=-C_n^{-}\succeq0$ and
$C_n^{+}=C_n+K_n^{-}\succeq0$, where $C_n^{-}$ is the negative spectral
part of $C_n$. The conservative estimator is
\[
\widehat\sigma_{n,\mathrm{cons}}^2
=\frac1n\widehat v^\top C_n^{+}\widehat v,
\qquad
\widehat v=P_n\widehat\psi.
\]
For the oracle scores, write
\[
v=P_n\psi,
\qquad
m=P_n\mu,
\qquad
u=P_n(\psi-\mu),
\]
so that $v=m+u$. Finally, let
\[
\|x\|_{1,n}=\frac1n\sum_{i\in N_n}|x_i|.
\]
Let $\Delta=\widehat\psi-\psi$. The plug-in bounds in the proof of
Theorem~\ref*{thm:var_est}, in particular Lemmas~\ref{lem:VA} and
\ref{lem:VC}, imply
\begin{equation}
\label{eq:spectral_plugin_rate}
\|\widehat v-v\|_{1,n}
=\|P_n\Delta\|_{1,n}
\le2\|\Delta\|_{1,n}
=O_p(\epsilon_n+n^{-1/2}),
\qquad
\|\widehat v\|_\infty+\|v\|_\infty=O_p(1).
\end{equation}
Because $K_n^{-}$ is symmetric and
$\rho_n^{-}=\max_i\sum_j|(K_n^{-})_{ij}|$, it follows that
\begin{equation}
\label{eq:spectral_plugin_quadratic}
\begin{split}
&\frac1n
\left|
\widehat v^\top K_n^{-}\widehat v
-v^\top K_n^{-}v
\right|\\
&\qquad\le
\rho_n^{-}\|\widehat v-v\|_{1,n}
\left(\|\widehat v\|_\infty+\|v\|_\infty\right)
=O_p\{\rho_n^{-}(\epsilon_n+n^{-1/2})\}
=o_p(1).
\end{split}
\end{equation}

Recall that $C_n=P_nH_nP_n$, $C_n^{+}=C_n+K_n^{-}$, and
\[
\operatorname{Bias}_n=\frac1n m^\top C_nm.
\]
By Theorem~\ref*{thm:var_est}, (\ref*{eq:plug_in_cons_var}), and
\eqref{eq:spectral_plugin_quadratic},
\begin{equation}
\label{eq:spectral_cons_decomp}
\begin{split}
\widehat\sigma_{n,\mathrm{cons}}^2
={}&
\sigma_n^2
+\frac1n m^\top C_nm
+\frac1n v^\top K_n^{-}v
+o_p(1)\\
={}&
\sigma_n^2
+\frac1n m^\top C_n^{+}m
+\frac1n u^\top K_n^{-}u
+\frac2n m^\top K_n^{-}u
+o_p(1).
\end{split}
\end{equation}
The first two quadratic forms after $\sigma_n^2$ are nonnegative because
$C_n^{+}$ and $K_n^{-}$ are positive semidefinite.

It remains to control the cross term.  Since
$C_n\mathbf1=P_nH_nP_n\mathbf1=0$, the spectral definition of $K_n^{-}$
implies $K_n^{-}\mathbf1=0$.  Hence
$m^\top K_n^{-}u=m^\top K_n^{-}(\psi-\mu)$. Let
$a=K_n^{-}m$ and let $\Gamma_n$ be the covariance matrix of
$\psi-\mu$. Scores more than two steps apart are independent and the scores
are uniformly bounded, so every row of $\Gamma_n$ has absolute row sum at
most $C(1+d_{mx}^{(2)})$. Also,
$\|K_n^{-}\|_{\mathrm{op}}\le\rho_n^{-}$ and $\|m\|_2^2\le Cn$.
Consequently,
\begin{equation}
\label{eq:spectral_cross_bound}
\begin{split}
E\left(\frac1n m^\top K_n^{-}u\right)^2
&=
\frac1{n^2}a^\top\Gamma_na\\
&\le
C\frac{1+d_{mx}^{(2)}}{n^2}\|a\|_2^2\\
&\le
C(\rho_n^{-})^2\frac{1+d_{mx}^{(2)}}{n}
=o(1)
\end{split}
\end{equation}
by condition (\ref*{eq:cons_rate_condition}). Thus the cross term in
\eqref{eq:spectral_cons_decomp} is $o_p(1)$.  Define
\[
\mathcal V_n^{+}
=\frac1n m^\top C_n^{+}m+\frac1n u^\top K_n^{-}u\ge0.
\]
Then \eqref{eq:spectral_cons_decomp} can be written as
\[
\widehat\sigma_{n,\mathrm{cons}}^2
=\sigma_n^2+\mathcal V_n^{+}+r_n,
\qquad r_n=o_p(1),
\]
Therefore,
\[
\sigma_n^2
\le
\widehat\sigma_{n,\mathrm{cons}}^2+o_p(1).
\]

Since $\inf_n\sigma_n^2>0$, write
$c_\sigma=\inf_n\sigma_n^2$.  The preceding additive lower bound implies
that, for every fixed $\delta\in(0,1)$,
\[
\Pr\left\{
\widehat\sigma_{n,\mathrm{cons}}^2
<\sigma_n^2-c_\sigma\delta
\right\}
\le\Pr\{r_n<-c_\sigma\delta\}
\longrightarrow0.
\]
To make the coverage implication explicit, let
\[
T_n=\frac{\sqrt n\{\widehat\tau_{nq}(g,g')-
\tau_{nq}(g,g')\}}{\sigma_n}.
\]
Because
$\sigma_n^2-c_\sigma\delta\ge(1-\delta)\sigma_n^2$, we have
\[
\begin{aligned}
&\Pr\{\tau_{nq}(g,g')\in
\mathrm{CI}_{1-\alpha}^{\mathrm{cons}}(g,g')\}
&\ge
\Pr\{|T_n|\le z_{1-\alpha/2}\sqrt{1-\delta}\}
-\Pr\{\widehat\sigma_{n,\mathrm{cons}}^2
<\sigma_n^2-c_\sigma\delta\}.
\end{aligned}
\]
The second probability tends to zero, while
$T_n\Rightarrow N(0,1)$ by whichever of Theorems~\ref*{thm:asy_nor}
and~\ref*{thm:asy-nor-janson} supplies the assumed CLT. Taking the lower limit
and then letting $\delta\downarrow0$ gives coverage at least $1-\alpha$.
\end{proof}

\section{Additional Simulation Details and Results}
\label{sec:supplementary_simulations}

\subsection{Additional Results on Random Graph under Linear Design}
\label{sec:random_graph_linear_details}

This subsection supplements the random graph linear simulation reported in
Table~\ref*{tab:random_graph_linear} of the main paper.  Table~\ref{tab:random_graph_linear_diagnostics}
reports bias, empirical standard deviation, and mean standard error for the
quantile effects discussed in the main paper, while
Table~\ref{tab:random_graph_linear_spillover} reports interval coverage and
width for the spillover effect.  Specifically, we consider
\[
 \mathrm{QSE}_{0}(q)
 :=\tau_{nq}\bigl((0,3/4),(0,1/4)\bigr),
\]
which compares neighborhoods with three and one treated members while holding
own treatment at zero.  Under the linear outcome model in the main paper,
$\mathrm{QSE}_{0}(q)=25(3/4-1/4)=12.5$ at every reported quantile.

\begin{table}[ht]
\centering
\caption{Bias, empirical standard deviation, and mean standard error for the
three proposed estimands in the random-graph linear design.}
\label{tab:random_graph_linear_diagnostics}
\setlength{\tabcolsep}{6pt}
\renewcommand{\arraystretch}{1.02}
\small
\begin{tabular}{lllrrrr}
\toprule
Estimand & $q$ & Distribution & Target & Bias & Emp. SD & Mean SE \\
\midrule
\multirow{9}{*}{$\mathrm{QDE}_{1/2}$} & \multirow{3}{*}{.2}
& Gaussian & 5.000 & -0.005 & 0.375 & 0.396 \\
& & Student-t & 5.000 & -0.016 & 0.384 & 0.393 \\
& & Slash & 5.000 & -0.025 & 0.379 & 0.389 \\
\addlinespace
& \multirow{3}{*}{.5} & Gaussian & 5.000 & 0.003 & 0.345 & 0.356 \\
& & Student-t & 5.000 & -0.017 & 0.333 & 0.354 \\
& & Slash & 5.000 & -0.005 & 0.345 & 0.349 \\
\addlinespace
& \multirow{3}{*}{.8} & Gaussian & 5.000 & -0.011 & 0.398 & 0.389 \\
& & Student-t & 5.000 & -0.006 & 0.369 & 0.390 \\
& & Slash & 5.000 & 0.011 & 0.393 & 0.384 \\
\addlinespace[3pt]
\multirow{9}{*}{$\mathrm{QEE}(1,0)$} & \multirow{3}{*}{.2}
& Gaussian & 5.000 & -0.086 & 0.691 & 0.761 \\
& & Student-t & 5.000 & -0.024 & 0.729 & 0.760 \\
& & Slash & 5.000 & -0.016 & 0.719 & 0.762 \\
\addlinespace
& \multirow{3}{*}{.5} & Gaussian & 5.000 & -0.091 & 0.635 & 0.730 \\
& & Student-t & 5.000 & -0.017 & 0.675 & 0.723 \\
& & Slash & 5.000 & 0.013 & 0.652 & 0.726 \\
\addlinespace
& \multirow{3}{*}{.8} & Gaussian & 5.000 & -0.087 & 0.711 & 0.757 \\
& & Student-t & 5.000 & 0.025 & 0.749 & 0.755 \\
& & Slash & 5.000 & -0.003 & 0.715 & 0.757 \\
\addlinespace[3pt]
\multirow{9}{*}{$\mathrm{QSE}_{0}$} & \multirow{3}{*}{.2}
& Gaussian & 12.500 & 0.018 & 0.511 & 0.486 \\
& & Student-t & 12.500 & 0.026 & 0.477 & 0.484 \\
& & Slash & 12.500 & 0.001 & 0.480 & 0.477 \\
\addlinespace
& \multirow{3}{*}{.5} & Gaussian & 12.500 & -0.020 & 0.394 & 0.436 \\
& & Student-t & 12.500 & 0.026 & 0.412 & 0.433 \\
& & Slash & 12.500 & -0.015 & 0.417 & 0.433 \\
\addlinespace
& \multirow{3}{*}{.8} & Gaussian & 12.500 & -0.020 & 0.463 & 0.476 \\
& & Student-t & 12.500 & 0.037 & 0.488 & 0.477 \\
& & Slash & 12.500 & -0.024 & 0.479 & 0.468 \\
\addlinespace[2pt]
\bottomrule
\end{tabular}
\end{table}

\begin{table}[ht]
\centering
\caption{Coverage and average interval width for the quantile spillover effect
in the random-graph linear design.}
\label{tab:random_graph_linear_spillover}
\setlength{\tabcolsep}{6pt}
\renewcommand{\arraystretch}{1.02}
\small
\begin{tabular}{llrrrr}
\toprule
& & \multicolumn{4}{c}{$\mathrm{QSE}_{0}$} \\
\cmidrule(lr){3-6}
& & \multicolumn{2}{c}{Plug-in} & \multicolumn{2}{c}{Conservative} \\
\cmidrule(lr){3-4} \cmidrule(lr){5-6}
Distribution & $q$ & Coverage & Width & Coverage & Width \\
\midrule
\multirow{3}{*}{Gaussian} & .2 & 88.2 & 1.599 & 96.6 & 2.331 \\
& .5 & 91.8 & 1.433 & 99.2 & 2.091 \\
& .8 & 90.0 & 1.567 & 99.2 & 2.291 \\
\addlinespace
\multirow{3}{*}{Student-t} & .2 & 89.8 & 1.592 & 97.8 & 2.325 \\
& .5 & 90.0 & 1.425 & 98.0 & 2.084 \\
& .8 & 89.8 & 1.570 & 98.8 & 2.289 \\
\addlinespace
\multirow{3}{*}{Slash} & .2 & 90.0 & 1.570 & 99.2 & 2.287 \\
& .5 & 91.2 & 1.424 & 98.8 & 2.071 \\
& .8 & 87.8 & 1.541 & 98.2 & 2.251 \\
\addlinespace[2pt]
\bottomrule
\end{tabular}
\end{table}

For $\mathrm{QSE}_{0}$, the largest absolute bias is $0.037$, and plug-in
coverage ranges from $87.8\%$ to $91.8\%$.  The conservative intervals raise
coverage to $96.6\%$--$99.2\%$, providing the intended protection at the cost
of greater width.

\subsection{Variance Comparison on a Non-community Network}
\label{sec:noncommunity_variance_comparison}

We supplement the main simulation with the variance comparison retained from
the preceding version of the paper.  For each even degree
$d\in\{2,4,6,8\}$, the $n=1000$ units form a connected ring lattice in which
unit $i$ is adjacent to the $d/2$ nearest units on either side.  This is one
connected network without an imposed community partition.  Treatments are
independently assigned as $W_i\sim\operatorname{Bernoulli}(0.08)$, and the
binary exposure is
\[
 G_i=I\!\left(W_i+\sum_{j\ne i}A_{ij}W_j>0\right).
\]
Consequently, $\Pr(G_i=0)=0.92^{d+1}$ and
$\Pr(G_i=1)=1-0.92^{d+1}$.

To generate a locally correlated baseline, start with independent standard
normal draws $B_i^{(0)}$ on the ring and apply three smoothing steps,
\[
 B_i^{(r)}=\tfrac12B_i^{(r-1)}
 +\tfrac14B_{i-1}^{(r-1)}+\tfrac14B_{i+1}^{(r-1)},
 \qquad r=1,2,3,
\]
with indices interpreted modulo $n$.  The result is standardized to empirical
mean zero and variance one and denoted by $B_i$.  Outcomes follow
\[
 Y_i(\bm W)=B_i+3G_i,
\]
so the QEE comparing $G_i=1$ with $G_i=0$ equals three at every quantile.  We
use 500 replications and $q\in\{0.1,0.2,\ldots,0.9\}$.  The QEE plug-in and
conservative estimates use the Epanechnikov kernel as in the main paper.  The empirical QEE curve is the Monte Carlo
variance of the QEE estimation error.  The QTE curve is the conventional
i.i.d. variance estimate accompanying the marginal treatment-based estimator following the method in~\citet{firpo2007efficient}.

\begin{figure}[ht]
\centering
\includegraphics[width=0.485\textwidth]{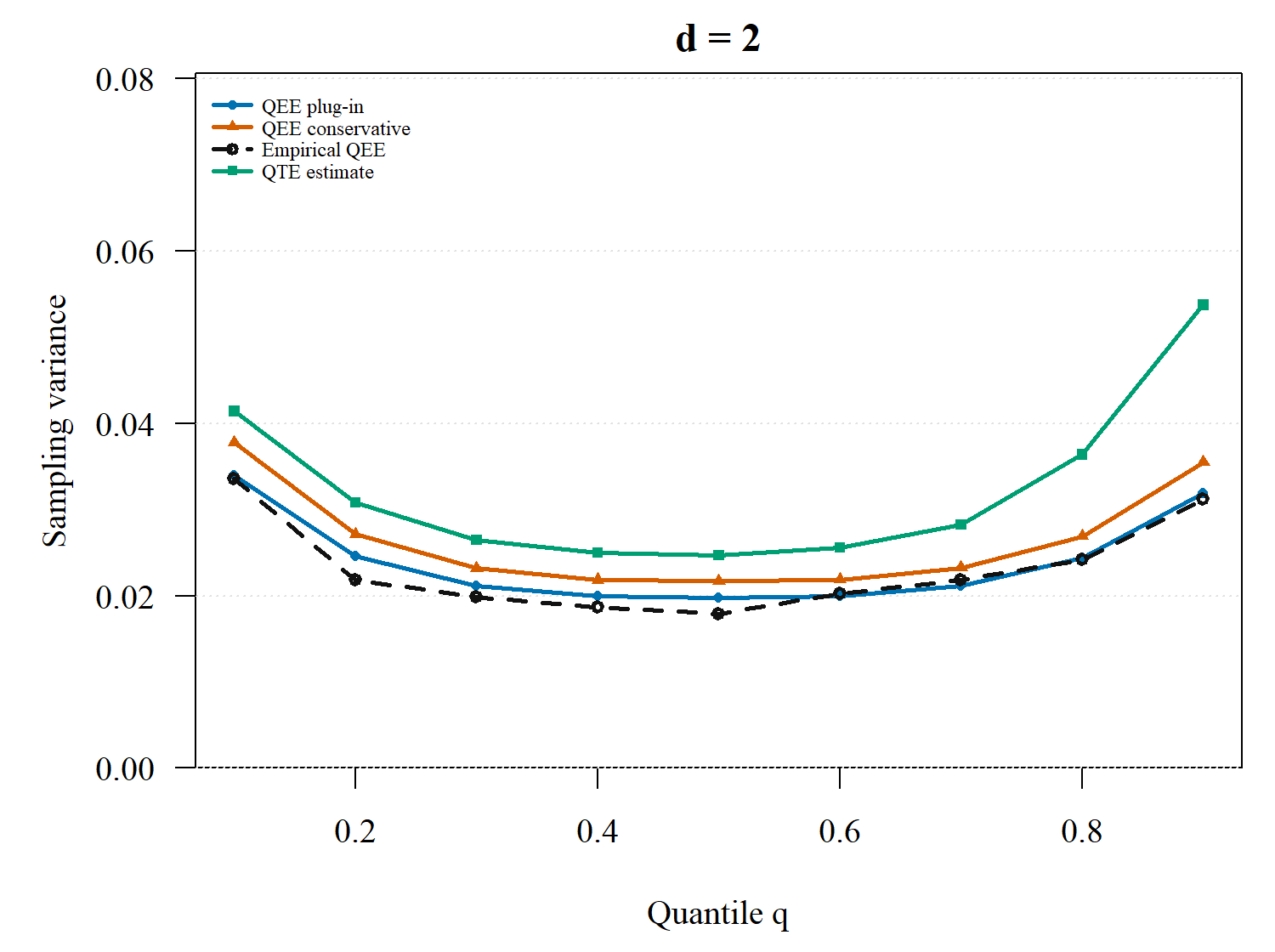}\hfill
\includegraphics[width=0.485\textwidth]{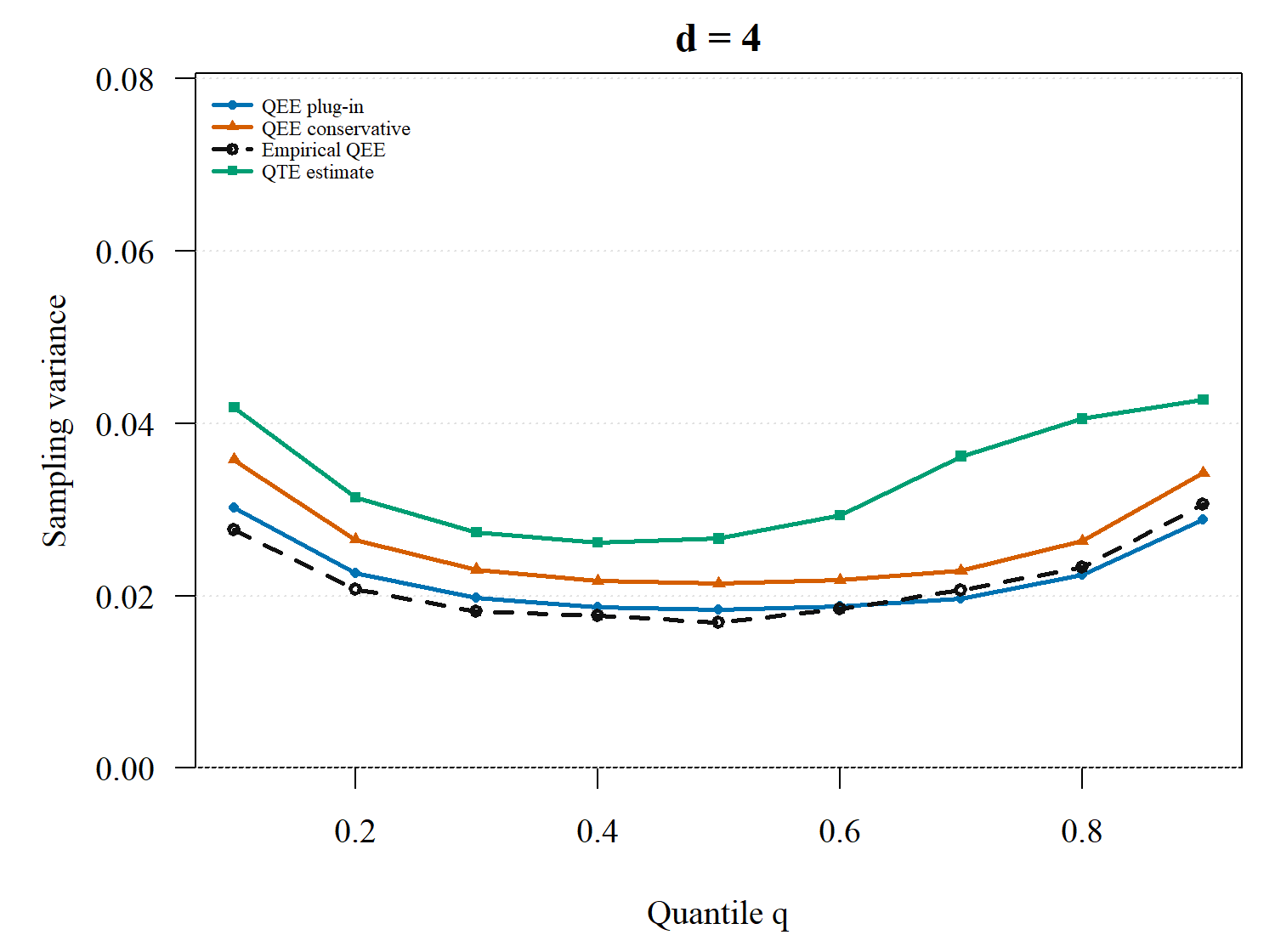}\\[2pt]
\includegraphics[width=0.485\textwidth]{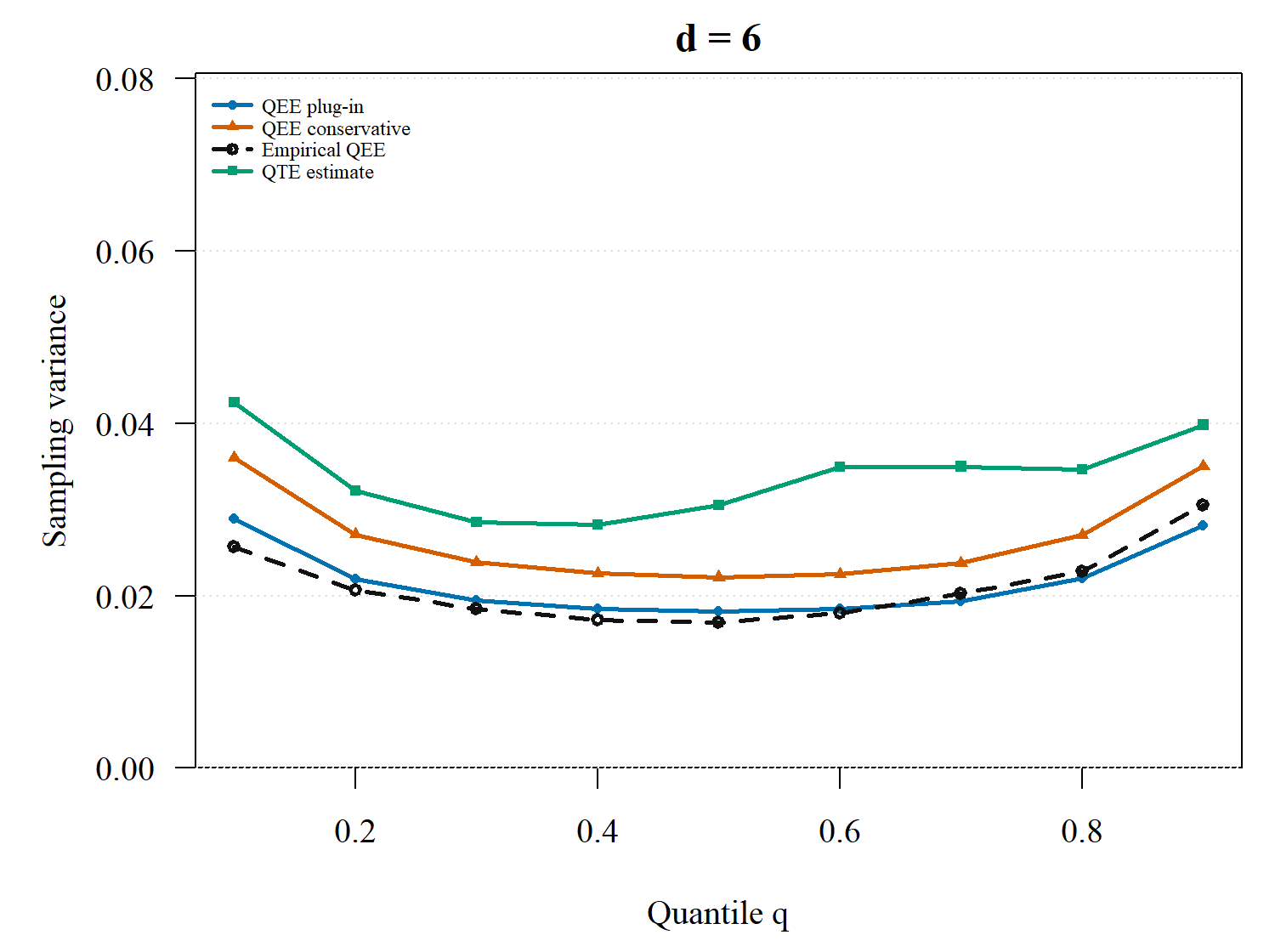}\hfill
\includegraphics[width=0.485\textwidth]{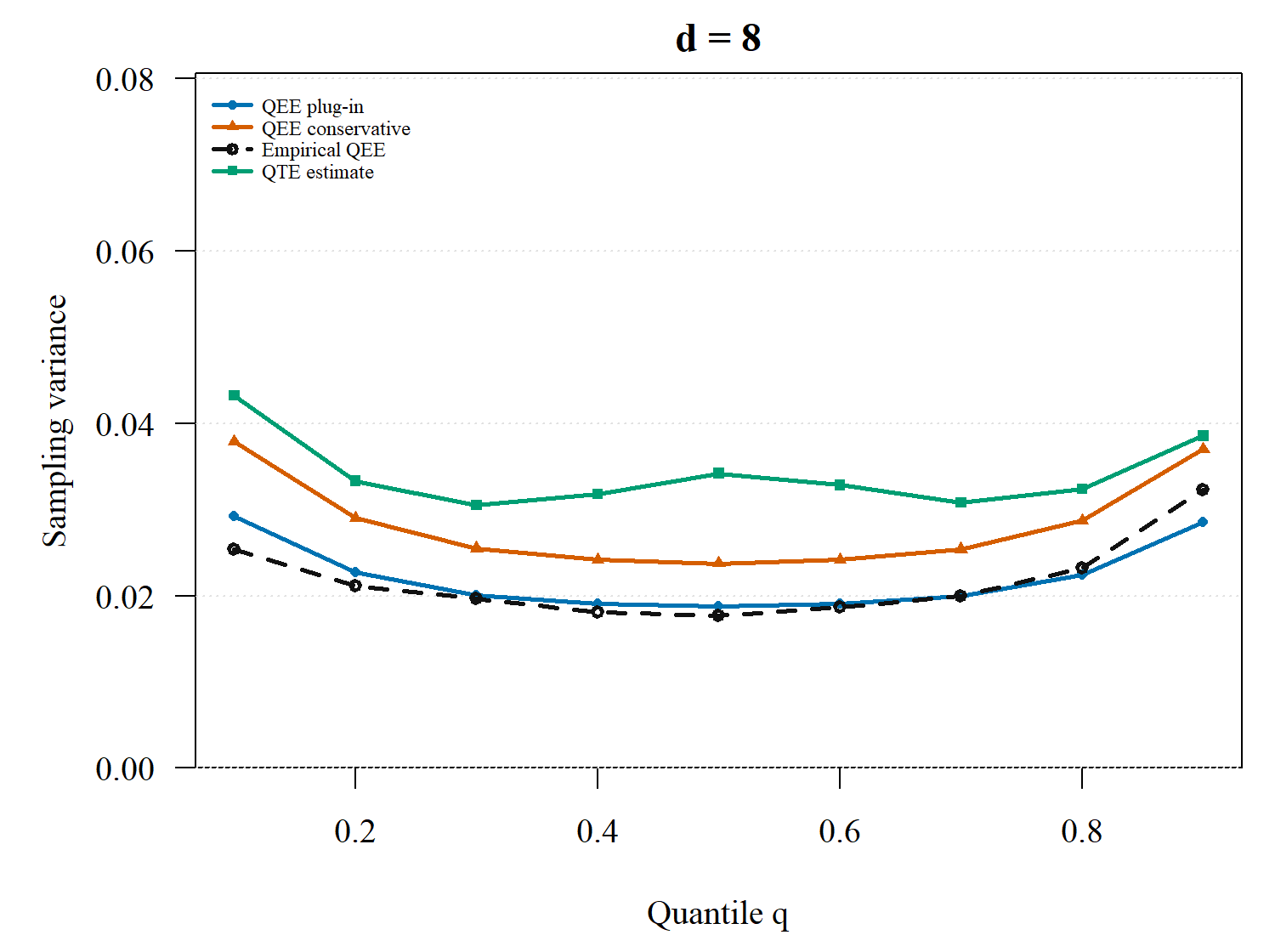}
\caption{Variance estimates over $q\in\{0.1,\ldots,0.9\}$, separately for
degrees $d=2,4,6,$ and $8$.  The empirical curve refers only to QEE.}
\label{fig:noncommunity_variance_quantiles_full}
\end{figure}

\begin{figure}[ht]
\centering
\includegraphics[width=0.485\textwidth]{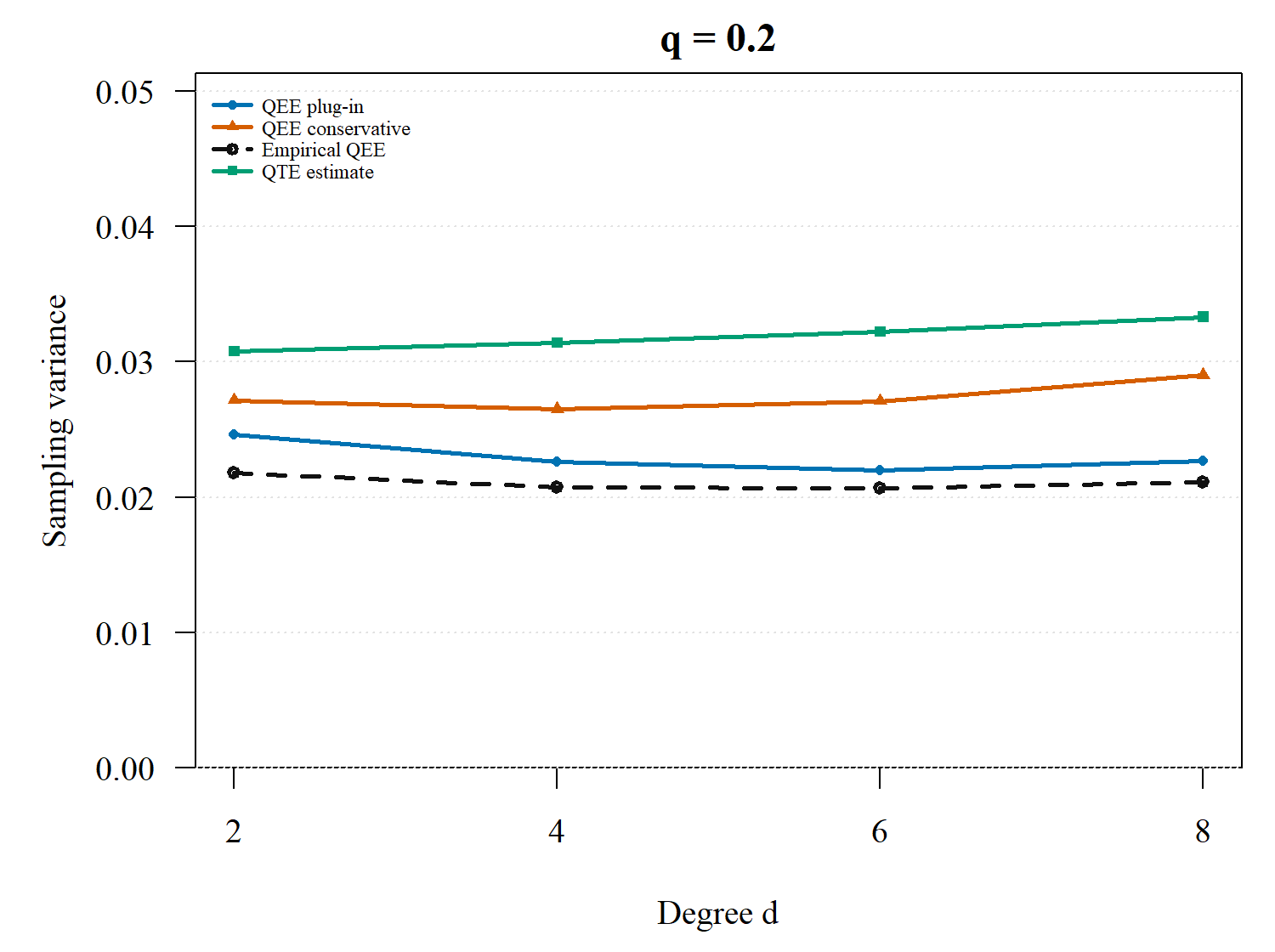}\hfill
\includegraphics[width=0.485\textwidth]{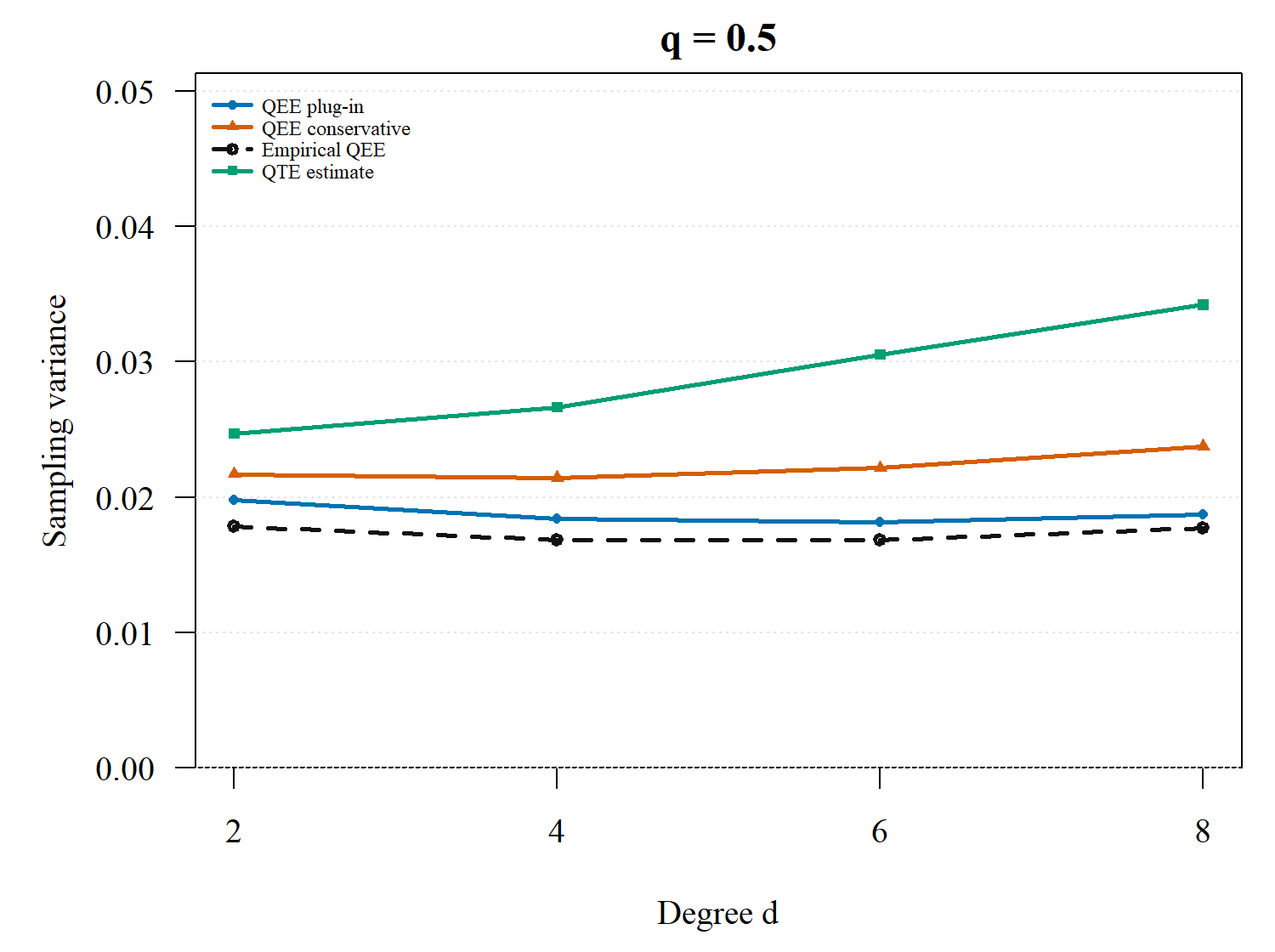}
\caption{Variance estimates over degrees $d\in\{2,4,6,8\}$ at
$q=0.2$ (left) and $q=0.5$ (right).  The empirical curve refers only to QEE.}
\label{fig:noncommunity_variance_degrees_full}
\end{figure}

Figures~\ref{fig:noncommunity_variance_quantiles_full}
and~\ref{fig:noncommunity_variance_degrees_full} show that the plug-in estimate
generally follows the empirical QEE variance over both varying quantiles and degrees.  The QTE variance estimates remain above the conservative estimate and the plug-in estimate, suggesting the necessity of taking account of neighboring correlation when interference exists.

\subsection{Nonlinear Interaction in a Community Network}
\label{sec:nonlinear_community}

We next consider a nonlinear interaction model on a community network.  The
$n=1000$ units are partitioned into 100 communities of size ten.  Within each
community, the units form a degree-four ring, with every unit connected to its
two nearest neighbors on either side, and different communities are disconnected.
The randomly labeled network is generated once and held fixed across the 500
replications. We independently draw $X_i\sim N(0,I_6)$ and assign treatments
independently as $W_i\sim\operatorname{Bernoulli}(1/2)$.
Consider the exposure mapping 
\[
 G_{i1}=W_i,\qquad
 G_{i2}=\frac14\sum_{j\ne i}A_{ij}W_j.
\]
Outcomes are generated from
\[
 Y_i(\bm W)
 =X_i^\top\boldsymbol 1_6+5G_{i1}+15G_{i2}
 +12G_{i1}G_{i2}^2+\epsilon_i,
 \qquad \epsilon_i\sim N(0,1).
\]
The quadratic interaction makes the own-treatment effect vary with
neighborhood exposure, and therefore QEE can identify the effect that can not be interpreted by QTE. However, to make a fair comparison, we focus on the exposure mapping $G_i=W_i$ and study
\[
 \mathrm{QEE}(q)=\tau_{nq}(1,0).
\]
Similar as before, this effect compares own treatment one with zero after integrating over the
neighbors' assignments. Let $Y_i(w,k)$ be the potential outcome when unit $i$ has
own treatment $w$ and $k$ of its four neighbors are treated, and define
\[
 F_{nw}(y)=\frac1n\sum_{i=1}^n\sum_{k=0}^4
 p_{ik}I\{Y_i(w,k)\le y\}.\]
 The true values of both QTE and QEE are computed by\[
 \tau_{nq}(1,0)=F_{n1}^{-1}(q)-F_{n0}^{-1}(q),
\]
where
\[
 p_{ik}=\Pr\!\left(\sum_{j\in\nu_n(i)}W_j=k\right)
 =\binom4k2^{-4}.
\]
This is because under independent assignment, the exposure-specific mixture distribution for
$G_i=w$ is exactly $F_{nw}$. Consequently, the QEE and QTE procedures in this
experiment estimate the same own-treatment quantile contrast, and differ only in how they estimate it and account for network dependence. The QEE procedure uses the exact propensity
$\pi_i(1)=\pi_i(0)=1/2$, the Epanechnikov density estimator, and the plug-in
and conservative variance estimators in (\ref*{eq:var_est}) and (\ref*{eq:plug_in_cons_var}) of the main paper.  The QTE procedure uses
the semiparametric IPW estimator and conventional i.i.d. variance formula of
\citet{firpo2007efficient}.

\begin{table}[ht]
\centering
\caption{Coverage and average interval width for the common own-treatment
quantile effect in the nonlinear community-network design.}
\label{tab:nonlinear_community}
\setlength{\tabcolsep}{3.5pt}
\renewcommand{\arraystretch}{1.04}
\small
\begin{tabular}{lrrrrrr}
\toprule
& \multicolumn{4}{c}{$\mathrm{QEE}(1,0)$}
& \multicolumn{2}{c}{QTE} \\
\cmidrule(lr){2-5} \cmidrule(l){6-7}
& \multicolumn{2}{c}{Plug-in} & \multicolumn{2}{c}{Conservative}
& \multicolumn{2}{c}{\citet{firpo2007efficient}} \\
\cmidrule(lr){2-3} \cmidrule(lr){4-5} \cmidrule(l){6-7}
$q$ & Coverage & Width & Coverage & Width & Coverage & Width \\
\midrule
.2 & 92.4 & 1.995 & 94.8 & 2.238 & 83.4 & 1.611 \\
.5 & 91.8 & 2.054 & 94.6 & 2.261 & 82.0 & 1.556 \\
.8 & 91.8 & 2.579 & 94.2 & 2.804 & 81.2 & 1.883 \\
\addlinespace[2pt]
\bottomrule
\end{tabular}
\end{table}

Table~\ref{tab:nonlinear_community} compares the proposed QEE based on the exposure mapping
$G_i=W_i$ and the conventional marginal QTE under the same true value at nominal $90\%$ coverage. The QTE intervals undercover at
all three quantiles. In contrast, QEE plug-in coverage ranges from $91.8\%$ to
$92.4\%$, and the conservative correction raises it to
$94.2\%$ to $94.8\%$. Hence, for this simulation under nonlinear model in community network, the network-aware
QEE intervals attain coverage substantially closer to the nominal level.\end{document}